\documentclass[numsec,webpdf,modern,medium,namedate]{oup-authoring-template}

\onecolumn

\hypersetup{
    colorlinks=true,
    linkcolor=blue,
    citecolor=blue,
    urlcolor=blue
}

\graphicspath{{Fig/}}

\newcounter{assumption}
\newenvironment{assumption}[1][]{%
  \refstepcounter{assumption}%
  \par\medskip
  \noindent\textbf{Assumption~\theassumption.}\ %
  \ifx\relax#1\relax\else\ \textbf{(#1).}\fi
  \itshape
}{%
  \par\medskip
}

\newcounter{proposition}

\newcounter{theorem}
\newenvironment{theorem}[1][]{%
  \refstepcounter{theorem}%
  \par\medskip
  \noindent\textbf{Theorem~\thetheorem.}\ %
  \ifx\relax#1\relax\else\ \textbf{(#1).}\fi
  \itshape
}{%
  \par\medskip
}

\newcounter{lemma}
\newenvironment{lemma}[1][]{%
  \refstepcounter{lemma}%
  \par\medskip
  \noindent\textbf{Lemma~\thelemma.}\ %
  \ifx\relax#1\relax\else\ \textbf{(#1).}\fi
  \itshape
}{%
  \par\medskip
}

\newcounter{definition}
\newenvironment{definition}[1][]{%
  \refstepcounter{definition}%
  \par\medskip
  \noindent\textbf{Definition~\thedefinition.}\ %
  \ifx\relax#1\relax\else\textbf{(#1).}\ \fi
}{%
  \par\medskip
}

\theoremstyle{thmstyletwo}%
\theoremstyle{thmstylethree}%

\usepackage[doublespacing]{setspace}
\usepackage[fontsize=12pt]{fontsize}

\usepackage{float}

\usepackage{bibunits}
\defaultbibliographystyle{abbrvnat}
\defaultbibliography{estimation_bibliography.bib}

\newcommand{\symbfit}[1]{\boldsymbol{#1}}

\newcommand{\norm}[1]{\left\|#1\right\|}

\newcommand{\abs}[1]{\left|#1\right|}

\begin{document}

%\journaltitle{Journals of the Royal Statistical Society}
%\DOI{DOI HERE}
%\copyrightyear{XXXX}
%\pubyear{XXXX}
%\access{Advance Access Publication Date: Day Month Year}
%\appnotes{Original article}

\firstpage{1}

%\subtitle{Subject Section}

\title[Matching random features]{Estimating dynamic models by matching\\ random features}

\author{Michael Wieck-Sosa\ORCID{0009-0006-5369-3674}}
\author[$\ast$]{Cosma Rohilla Shalizi\ORCID{0000-0002-9195-1308}}

\authormark{Wieck-Sosa and Shalizi}

\address{\orgdiv{Department of Statistics \& Data Science}, \orgname{Carnegie Mellon University}, \orgaddress{\street{5000 Forbes Avenue, Pittsburgh}, \state{PA} \postcode{15213}, \country{USA}}}

\corresp[$\ast$]{Address for correspondence: Cosma Rohilla Shalizi, Department of Statistics \& Data Science, Carnegie Mellon University, Pittsburgh, PA 15213, USA. \href{mailto:cshalizi@cmu.edu}{cshalizi@cmu.edu}}

%\received{Date}{0}{Year}
%\revised{Date}{0}{Year}
%\accepted{Date}{0}{Year}

%\editor{Associate Editor: Name}

%\abstract{
%\textbf{Motivation:} .\\
%\textbf{Results:} .\\
%\textbf{Availability:} .\\
%\textbf{Contact:} \href{name@email.com}{name@email.com}\\
%\textbf{Supplementary information:} Supplementary data are available at \textit{Journal Name}
%online.}

\abstract{Scientists increasingly express their ideas as dynamic models of complex processes. It is often much easier to simulate these models than to calculate the probability of their generating a particular outcome, making likelihood-based estimation infeasible. Existing likelihood-free approaches rely either on manually chosen summary statistics or on representations learned by neural networks. The former is error-prone and laborious, while the latter is computationally intensive, leaving many scientists in a difficult position. We show that, for a large class of dynamic models, parameters can be estimated by matching a small number of random features of the observed and simulated data. In particular, we show that models with a $p$-dimensional parameter can be identified from just $2p+1$ generic random Fourier features. We introduce two estimators for stationary and nonstationary processes, respectively, and we establish their consistency under mild regularity conditions. More broadly, our results serve as the foundation for a new class of methods for simulation-based inference.}
% 5-6 keywords in alphabetical order

% \boxedtext{
% \begin{itemize}
% \item Key boxed text here.
% \item Key boxed text here.
% \item Key boxed text here.
% \end{itemize}}

\maketitle

\begin{bibunit}

\section{Introduction}
\label{section:intro}

Across scientific disciplines, dynamic models are used to describe the mechanisms that generate natural phenomena over time. Yet the same mechanistic detail that makes these models scientifically compelling often places conventional likelihood-based statistical inference out of reach. In such cases, simulation offers a way forward. Simulation-based methods assess parameter values $\theta$ by comparing certain characteristics of the observations with those of the synthetic data generated from the model at $\theta$. However, it is often unclear which characteristics should be compared and which should be treated as noise.

Many simulation-based methods rely on carefully crafted summary statistics. These summary statistics are designed to be sensitive to parameter changes and ensure identifiability. Typically, the summary statistics are user-chosen~\citep{wood_likelihood,ABC_Beaumont_2019} or learned~\citep{cranmer_sbi_review,zammitmangion_sbi_review}. However, both approaches face difficulties. Hand-crafting summary statistics requires substantial problem-specific effort, while neural network-based methods are computationally demanding and sensitive to implementation choices.

We propose to replace both hand-crafted and learned summary statistics with \textit{randomly chosen} functions of the data. We establish that models with a $p$-dimensional parameter $\theta$ can be identified with generic choices of just $2p+1$ random Fourier features. Our results imply that simulation-based methods can be made far more automatic through matching random features.
% without meaningfully reducing precision. 
% compared to what? dont want reviewer to say "but actually, you didn't compare to anything yet"

Our approach combines classical statistical ideas with modern random feature methods from machine learning~\citep{Rahimi_Recht_2007_random_features}. However, we use random features in a fundamentally different way from their conventional role in approximating kernel methods. In particular, we repurpose random features for statistical estimation and inference by using them as low-dimensional witnesses (or probes) to distinguish probability distributions.

It is well-known that probability distributions can be uniquely characterized by the expectations they assign a sufficiently rich set of test functions. Indeed, probabilists \textit{define} convergence of distributions as the convergence of expectations of all bounded, continuous functions \citep[ch.\ 4]{Kallenberg-mod-prob}. The space of all test functions is too rich to be practically useful. However, there are smaller, more practical function spaces that are rich enough to be convergence-determining. For distributions on Euclidean spaces, the expectations of Fourier basis functions are convergence-determining \citep[ch.\ 5]{Kallenberg-mod-prob}.

Our first innovation is to realize that a \textit{parametric} model can only produce a very limited variety of distributions. Our second innovation is to realize that it is enough to pick $2p+1$ functions \textit{at random} from a sufficiently-rich, convergence-determining function class. In particular, we adapt the results from \cite{amir_et_al_finite_witness_thm} to show that the expectations of $2p+1$ generic random Fourier features will uniquely characterize the distributions produced by a model with a $p$-dimensional parameter space; see the Supplementary Material for the technical details. This can be done without looking at the data, or even understanding much about the model. Figure~\ref{fig:estimation_procedure_and_iid_gaussian_feature_trajectory_3d} illustrates this in a simple example: independent and identically distributed (iid) observations from a Gaussian distribution with unknown mean.

This, in and of itself, is a result at the level of probability
distributions and expectations, rather than data and estimators; it is
about identification, rather than consistency.  To get a parameter estimation procedure, we need to make additional assumptions
which allow us to reliably approximate expectation
values from observations.  We show this is possible in two
regimes: long-time (or ergodic) asymptotics, and dense-sampling (or
infill) asymptotics.

The first regime applies when
there is a stationary distribution, or at least an
asymptotically stationary distribution. Because a stationary process
is, by definition, one where the distribution of values over time
intervals is invariant under translation, we can meaningfully speak of
\textit{the} expectation value of any function. The expectation values of
the $2p+1$ random features can then be reliably approximated simply by taking averages of those features over time, using classical ergodic and mixing
theory.

The second regime applies to nonstationary processes, where the distribution changes meaningfully over time, so the expectations of random features also vary over time. The
correct framework for such situations is that of ``infill asymptotics'', where we gain more and more observations related to each local structure of the nonstationary process.  In this regime, the local expectations of the $2p+1$ random features can be reliably approximated by averages of those features over rolling windows, which can be specified without knowledge of the data generating process.

In both regimes, we represent the stochastic processes as transformations of sequences of noise inputs. Most dynamic models admit such a representation. This can be checked by recursively expressing each observation in terms of its past noise inputs. Examples include discrete-time dynamical systems and ordinary differential equations (ODEs) with observational noise, discretized stochastic differential equations (SDEs), linear and nonlinear time series models, and state-space models.

To establish the
convergence of averages to the corresponding expectations, we require some mild,
technical regularity assumptions, which say that the impact of any single noise input must eventually decay over time. These assumptions are, in some cases, equivalent to mixing-type assumptions concerning the decay of temporal dependence. Under these assumptions, we show that our estimators are consistent.

\subsection{Related work}\label{subsection:related_work}

We present a selective literature review here. The Supplementary Material contains a comprehensive review of random feature methods, embedding results in nonlinear dynamics, and limit theorems for stochastic processes.

Classical simulation-based approaches include the method of simulated moments~\citep{mcfadden_method_sim_moments} and indirect inference~\citep{indirect_inference_gourieroux_monfort_1993}. For an overview, see~\cite{sim_based_econometrics_gourieroux_monfort}.~\cite{wood_likelihood} extends these ideas by introducing the synthetic likelihood, which is based on the idea of treating user-chosen summary statistics as normally distributed. 
%%% Note: This work inspired an entire literature on synthetic likelihoods.

\cite{BSL_price_2018} introduce the Bayesian synthetic likelihood approach, which assigns a prior distribution on the parameter and uses Markov chain Monte Carlo methods to sample from the resulting posterior distribution of the parameter conditional on the observed summary statistics. Approximate Bayesian computation aims to approximate these posterior distributions by determining whether the simulated summary statistics are sufficiently close to the observed summary statistics based on a user-chosen tolerance level~\citep{ABC_Beaumont_2019}.

Lastly, we point to reviews of neural network-based methods. These methods replace user-chosen summary statistics with representations learned from neural networks trained on simulated data--parameter pairs~\citep{cranmer_sbi_review,zammitmangion_sbi_review}. Overall, simulation-based estimation methods share the idea of substituting unknown quantities with simulated analogues.

\section{Methods}
\label{section:meth}

In this section, we introduce the notation, setting, random features, and estimators. 

\subsection{Notation and setting}\label{subsection:notation_and_setting}

We observe a time series $X_{t}^{\mathrm{obs}}$, $t=1,\ldots,n$, taking values in $\mathbb{R}^d$ for some fixed dimension $d\in\mathbb{N}$. The observed time series is assumed to arise from a known generative model with an unknown $p$-dimensional parameter $\theta_0\in\Theta\subset \mathbb{R}^{p}$, $p\in\mathbb{N}$. Each value of $\theta\in\Theta$ determines a law of a stochastic process. For any value of $\theta$, we can simulate $s\in\mathbb{N}$ realizations of a length $n$ time series $(X_{1:n}^{(r)}(\theta))_{r \in [s]}$. When we do not need to refer to a particular realization, we drop the superscript.

\textsc{Notation}: For real numbers $x,y\in\mathbb{R}$, denote $x \land y = \min(x,y)$ and $x\lor y = \max(x,y)$.
For a vector $x\in\mathbb{R}^d$, denote the $\ell^q$ norm by $\left\|x\right\|_q$ and the Euclidean norm by $\left\|x\right\| = \left\|x\right\|_2$. For a random vector $X$ with distribution depending on $\theta\in\Theta$, $\left\|X\right\|_{\mathcal{L}^q(\theta)}=\left(\mathbb{E}_{\theta}\left\|X\right\|^q\right)^{1/q}$, where $\mathbb{E}_{\theta}(\cdot)$ denotes expectation with respect to this $\theta$-dependent distribution. For two sequences $a_j$, $b_j$, $j\in\mathbb{N}$, we write $a_j\ll b_j$ or $b_j\gg a_j$ to denote $a_j / b_j \xrightarrow[]{} 0$ as $j\xrightarrow[]{}\infty$. For any natural number $j\in\mathbb{N}$, denote $[j]=\{1,2,\ldots,j\}$. For a sequence of random vectors $X_t$, $t=1,\ldots,n$, each taking values in $\mathbb{R}^d$, we denote a subsequence as $X_{t_1:t_2}=(X_{t_1},\ldots,X_{t_2})^{\top}$, which takes values in $\mathbb{R}^{(t_2-t_1+1) \times d}$ for some $t_1,t_2 \in [n]$, $t_1\leq t_2$.
%%% Note: \mathcal{L}^q(\theta) instead of L^q because using L for lag 

\subsection{Random Fourier features}\label{subsection:random_Fourier_features}

The central idea is that we should estimate the $p$-dimensional parameter $\theta_0$ by matching the values of a small number of random features of the simulated and observed data. Consider $k=2p+1$ randomly drawn functions $\varphi_1,\ldots,\varphi_k$ from a distribution over a suitably nice function class $\mathcal{F}$. The $k$ functions should be sampled independently of one another and of the data. Each function $\varphi_i:\mathbb{R}^{(m+1) \times d}\xrightarrow[]{} \mathbb{R}$ takes as input a length $m+1$ subsequence of the $d$-dimensional time series, where $m$ is chosen based on the transition function of the model. For instance, if the model is an order $m^{\ast}$ Markov process, then set $m=m^{\ast}$.

% probably still true as long as they are convergence determining in the sense of the finite whitness theorem. but more clear to just state everything in terms of random Fourier features  
% Here we focus on random Fourier features, though our framework can be used with different random features.
Let $x_1,\ldots,x_{m+1}$ be vectors in $\mathbb{R}^d$, and define $x=(x_1,\ldots,x_{m+1})^{\top}\in\mathbb{R}^{(m+1)\times d}$. Specifically, we consider random Fourier features of the form 
\begin{equation}\label{eqn:random_Fourier_features_order_m}\varphi_{i}(x)=\cos\left(\sum_{j=1}^{m+1} \Omega_{i,j} \cdot x_{j} + \alpha_i\right),\end{equation} where $\Omega_{i,j} \overset{\mathrm{iid}}{\sim} N(0,I_d)$ and $\alpha_i \overset{\mathrm{iid}}{\sim} U(-\pi,\pi)$ for $i=1,\ldots,k$ and $j=1,\ldots,m+1$. Denote the parameters for the $i$-th random Fourier feature, $\varphi_i$, by \begin{equation}\label{eqn:parameters_for_the_random_features}\vartheta_i=(\Omega_{i,1}^{\top},\ldots,\Omega_{i,m+1}^{\top},\alpha_i)^{\top},\end{equation} and let $\vartheta=(\vartheta_1,\ldots,\vartheta_k)$, which is a $\mathbb{R}^{(d(m+1)+1)\times k}$-valued random variable. Similarly, denote all $k$ of these functions by \begin{equation}\label{eqn:all_k_random_Fourier_features_order_m}\varphi=(\varphi_1,\ldots,\varphi_k),\end{equation} so that $\varphi:\mathbb{R}^{(m+1) \times d}\xrightarrow[]{} \mathbb{R}^k$. Define the $i$-th random feature of the observed and $r$-th simulated time series at time $t=m+1,\ldots,n$ as  $$f_{t,i}^{\mathrm{obs}}=\varphi_i\left(X_{t-m:t}^{\mathrm{obs}}\right), \quad f_{t,i}^{(r)}(\theta)=\varphi_i\left( X_{t-m:t}^{(r)}(\theta)\right).$$ Lastly, write all $k$ random features at time $t$ as \begin{equation}\begin{aligned}\label{eqn:random_features_at_time_t} f_{t}^{\mathrm{obs}} & =\left(f_{t,1}^{\mathrm{obs}},\ldots,f_{t,k}^{\mathrm{obs}}\right), \quad f_{t}^{(r)}(\theta) & =\left(f_{t,1}^{(r)}(\theta),\ldots,f_{t,k}^{(r)}(\theta)\right).\end{aligned}\end{equation}

\begin{figure}[H]
\centering
\includegraphics[width=0.5\linewidth]{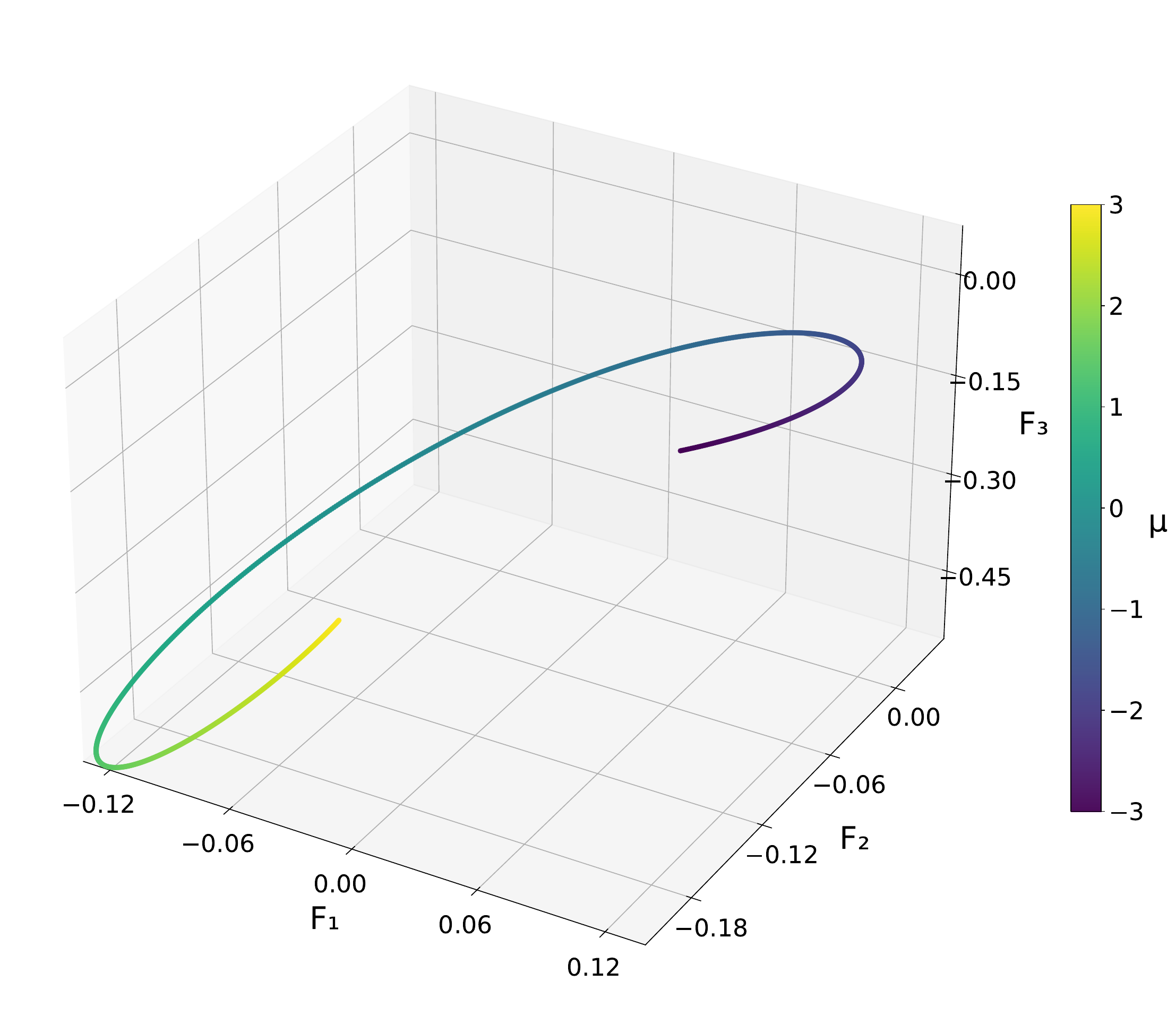}
\caption{\textbf{Example with normal distribution.} Let $X_t^{(r)}(\mu)\overset{\mathrm{iid}}{\sim} N(\mu,1)$ for $t=1,\ldots,n\equiv 1{,}000$ and $r=1,\ldots,s\equiv 10$, and let
$\varphi=(\varphi_1,\varphi_2,\varphi_3)$ be $2p+1=3$ random Fourier features from~\eqref{eqn:random_Fourier_features_order_m} with $m=0$. The plot shows $s^{-1}\sum_{r=1}^s n^{-1}\sum_{t=1}^n \varphi(X_t^{(r)}(\mu)) $ as a function of the parameter $\theta=\mu$, with colour indicating the value of $\mu\in [-3,3]$.}
\par\smallskip{\raggedright
\noindent\textbf{Alt text:} Curve showing averages of $3$ random Fourier features as the mean of the normal distribution varies from $-3$ to $3$. The curve does not intersect itself.\par}
\label{fig:estimation_procedure_and_iid_gaussian_feature_trajectory_3d}
\end{figure}
%%% Note: UK spelling so colour not color

\subsection{Estimators}
\label{subsection:estimators}
\paragraph*{Time-average estimator.}

For processes that are stationary, or at least asymptotically mean stationary, it suffices to consider the time-average of $k=2p+1$ random features. Denote the time-averages of the observed and simulated random features by
\begin{equation}\begin{aligned}\label{eqn:time_average_random_features}  F^{\mathrm{obs}} &= \frac{1}{n-m}\sum_{t=m+1}^n f_{t}^{\mathrm{obs}}, \quad  \bar{F}^{\mathrm{sim}}(\theta)&= \frac{1}{n-m}\sum_{t=m+1}^n \bar{f}_{t}^{\mathrm{sim}}(\theta), \end{aligned}\end{equation}  where $\bar{f}_{t}^{\mathrm{sim}}(\theta)=\frac{1}{s}\sum_{r=1}^s f_{t}^{(r)}(\theta)$. The sample discrepancy function is defined as \begin{equation}\label{eqn:time_average_sample_discrepancy}\hat{Q}_n^{\mathrm{TA}}(\theta) =F^{\mathrm{obs}} - \bar{F}^{\mathrm{sim}}(\theta).\end{equation}

Let the time-average estimator $\hat{\theta}^{\mathrm{TA}}$ be a $\Theta$-valued random variable satisfying \begin{equation}\label{eqn:time_average_estimator}\norm{\hat{Q}_n^{\mathrm{TA}}(\hat{\theta}^{\mathrm{TA}})} \leq \inf_{\theta \in \Theta}\norm{\hat{Q}_n^{\mathrm{TA}}(\theta)} + o_p(1).\end{equation} $\hat{\theta}^{\mathrm{TA}}$ is obtained from an optimization procedure that aims to minimize $\left\|\hat{Q}_n^{\mathrm{TA}}(\cdot)\right\|$.

\paragraph*{Rolling-window estimator.}

For many nonstationary processes, a different approach is required because the time-average of the (time-varying) expectations of the random features may not converge to a limiting mean. Even when it does converge, the limiting mean may not uniquely identify each $\theta$. Hence, we introduce the following \textit{rolling-window} estimator. 

For some window size $w\in\mathbb{N}$, denote the rolling-window averages of the observed and simulated random features at time $t$ by \begin{equation}\begin{aligned}\label{eqn:rolling_window_random_features}  F_{t}^{\mathrm{obs}}&=\frac{1}{N_t}\sum_{j=(t-w)\lor (m+1)}^t f_{j}^{\mathrm{obs}},\quad  \bar{F}_{t}^{\mathrm{sim}}(\theta) &=  \frac{1}{N_t}\sum_{j=(t-w)\lor (m+1)}^t \bar{f}_{j}^{\mathrm{sim}}(\theta),
\end{aligned}\end{equation} where $N_t=N_t(w,m)=(w+1)\land (t-m)$ and  $\bar{f}_{t}^{\mathrm{sim}}(\theta)=\frac{1}{s}\sum_{r=1}^s f_{t}^{(r)}(\theta)$. The sample discrepancy function is defined as \begin{equation}\label{eqn:rolling_window_sample_discrepancy} \begin{aligned} \hat{Q}_n^{\mathrm{RW}}(\theta)&= \frac{1}{n-m}\sum_{t=m+\tau+L}^{n} \left[\left\|F_{t-L}^{\mathrm{obs}}-\bar{F}_{t-L}^{\mathrm{sim}}(\theta)\right\|^2 + K_{t}(\theta)\right],\\   K_{t}(\theta)&=2\left(F_{t-L}^{\mathrm{obs}}-\bar{F}_{t-L}^{\mathrm{sim}}(\theta) \right)^{\top}\left(\left[f_t^{\mathrm{obs}}-\bar{f}_{t}^{\mathrm{sim}}(\theta)\right]-\left[F_{t-L}^{\mathrm{obs}}-\bar{F}_{t-L}^{\mathrm{sim}}(\theta)\right]\right),\end{aligned}\end{equation} where $\tau\in\mathbb{N}$ is an initial time-offset and $L\in\mathbb{N}$ is a lag. In our theory, $w=w_n$, $\tau=\tau_n$, and $L=L_n$ each grow with the sample size $n$ at some rate; see Section~\ref{subsection:result_for_rolling_window_estimator}.

Let the rolling-window estimator $\hat{\theta}^{\mathrm{RW}}$ be a $\Theta$-valued random variable satisfying \begin{align}\label{eqn:rolling_window_estimator}\abs{\hat{Q}_n^{\mathrm{RW}}(\hat{\theta}^{\mathrm{RW}})} \leq \inf_{\theta \in \Theta}\abs{\hat{Q}_n^{\mathrm{RW}}(\theta)} + o_p(1).\end{align} $\hat{\theta}^{\mathrm{RW}}$ is obtained from an optimization procedure that aims to minimize $\left|\hat{Q}_n^{\mathrm{RW}}(\cdot)\right|$.

%%% Note: \paragraph*{Extensions.} Some extensions are possible. First, both estimators can be generalized by replacing the Euclidean norm with a weighted norm, using a weight matrix given by the inverse of a suitable estimator of the long-run covariance matrix. Second, the sample averages used in both estimators can be replaced with different mean estimators. 

\section{Theory}
\label{section:theory}

In this section, we present the assumptions and theoretical results for the estimators.

\subsection{Noise inputs}\label{subsection:noise_inputs}

Let $\left(\varepsilon_i\right)_{i\in\mathbb{Z}}$ be an iid sequence of random variables. Let $\left(\varepsilon_i^{(r)}\right)_{i\in\mathbb{Z}}$, $r=0,1,\ldots,s$, be iid copies of $\left(\varepsilon_i\right)_{i\in\mathbb{Z}}$. Define the sequence of noise inputs up to time $t$ as \begin{equation}\label{eqn:noise_input_sequence_up_to_time_t}\symbfit{\varepsilon}_t=(\varepsilon_t,\varepsilon_{t-1},\ldots),\end{equation} and define $\symbfit{\varepsilon}_t^{(r)}$, $r=0,1,\ldots,s$, analogously. We express $X_{t}^{\mathrm{obs}}$ and $X_{t}^{(r)}(\theta)$ as functions of these noise inputs. In practice, simulations only use finitely many noise inputs, but considering a countably infinite sequence of noise inputs is convenient. For instance, this allows us to easily accommodate sequences of models that are approximations to the true data generating process, where the approximation improves by letting the number of burn-in noise inputs grow.

Let $\left(\varepsilon^{\ast}_i\right)_{i\in\mathbb{Z}}$ be an iid copy of $\left(\varepsilon_i\right)_{i\in\mathbb{Z}}$. For $t\in\mathbb{Z}$, $j\in\mathbb{N}_0$, define \begin{equation}\label{eqn:replaced_j_past_noise_input_sequence_up_to_time_t}\tilde{\symbfit{\varepsilon}}_{t,j}=\left(\varepsilon_{t},\ldots,\varepsilon_{t-j+1},\varepsilon^{\ast}_{t-j}, \varepsilon_{t-j-1},\ldots\right),\end{equation} and define $\tilde{\symbfit{\varepsilon}}_{t,j}^{(r)}$, $r=0,1,\ldots,s$, analogously. $\tilde{\symbfit{\varepsilon}}_{t,j}$ is the sequence of noise inputs $\symbfit{\varepsilon}_{t}$ with the innovation from $j$ steps in the past, $\varepsilon_{t-j}$, replaced by an iid copy $\varepsilon_{t-j}^{\ast}$. The output of the generative model at time $t$ changes if we replace the noise input from $j$ steps ago with an iid copy because the process at time $t$ is a function of the past noise inputs. Our central assumption is that the generative model's current output becomes increasingly insensitive to the $j$-th noise input in the past as $j$ grows.

To formalize this, we use the physical dependence measure of~\cite{wu_funct_dep_meas}. Under certain conditions, these physical dependence measure conditions imply and are implied by $\beta$-mixing and strong mixing conditions; see the Supplementary Material.

\subsection{Results for time-average estimator}\label{subsection:result_for_time_average_estimator}

We establish that the time-average estimator $\hat{\theta}^{\mathrm{TA}}$ is consistent under mild conditions. In a separate paper focused on statistical inference, we will establish asymptotic normality under stronger conditions, which implies a convergence rate of $O_p(n^{-1/2})$.

In this setting, we use long-run asymptotics, so $n$ increasing corresponds to observing the process further into the future. We allow the time series to be generated by approximations to a limiting data generating process. This accommodates many settings with imperfect observations (e.g., observing residuals rather than errors), discretization errors, unknown initial values, and growing burn-in periods.

For each $t\in\mathbb{N}$, $i\in\mathbb{N}$, $r=0,1,\ldots,s$, the measurable mapping $G_t^{(i,r)}$ is from $\mathbb{R}^{\infty}\times \Theta$ to $\mathbb{R}^d$ and $(G_t^{(i,r)}(\symbfit{\varepsilon}_j,\theta))_{j\in\mathbb{Z}}$ is a stationary ergodic process. The observed and simulated time series may be nonstationary because the mappings may change over time $t\in \mathbb{N}$. Throughout, we endow $\mathbb{R}^{\infty}$ with the $\sigma$-algebra generated by all finite projections. The mapping $G_t^{(i,r)}$ can be understood as an approximation to some limiting mapping $G_t^{(r)}$, where the approximation improves as $i$ grows and the dependence on $r$ allows for random initial conditions; see Assumption~\ref{asmpt:limiting_mapping_discrete_time}. The superscript may be ignored if approximations are unnecessary.

The index $n\in\mathbb{N}$ is linked to both the approximation and the sample size. 

\begin{assumption}\label{asmpt:algorithmic_dynamic_model_discrete_time} The observed and simulated time series of length $n\in\mathbb{N}$ are  \begin{align}\label{eqn:algorithmic_dynamic_model_discrete_time} X_{t}^{\mathrm{obs}}&=G_t^{(n,0)}(\symbfit{\varepsilon}_{t}^{(0)},\theta_0), \quad X_{t}^{(r)}(\theta)=G_t^{(n,r)}(\symbfit{\varepsilon}_{t}^{(r)},\theta),\quad t=1,\ldots,n, \end{align} for all $r=1,\ldots,s$ and $\theta\in\Theta$.
\end{assumption}

We can often take the observed time series to be generated by the limiting mapping, i.e., $G_t^{(n,0)}\equiv G_t^{(0)}$ for all $n\in\mathbb{N}$. The following assumption imposes the convergence to limiting mappings as the approximation level grows, with no required rate of convergence. We write $\symbfit{\varepsilon}_0$ because it has the same distribution as each $\symbfit{\varepsilon}_t$.

\begin{assumption}\label{asmpt:limiting_mapping_discrete_time} For each $t\in\mathbb{N}$ and $r=0,1,\ldots,s$, there exists a measurable mapping $G_t^{(r)}:\mathbb{R}^{\infty}\times \Theta\xrightarrow[]{}\mathbb{R}^d$ such that $(G_t^{(r)}(\symbfit{\varepsilon}_j,\theta))_{j\in\mathbb{Z}}$ is a stationary ergodic process, and for some $q>2$, as $n\xrightarrow[]{}\infty$, we have  \begin{align*}\underset{\theta\in\Theta}{\sup} \left( \frac{1}{n}\sum_{t=1}^n \left\|G_t^{(n,r)}(\symbfit{\varepsilon}_0,\theta)-G_t^{(r)}(\symbfit{\varepsilon}_0,\theta)\right\|_{\mathcal{L}^q(\theta)}\right)\xrightarrow[]{} 0. \end{align*}\end{assumption}

Under Assumption~\ref{asmpt:algorithmic_dynamic_model_discrete_time}, the $k= 2p+1$ time-averages from~\eqref{eqn:time_average_random_features} can be written as
\begin{equation}\label{eqn:time_average_random_features_G_general}\begin{aligned} F^{\mathrm{obs}} &= \frac{1}{n-m}\sum_{t=m+1}^n \varphi\left(\left[G_{t-m}^{(n,0)}(\symbfit{\varepsilon}_{t-m}^{(0)},\theta_0),\ldots,G_t^{(n,0)}(\symbfit{\varepsilon}_t^{(0)},\theta_0)\right]^{\top}\right), \\  \bar{F}^{\mathrm{sim}}(\theta) &= \frac{1}{n-m}\sum_{t=m+1}^n \left[\frac{1}{s}\sum_{r=1}^s \varphi\left(\left[G_{t-m}^{(n,r)}(\symbfit{\varepsilon}_{t-m}^{(r)},\theta),\ldots,G_t^{(n,r)}(\symbfit{\varepsilon}_t^{(r)},\theta)\right]^{\top}\right)\right].\end{aligned}\end{equation}

Denote by $P_{\theta,t}^{(r)}$ the distribution of $\left[G_{t-m}^{(r)}(\symbfit{\varepsilon}_{-m},\theta),\ldots,G_t^{(r)}(\symbfit{\varepsilon}_{0},\theta)\right]^{\top}$, where $m$ is the number of lags in the random features $\varphi$ from~\eqref{eqn:all_k_random_Fourier_features_order_m}. We impose asymptotic mean stationarity by requiring the time-average of the distributions $(P_{\theta,t}^{(r)})_{t\in\mathbb{N}}$ to converge weakly to a limiting distribution $P_{\theta}$ which depends neither on $t$ nor $r$, \textit{uniformly} over $\theta\in\Theta$. This allows us to consider discrete-time dynamical systems with unknown initial values. For the following assumption, denote $\bar{P}_{\theta,n}^{(r)}=\frac{1}{n-m}\sum_{t=m+1}^n P_{\theta,t}^{(r)}$.

\begin{assumption}\label{asmpt:weak_convergence_of_time_avg_of_distributions_discrete_time} 
    For all $\theta \in \Theta$, there exists a limiting distribution $P_{\theta}$, such that, $$\underset{\theta\in\Theta}{\sup} \ \underset{0\leq r \leq s}{\max} \ \left|
\begin{aligned}
 \int_{\mathbb{R}^{(m+1)\times d}}\phi(x) \, d\bar{P}_{\theta,n}^{(r)}(x) - \int_{\mathbb{R}^{(m+1)\times d}}\phi(x) \, dP_{\theta}(x)\end{aligned} \right|\xrightarrow[]{} 0, \quad n\xrightarrow[]{}\infty,$$ for all continuous, bounded functions $\phi:\mathbb{R}^{(m+1)\times d}\xrightarrow[]{}\mathbb{R}$.
    
\end{assumption}

Next, we introduce mappings $\Phi:\Theta\xrightarrow[]{}\mathbb{R}^k$, defined as \begin{equation}\label{eqn:limiting_expectation_random_features_discrete_time} \Phi(\theta)=\int_{\mathbb{R}^{(m+1)\times d}}\varphi(x) \, dP_{\theta}(x), \quad \theta \in \Theta.\end{equation} Under Assumptions~\ref{asmpt:algorithmic_dynamic_model_discrete_time}-\ref{asmpt:weak_convergence_of_time_avg_of_distributions_discrete_time}, the limiting expectation values of the $k=2p+1$ random features of the time series coincide with $\Phi$.

The definition of a semialgebraic set is required for the next assumption. We simplify the definition from \cite{bochnak_coste_roy_real_algebraic_geometry} to Euclidean spaces.

\begin{definition}[Definition 2.1.4 in \cite{bochnak_coste_roy_real_algebraic_geometry}]\label{def:semialgebraic_set}
Let $p,a,b_1,\ldots,b_a \in \mathbb{N}$. A \textit{semialgebraic subset} of $\mathbb{R}^p$
is a subset of the form
\[
\bigcup_{i=1}^{a}\bigcap_{j=1}^{b_i}
\left\{x \in \mathbb{R}^p
\mid f_{i,j}(x) \mathbin{*_{i,j}} 0\right\},
\]
where $f_{i,j} \in \mathbb{R}[x_1,\ldots,x_p]$, $*_{i,j}$ is either
$<$ or $=$, for $i=1,\ldots,a$ and $j=1,\ldots,b_i$, and $\mathbb{R}[x_1,\ldots,x_p]$ is the ring of polynomials in $x_1,\ldots,x_p$ with coefficients in $\mathbb{R}$.
\end{definition}

%We introduce some notation for the next assumption. Let $\mathcal{P}\left(\mathbb{R}^{(m+1)\times d}\right)$ denote the space of probability measures on $\mathbb{R}^{(m+1)\times d}$. 

Let $\psi_{\theta}(\omega)$ denote the characteristic function of $P_{\theta}$ evaluated at $\omega\in\mathbb{R}^{(m+1)\times d}$.

\begin{assumption}\label{asmpt:statistical_model_discrete_time}
The following conditions hold:
{\renewcommand{\labelenumi}{(\alph{enumi})}\begin{enumerate} 
    \item $\Theta$ is a compact, semialgebraic subset of $\mathbb{R}^{p}$ with non-empty interior. 
    \item The map $\theta\mapsto P_{\theta}$ is injective.  
    \item The real and imaginary parts of the map $(\theta,\omega)\mapsto \psi_{\theta}(\omega)$ each admit a jointly real analytic extension to $O \times \mathbb{R}^{(m+1)\times d}$, where $O$ is open and $\Theta \subset O\subset \mathbb{R}^p$. % i.e., $O$ $is an open neighborhood of $\Theta$ 
\end{enumerate}}    
\end{assumption}

The first condition is used to establish identifiability and consistency. As a simple example, for $a_j,b_j\in\mathbb{R}$ with $a_j < b_j$, $j=1,\ldots,p$, we may take $\Theta=\prod_{j=1}^p [a_j,b_j]$. The second condition is standard for identifiability. The third condition is used to establish the regularity of each map $(\theta,\vartheta_i)\mapsto \int_{\mathbb{R}^{(m+1)\times d}}\varphi_i(x) \, dP_{\theta}(x)$, $i=1,\ldots,k$.

\begin{theorem}\label{thm:injective_discrete_time}
Suppose Assumptions~\ref{asmpt:algorithmic_dynamic_model_discrete_time}-\ref{asmpt:statistical_model_discrete_time} hold. Then, with probability one over the draw of random parameters $(\vartheta_1,\ldots,\vartheta_k)$ for the $k=2p+1$ random Fourier features $(\varphi_1,\ldots,\varphi_k)$, we have that the mapping $\theta\mapsto\Phi(\theta)$ is one-to-one, where $\Phi(\theta)$, $\theta\in\Theta$, are the limiting expectations of the random Fourier features.
\end{theorem}

To simplify the presentation of the results below, we make the weak assumption that the random parameters we have drawn for the random Fourier features yield a function from~\eqref{eqn:limiting_expectation_random_features_discrete_time} that is one-to-one. We say this assumption is weak because by Theorem~\ref{thm:injective_discrete_time} it holds with probability one over the draw of random parameters. 
\begin{assumption}\label{asmpt:injective_discrete_time}   The random parameters $(\vartheta_1,\ldots,\vartheta_k)$ we have drawn for the $k=2p+1$ random Fourier features $(\varphi_1,\ldots,\varphi_k)$ yield a function $\theta\mapsto\Phi(\theta)$ that is one-to-one. 
\end{assumption}

If one desires to avoid this assumption, one may interpret the following results as holding with probability one over the draw of random parameters $(\vartheta_1,\ldots,\vartheta_k)$.

The next stochastic equicontinuity-type condition is used to establish consistency.

\begin{assumption}\label{asmpt:stochastic_equicontinuity_discrete_time} For all $\epsilon >0$, there exists an $\eta_G =\eta_G(\epsilon)>0$ such that 
$$\underset{i \in\mathbb{N}}{\sup} \ \underset{t \in\mathbb{N}}{\sup} \ \underset{0\leq r\leq s}{\max} \ \mathbb{E}\left(
\underset{ \left\|\theta-\theta'\right\|\leq \epsilon}{\underset{\theta,\theta'\in\Theta}{\sup}} \left\|
G_{t}^{(i,r)}(\symbfit{\varepsilon}_{0},\theta) - 
G_{t}^{(i,r)}(\symbfit{\varepsilon}_{0},\theta')\right\|\right)< \eta_G,$$ with $\eta_G \xrightarrow[]{} 0$ as $\epsilon \xrightarrow[]{} 0$, where $\mathbb{E}(\cdot)$ is taken w.r.t.\ the law of the noise inputs.
\end{assumption}

The next assumption controls the temporal dependence uniformly over $\Theta$, so that we can estimate the expectations of the random features based on their time-averages. Specifically, we impose a uniform moment bound and a uniform polynomial decay of the physical dependence measure of~\cite{wu_funct_dep_meas}.  Recall $\tilde{\symbfit{\varepsilon}}_{0,j}$ from~\eqref{eqn:replaced_j_past_noise_input_sequence_up_to_time_t}.
\begin{assumption}\label{asmpt:temporal_dependence_discrete_time}
    There exist $\Psi>0$, $\beta>2$ such that, for some $q > 2$ and all $j\in\mathbb{N}_0$, \begin{align*}  \underset{i\in\mathbb{N}}{\sup}  \ \underset{\theta\in\Theta}{\sup} \ \underset{t\in\mathbb{N}}{\sup}   \  \underset{0\leq r \leq s}{\max} \   \left\|G_t^{(i,r)}(\symbfit{\varepsilon}_{0},\theta) - G_t^{(i,r)}(\tilde{\symbfit{\varepsilon}}_{0,j},\theta)\right\|_{\mathcal{L}^q(\theta)}   & \leq   \Psi (j \lor 1)^{-\beta},\\  \underset{i\in\mathbb{N}}{\sup} \ \underset{\theta\in\Theta}{\sup} \ \underset{t\in\mathbb{N}}{\sup}   \  \underset{0\leq r \leq s}{\max} \  \left\|G_t^{(i,r)}(\symbfit{\varepsilon}_{0},\theta)\right\|_{\mathcal{L}^q(\theta)}  &\leq  \Psi. \end{align*} 
\end{assumption}

We now establish the consistency of the time-average estimator $\hat{\theta}^{\mathrm{TA}}$.

\begin{theorem}\label{thm:consistency_time_average_estimator}
    If Assumptions~\ref{asmpt:algorithmic_dynamic_model_discrete_time}-\ref{asmpt:temporal_dependence_discrete_time} hold, then $\hat{\theta}^{\mathrm{TA}}\xrightarrow[]{p}\theta_0$.
\end{theorem}

\subsection{Results for rolling-window estimator}\label{subsection:result_for_rolling_window_estimator}

We show that the rolling-window estimator $\hat{\theta}^{\mathrm{RW}}$ is consistent under mild conditions. In a separate manuscript focused on statistical inference, we will prove asymptotic normality under stronger conditions, which implies an $O_p(n^{-1/2})$ rate of convergence.

In this section, we rescale time to the unit interval and use infill asymptotics as in the literature on locally stationary processes by~\cite{dahlhaus1997}. As $n$ grows, an increasing number of observations related to each local structure of the process become available. The rolling-window estimator can be used to fit models with very general forms of nonstationarity, such as ordinary differential equations (ODEs), state-space models with covariates and control inputs, and structural time series models with trends, seasonality, change-points, and ``rough'' nonstationarity.

% Note: I haven't read too much of this literature, so I don't want to overstate things. 
% Prior work has used infill asymptotics for estimating the parameters of ordinary differential equations (ODEs) with iid observational noise~\cite{sieve_ODE_2010, ode_estimation}. However, to the best of our knowledge, our work is the first to provide theoretical guarantees for settings where the observational noise for the ODE has temporal, cross-sectional, and state dependence. Beyond ODEs, our rolling-window estimator can be used to fit the parameters of models with very general forms of nonstationarity, such as state-space models and structural time series models with complicated trends, seasonality, change-points, and ``rough'' nonstationarity.

The time series may be generated by approximations to a limiting data generating process. For example, this distinction allows the observed time series to be based on observed residuals rather than errors, or to consist of imperfect observations of an exact solution of an ODE. Similarly, the simulated time series can be based on an ODE solver with observational noise processes with growing burn-in periods.

We represent the generative model using measurable mappings. For each $u\in [0,1]$, $i\in\mathbb{N}$, $r=0,1,\ldots,s$, the mapping $G_u^{(i,r)}$ is from $\mathbb{R}^{\infty}\times \Theta$ to $\mathbb{R}^d$ and $(G_u^{(i,r)}(\symbfit{\varepsilon}_j,\theta))_{j\in\mathbb{Z}}$ is a stationary ergodic process. The time series may be nonstationary because the mappings may change over rescaled time $u\in [0,1]$. The mapping $G_u^{(i,r)}$ is an approximation to a limiting mapping. When an approximation is not needed, the superscript can be ignored.

The index $n\in\mathbb{N}$ is linked to the sample size and approximation.

\begin{assumption}\label{asmpt:algorithmic_dynamic_model_continuous_time} The observed and simulated time series of length $n\in\mathbb{N}$ are \begin{align}\label{eqn:algorithmic_dynamic_model_continuous_time} X_{t}^{\mathrm{obs}}&=G_{t/n}^{(n,0)}(\symbfit{\varepsilon}_{t}^{(0)},\theta_0), \quad X_{t}^{(r)}(\theta)=G_{t/n}^{(n,r)}(\symbfit{\varepsilon}_{t}^{(r)},\theta),\quad t=1,\ldots,n,\end{align} for all $r=1,\ldots,s$ and $\theta\in\Theta$.
\end{assumption} 
% Note: we assume the observed time series comes from the limiting mapping (generated by nature) and that the simulated time series comes from the approximation. We can transfer the assumptions, which we make on the approximations, to the observed process by asymptotically replacing the observed time series with the limiting mapping with the approximation mapping.

In many situations, we can take the observed time series to be generated by the limiting mapping, i.e., $G_{u}^{(n,0)}\equiv G_{u}$ for all $u\in [0,1]$. The next assumption imposes convergence to the limiting mapping, with no required convergence rate. Note that $\symbfit{\varepsilon}_0$ has the same distribution as $\symbfit{\varepsilon}_t$ for all $t\in\mathbb{Z}$. 

\begin{assumption}\label{asmpt:limiting_mapping_continuous_time} For each $u\in [0,1]$, there exists a measurable mapping $G_u:\mathbb{R}^{\infty}\times \Theta\xrightarrow[]{}\mathbb{R}^d$ such that $(G_u(\symbfit{\varepsilon}_j,\theta))_{j\in\mathbb{Z}}$ is a stationary ergodic process, and for some $q>2$,   \begin{align*}\underset{\theta\in\Theta}{\sup} \ \underset{u\in [0,1]}{\sup} \  \underset{0\leq r \leq s}{\max} \ \left\|G_u^{(i,r)}(\symbfit{\varepsilon}_0,\theta)-G_u(\symbfit{\varepsilon}_0,\theta)\right\|_{\mathcal{L}^q(\theta)}\xrightarrow[]{} 0, \quad i\xrightarrow[]{} \infty. \end{align*}
\end{assumption}

Under Assumption~\ref{asmpt:algorithmic_dynamic_model_continuous_time}, the observed and simulated rolling-window random features at time $t$ can be written as \begin{align}\nonumber F_{t}^{\mathrm{obs}}&=\frac{1}{N_t}\sum_{j=(t-w)\lor (m+1)}^t \varphi\left(\left[G_{(j-m)/n}^{(n,0)}(\symbfit{\varepsilon}_{j-m}^{(0)},\theta_0),\ldots,G_{j/n}^{(n,0)}(\symbfit{\varepsilon}_j^{(0)},\theta_0)\right]^{\top}\right),\\ \nonumber \bar{F}_{t}^{\mathrm{sim}}(\theta) &=  \frac{1}{N_t}\sum_{j=(t-w)\lor (m+1)}^t \left[\frac{1}{s}\sum_{r=1}^s\varphi\left(\left[G_{(j-m)/n}^{(n,r)}(\symbfit{\varepsilon}_{j-m}^{(r)},\theta),\ldots,G_{j/n}^{(n,r)}(\symbfit{\varepsilon}_j^{(r)},\theta)\right]^{\top}\right)\right].
\end{align}

For each $u\in [0,1]$, denote by $P_{\theta,u}$ the distribution of $\left[G_u(\symbfit{\varepsilon}_{-m},\theta),\ldots,G_u(\symbfit{\varepsilon}_{0},\theta)\right]^{\top}$, where $m$ is the number of lags used in the random features. We introduce mappings $\Phi_u:\Theta\xrightarrow[]{}\mathbb{R}^k$, $u\in [0,1]$, defined as \begin{equation}\label{eqn:original_random_feature_expectation_cont_time} \Phi_u(\theta)= \int_{\mathbb{R}^{(m+1)\times d}}\varphi(x) \, dP_{\theta,u}(x).\end{equation} Under Assumptions~\ref{asmpt:algorithmic_dynamic_model_continuous_time} and \ref{asmpt:limiting_mapping_continuous_time}, the expectation values of the $k= 2p+1$ random features at rescaled time $u\in [0,1]$ coincide with $\Phi_u$.

Unfortunately, in general, we cannot injectively map $\Theta$ into $\mathbb{R}^k$ using $\Phi_u$ at a particular $u\in [0,1]$. To see this, consider a family of stochastic processes defined by $X_t(\theta)=\cos(\theta 2\pi t/n)+\varepsilon_t$, $t=1,\ldots,n$, where each $\varepsilon_t \overset{\mathrm{iid}}{\sim} N(0,1)$ and $\theta \in [0,10]$. For $u\in [0,1]$ and $m=0$, denote by $P_{\theta,u}$ the distribution of $\cos(\theta 2\pi u)+\varepsilon_t$, i.e., $P_{\theta,u}=N(\cos(\theta 2\pi u),1)$. Clearly, for two different parameters $\theta_1\neq\theta_2$, it is possible for $P_{\theta_1,u}=P_{\theta_2,u}$. Of course, this also occurs in more complicated dynamic models. To deal with these situations, we make the following reparametrization assumption.

The definition of a semialgebraic mapping is required for the next assumption. We simplify the definition from \cite{bochnak_coste_roy_real_algebraic_geometry} to Euclidean spaces.

\begin{definition}[Definition 2.2.5 in \cite{bochnak_coste_roy_real_algebraic_geometry}]\label{def:semialgebraic_map}
Let $p_1,p_2 \in \mathbb{N}$, and let
$A \subset \mathbb{R}^{p_1}$ and $B \subset \mathbb{R}^{p_2}$ be two
semialgebraic sets (Definition~\ref{def:semialgebraic_set}). A mapping $g \colon A \to B$ is semialgebraic if
its graph $G(g)=\{(x,g(x)):x\in A\}$ is semialgebraic in $\mathbb{R}^{p_1+p_2}$.
\end{definition}

\begin{assumption}\label{asmpt:reparametrization_continuous_time} For each $u\in [0,1]$, there exists a Lipschitz, semialgebraic mapping $g_u:\Theta\xrightarrow[]{}\mathbb{R}^{\tilde{p}}$ defined by $g_u(\theta)=\tilde{\theta}_u$, where $\tilde{\Theta}_u=\{g_u(\theta):\theta\in\Theta\}\subset\mathbb{R}^{\tilde{p}}$, $\tilde{p}\in \mathbb{N}$, and there exists a measurable mapping $\tilde{G}_u:\mathbb{R}^{\infty}\times \tilde{\Theta}_u\xrightarrow[]{}\mathbb{R}^d$ such that $$\tilde{G}_u(\symbfit{\varepsilon}_0,\tilde{\theta}_u)=G_u(\symbfit{\varepsilon}_0,\theta),$$ where $(\tilde{G}_u(\symbfit{\varepsilon}_j,\tilde{\theta}_u))_{j\in\mathbb{Z}}$ is a stationary ergodic process for each fixed $u\in [0,1]$. Furthermore, assume that the map $(u,\theta)\mapsto g_u(\theta)$ is jointly Borel measurable.
\end{assumption}

The assumption that $g_u$ is semialgebraic (Definition~\ref{def:semialgebraic_map}) can be weakened to only requiring $g_u$ to be $\sigma$-subanalytic; see the Supplementary Material for the details.

Going forward, let $\tilde{P}_{\tilde{\theta}_u,u}$ denote the distribution of $\left[\tilde{G}_u(\symbfit{\varepsilon}_{-m},\tilde{\theta}_u),\ldots,\tilde{G}_u(\symbfit{\varepsilon}_0,\tilde{\theta}_u)\right]^{\top}$. Now, we introduce the mappings $\tilde{\Phi}_u:\tilde{\Theta}_u\xrightarrow[]{}\mathbb{R}^k$, defined as \begin{equation}\label{eqn:reparam_random_feature_expectation}\tilde{\Phi}_u(\tilde{\theta}_u) = \int_{\mathbb{R}^{(m+1)\times d}}\varphi(x) \, d\tilde{P}_{\tilde{\theta}_u,u}(x).\end{equation} Under Assumptions~\ref{asmpt:algorithmic_dynamic_model_continuous_time}-\ref{asmpt:reparametrization_continuous_time}, the limiting local expectation values at rescaled time $u\in [0,1]$ of the $k=2p+1$ random features of the time series are given by $\tilde{\Phi}_u(\tilde{\theta}_u)$, $\tilde{\theta}_u\in\tilde{\Theta}_u$. We would like to say that, for generic choices of the $k=2p+1$ random Fourier features, we have that $\tilde{\Phi}_u$ is one-to-one.

% Let $\mathcal{P}\left(\mathbb{R}^{(m+1)\times d}\right)$ denote the space of probability measures on $\mathbb{R}^{(m+1)\times d}$

Let $\tilde{\psi}_{\tilde{\theta}_u,u}(\omega)$ be the characteristic function of $\tilde{P}_{\tilde{\theta}_u,u}$ evaluated at $\omega\in\mathbb{R}^{(m+1)\times d}$.

\begin{assumption}\label{asmpt:statistical_model_continuous_time}
The following conditions hold for all rescaled times $u\in [0,1]$: 
{\renewcommand{\labelenumi}{(\alph{enumi})}\begin{enumerate} %%% Note: $\tilde{\Theta}_u\subset \mathbb{R}^{\tilde{p}}$ and $\Theta\subset \mathbb{R}^{p}$
    \item $\tilde{\Theta}_u$ is compact, and $\Theta$ is compact and semialgebraic with non-empty interior.
    \item The map $\tilde{\theta}_u\mapsto \tilde{P}_{\tilde{\theta}_u,u}$ is injective.
    \item The real and imaginary parts of the map $(\tilde{\theta}_u,\omega)\mapsto\tilde{\psi}_{\tilde{\theta}_u,u}(\omega)$ each admit a jointly real analytic extension to $\tilde{O}_u \times \mathbb{R}^{(m+1)\times d}$, where $\tilde{O}_u$ is open and $\tilde{\Theta}_u\subset \tilde{O}_u \subset\mathbb{R}^{\tilde{p}}$. % i.e., $\tilde{O}_u$ $is an open neighborhood of $\tilde{\Theta}_u$ 
\end{enumerate}}
\end{assumption}

The first condition is used to establish identifiability and consistency. For example, for some $a_j,b_j\in\mathbb{R}$ with $a_j < b_j$, $j \in [p]$, we can take $\Theta=\prod_{j=1}^p [a_j,b_j]$. The second condition is commonly used for identifiability. The third condition is used to control each $(\tilde{\theta}_u,\vartheta_i)\mapsto \int_{\mathbb{R}^{(m+1)\times d}}\varphi_i(x) \, d\tilde{P}_{\tilde{\theta}_u,u}(x)$, $i=1,\ldots,k$, $u\in [0,1]$.

We require the following regularity condition to establish identifiability.

\begin{assumption}\label{asmpt:measurability} The following conditions hold:
{\renewcommand{\labelenumi}{(\alph{enumi})}\begin{enumerate}
    \item The set-valued map $u\mapsto \tilde{\Theta}_u$ is weakly Borel measurable. 
    \item There exists a jointly Borel measurable function $\bar{\psi}:[0,1] \times \mathbb{R}^{\tilde{p}} \times \mathbb{R}^{(m+1)\times d}\xrightarrow[]{}\mathbb{C}$ such that $\bar{\psi}(u,\tilde{\theta}_u,\omega)= \tilde{\psi}_{\tilde{\theta}_u,u}(\omega)$ for all $u\in [0,1]$, $\tilde{\theta}_u\in\tilde{\Theta}_u$, and $\omega\in\mathbb{R}^{(m+1)\times d}$, and is continuous in $\tilde{\theta}_u$ for each fixed $(u,\omega)$. 
\end{enumerate}}   
\end{assumption}

The next result is essential for identifiability. We emphasize that we only require $k=2p+1$ features, as opposed to $2\tilde{p}+1$, where $\tilde{p}$ is from Assumption~\ref{asmpt:reparametrization_continuous_time}.

\begin{theorem}\label{thm:injective_continuous_time} If Assumptions~\ref{asmpt:algorithmic_dynamic_model_continuous_time}-\ref{asmpt:measurability} hold, then with probability one over the draw of parameters $(\vartheta_1,\ldots,\vartheta_k)$ for the $k=2p+1$ random Fourier features $(\varphi_1,\ldots,\varphi_k)$, the mappings $\tilde{\theta}_u \mapsto \tilde{\Phi}_u(\tilde{\theta}_u)$, $u\in [0,1]$, are one-to-one for almost every $u\in [0,1]$, where $\tilde{\Phi}_u(\tilde{\theta}_u)$, $\tilde{\theta}_u\in\tilde{\Theta}_u$, are the limiting local expectations of the random Fourier features.
\end{theorem}

To simplify the presentation of the results below, we assume that the parameters we have drawn for the random Fourier features yield functions from~\eqref{eqn:reparam_random_feature_expectation} that are one-to-one at almost every rescaled time $u$. This is a weak assumption because, by Theorem~\ref{thm:injective_continuous_time}, it holds with probability one over the draw of random parameters. 
\begin{assumption}\label{asmpt:injective_continuous_time}  The random parameters $(\vartheta_1,\ldots,\vartheta_k)$ we have drawn for the $k=2p+1$ random Fourier features $(\varphi_1,\ldots,\varphi_k)$ yield mappings $\tilde{\theta}_u\mapsto\tilde{\Phi}_u(\tilde{\theta}_u)$, $u\in [0,1]$, that are one-to-one for almost every $u\in [0,1]$.
\end{assumption}

If one wishes to avoid this assumption, one may interpret the results below as holding with probability one over the draw of random parameters $(\vartheta_1,\ldots,\vartheta_k)$.

We assume that the paths of distributions $(P_{\theta_1,u})_{u\in [0,1]}$, $(P_{\theta_2,u})_{u\in [0,1]}$ are sufficiently distinguishable, in some sense, for any two distinct parameters $\theta_1,\theta_2\in\Theta$. 
\begin{assumption}\label{asmpt:distinguish_paths_of_distributions_continuous_time} For all $\theta_1,\theta_2\in\Theta$, if $\theta_1\neq\theta_2$, then there exists a subset $\mathcal{U}_{\theta_1,\theta_2}\subseteq [0,1]$ with positive Lebesgue measure such that $P_{\theta_1,u}\neq P_{\theta_2,u}$ for all $u\in \mathcal{U}_{\theta_1,\theta_2}$.
\end{assumption}

 This is a mild condition, as we allow $P_{\theta_1,u}=P_{\theta_2,u}$ for all $u \in [0,1] \setminus \mathcal{U}_{\theta_1,\theta_2}$.

We have stated the core assumptions, so we may now sketch the main ideas of the proof of Theorem~\ref{thm:consistency_rolling_window_estimator}, which shows the consistency of $\hat{\theta}^{\mathrm{RW}}$. By Assumption~\ref{asmpt:distinguish_paths_of_distributions_continuous_time}, for any two different parameters $\theta_1,\theta_2\in\Theta$, we have that $P_{\theta_1,u}\neq P_{\theta_2,u}$ for all $u\in\mathcal{U}_{\theta_1,\theta_2}$. Since $P_{\theta_1,u}\neq P_{\theta_2,u}$ for all $u\in\mathcal{U}_{\theta_1,\theta_2}$, we also have that $\tilde{P}_{g_u(\theta_1),u}\neq \tilde{P}_{g_u(\theta_2),u}$ for all $u\in\mathcal{U}_{\theta_1,\theta_2}$ because $\tilde{P}_{g_u(\theta),u}=P_{\theta,u}$ is the distribution of $$[\tilde{G}_u(\symbfit{\varepsilon}_{-m},g_u(\theta)),\ldots,\tilde{G}_u(\symbfit{\varepsilon}_0,g_u(\theta))]^{\top}=[G_u(\symbfit{\varepsilon}_{-m},\theta),\ldots,G_u(\symbfit{\varepsilon}_0,\theta)]^{\top}.$$ Therefore, for all $u\in\mathcal{U}_{\theta_1,\theta_2}$, we have that $g_u(\theta_1)\neq g_u(\theta_2)$ because $\tilde{\theta}_u \mapsto \tilde{P}_{\tilde{\theta}_u,u}$ is injective for all $u\in [0,1]$ by Assumption~\ref{asmpt:statistical_model_continuous_time}. We also have that
$\tilde{\Phi}_u(g_u(\theta_1))\neq \tilde{\Phi}_u(g_u(\theta_2))$ for almost every $u\in\mathcal{U}_{\theta_1,\theta_2}$, because by Assumption~\ref{asmpt:injective_continuous_time}, $\tilde{\theta}_u\mapsto \tilde{\Phi}_u(\tilde{\theta}_u)$ is one-to-one for almost every $u\in[0,1]$. This also implies that $\Phi_u(\theta_1)\neq \Phi_u(\theta_2)$ for almost every $u\in\mathcal{U}_{\theta_1,\theta_2}$ by the definitions of $\Phi_u$ and $\tilde{\Phi}_u$ from~\eqref{eqn:original_random_feature_expectation_cont_time} and~\eqref{eqn:reparam_random_feature_expectation}, respectively, because $\tilde{P}_{g_u(\theta),u}=P_{\theta,u}$ as discussed above. Thus, if $\Phi_u$ were known, we could distinguish $\theta_1$ and $\theta_2$ by comparing the expectation values of the random features $\Phi_u(\theta_1)$ and $\Phi_u(\theta_2)$ at almost every $u\in \mathcal{U}_{\theta_1,\theta_2}$.

We use a stochastic equicontinuity-type condition to establish consistency.

\begin{assumption}\label{asmpt:stochastic_equicontinuity_continuous_time} For all $\epsilon >0$, there exists an $\eta_G =\eta_G(\epsilon) >0$ such that $$\underset{i \in\mathbb{N}}{\sup} \ \underset{u \in [0,1]}{\sup} \ \underset{0\leq r\leq s}{\max} \ \mathbb{E}\left(
\underset{ \left\|\theta-\theta'\right\|\leq \epsilon}{\underset{\theta,\theta'\in\Theta}{\sup}} \left\|
G_{u}^{(i,r)}(\symbfit{\varepsilon}_{0},\theta) - 
G_{u}^{(i,r)}(\symbfit{\varepsilon}_{0},\theta')\right\|^2\right)< \eta_G,$$ with $\eta_G \xrightarrow[]{} 0$ as $\epsilon \xrightarrow[]{} 0$, where $\mathbb{E}(\cdot)$ is taken w.r.t.\ the law of the noise inputs.
\end{assumption}

The next assumptions control the temporal dependence and nonstationarity uniformly over $\Theta$, so that we can estimate the expectations of the random features $\Phi_u(\theta)$ at all $u\in [0,1]$ and $\theta \in \Theta$. First, we impose a uniform moment bound and a uniform exponential decay of the physical dependence measure of~\cite{wu_funct_dep_meas}.

\begin{assumption}\label{asmpt:temporal_dependence_continuous_time} Recall $\tilde{\symbfit{\varepsilon}}_{0,j}$ from~\eqref{eqn:replaced_j_past_noise_input_sequence_up_to_time_t}. There exist constants $\rho\in (0,1)$, $\Lambda >0$, such that, for some $q > 2$ and all $j\in\mathbb{N}_0$, we have  \begin{align*}  \underset{i\in\mathbb{N}}{\sup} \ \underset{\theta\in\Theta}{\sup} \ \underset{u\in [0,1]}{\sup}  \  \underset{0 \leq r \leq s}{\max} \  \left\|G_u^{(i,r)}(\symbfit{\varepsilon}_{0},\theta) - G_u^{(i,r)}(\tilde{\symbfit{\varepsilon}}_{0,j},\theta)\right\|_{\mathcal{L}^q(\theta)}   &\leq \Lambda \rho^j, \\  \underset{i\in\mathbb{N}}{\sup} \ \underset{\theta\in\Theta}{\sup} \  \underset{u\in [0,1]}{\sup}  \  \underset{0 \leq r \leq s}{\max} \  \left\|G_u^{(i,r)}(\symbfit{\varepsilon}_{0},\theta)\right\|_{\mathcal{L}^q(\theta)}   &\leq \Lambda.\end{align*} 

\end{assumption}

Second, we impose a condition to control the nonstationarity, which allows for smooth changes, abrupt change-points, and ``rough'' nonstationarity. Define \begin{align*} \left\|(G_u^{(i,r)}(\symbfit{\varepsilon}_0,\theta))_u\right\|_{\kappa\text{-}\mathrm{var},\mathcal{L}^q(\theta)}&=  \underset{\ell\in\mathbb{N}}{\underset{0=u_0<\ldots<u_{\ell}=1}{\sup}} (\sum_{j=1}^{\ell} \left\| G_{u_j}^{(i,r)}(\symbfit{\varepsilon}_{0},\theta) - G_{u_{j-1}}^{(i,r)}(\symbfit{\varepsilon}_{0},\theta)\right\|_{\mathcal{L}^q(\theta)}^{\kappa})^{\frac{1}{\kappa}}. \end{align*}

\begin{assumption}\label{asmpt:nonstationarity_continuous_time}
    For some $q > 2$ and some $\kappa\in [1,4)$, the constant $\Lambda >0$ from Assumption~\ref{asmpt:temporal_dependence_continuous_time} also satisfies  \begin{align*}   
\underset{i \in\mathbb{N}}{\sup} \ \underset{\theta\in\Theta}{\sup} \ \underset{0\leq r \leq s}{\max} \left( \underset{u\in [0,1]}{\sup} \ \left\|G_u^{(i,r)}(\symbfit{\varepsilon}_{0},\theta)\right\|_{\mathcal{L}^q(\theta)}  +  \left\|\left(G_u^{(i,r)}(\symbfit{\varepsilon}_0,\theta)\right)_u\right\|_{\kappa\text{-}\mathrm{var},\mathcal{L}^q(\theta)} \right)  &\leq \Lambda.\end{align*} 

\end{assumption} 

Next, we specify the rates of growth for the window size $w_n$, the offset $\tau_n$, and the lag $L_n$ from~\eqref{eqn:rolling_window_sample_discrepancy}. This ensures that we can improve our estimates of the local expectation values of the random features at $u\in [0,1]$ as we gain more observations.

\begin{assumption}\label{asmpt:growth_of_offset_and_window_size_continuous_time} Recall $q>2$ from Assumption~\ref{asmpt:temporal_dependence_continuous_time} and $\kappa\in [1,4)$ from Assumption~\ref{asmpt:nonstationarity_continuous_time}. Assume that $q$, $\kappa$, $(w_n)_{n\in\mathbb{N}}$, $(\tau_n)_{n\in\mathbb{N}}$, and $(L_n)_{n\in\mathbb{N}}$ satisfy:
{\renewcommand{\labelenumi}{(\alph{enumi})}\begin{enumerate}
\item $\tau_n = o(n)$.
\item $w_n \gg n^{\frac{2}{q}}$, $w_n=o(n)$, and $w_n n^{\max(1-2/\kappa,0)}=o(n)$.
\item $L_n \gg \log(n)^{1+a}$ for some $a\in (0,1)$, $L_n \ll \tau_n$, $L_n \ll n^{\frac{1}{\kappa}-\frac{1}{4}}$, and $L_n =o(n^{\frac{1}{2 \kappa}})$.
\end{enumerate}}
\end{assumption}

We now establish the consistency of the rolling-window estimator $\hat{\theta}^{\mathrm{RW}}$.

\begin{theorem}\label{thm:consistency_rolling_window_estimator} 
    If Assumptions~\ref{asmpt:algorithmic_dynamic_model_continuous_time}-\ref{asmpt:growth_of_offset_and_window_size_continuous_time} hold, then $\hat{\theta}^{\mathrm{RW}}\xrightarrow[]{p}\theta_0$. 
\end{theorem}

\section{Experiments}
\label{section:experiments}

We apply our methods to fit several models. Additional results are reported in the Supplementary Material. Each density plot is based on $1{,}000$ independent estimates of the $p$-dimensional parameter using $2p+1$ random features and $10$ simulations per parameter value. Each estimate is based on an independent realization of the process and an independent draw of $2p+1$ random features.

The experiments reported below were conducted in Python. We optimized the objective functions using the differential evolution solver \texttt{differential\_evolution} from the \texttt{scipy.optimize} module in the SciPy library, though other optimizers perform comparably. For the examples with differential equations (ODEs), we used the Runge--Kutta 5(4) solver \texttt{RK45} from the \texttt{scipy.integrate} module in SciPy.

For the rolling-window estimator, the parameters $L_n$, $w_n$, and $\tau_n$ are chosen as follows. The theory treats these parameters as prespecified deterministic sequences, but here we choose $w_n$ based on a practical selection rule. The lag $L_n$ grows as $\lceil \frac{1}{10}\log(n)^2\rceil$, the window size $w_n$ is selected as the integer within $[n^{\frac{1}{2}},n^{\frac{3}{4}}]$ that minimizes $\sum_{t=m+1+L_n}^n \left\|F_{t-L_n}^{\mathrm{obs}}-f_t^{\mathrm{obs}}\right\|^2$, and the offset is set to $\tau_n=w_n$.

\subsection{Experiments with time-average estimator}\label{subsection:experiments_with_ta_estimator}

\textbf{Moving average process.} For $t=1,\ldots,n$, we observe
\begin{equation*}
X_t =  \mu+ \psi \varepsilon_{t-1} + \varepsilon_{t}, 
\end{equation*} where the noise inputs $\varepsilon_{t}$ are sampled iid from the intractable g-and-k distribution~\citep{gk_MLE}. We fix the location, kurtosis, and $c$ components to $0$, $0.1$, and $0.8$, respectively. The $\sigma$ and $g$ components control the scale and skewness. We aim to estimate $\psi_0=0.5$, $\mu_0=-3$, $\sigma_0=0.2$, $g_0=0.1$.

\begin{figure}[H]
%\captionsetup{width=\linewidth, margin=0pt}
  \centering \includegraphics[width=1\linewidth]{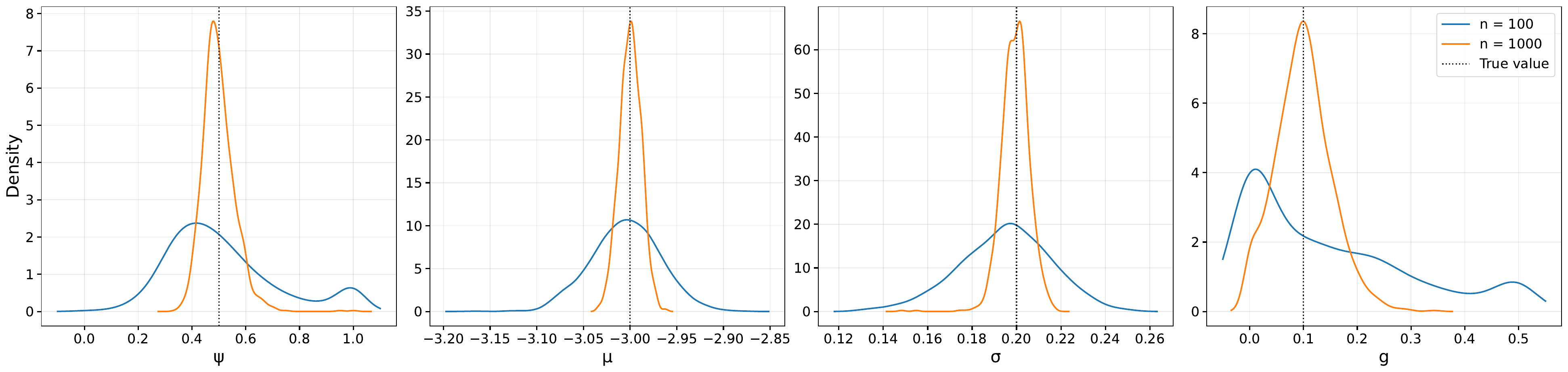}
  \caption{\textbf{Moving average process.} Density of estimates using $9$ features.}
\par\smallskip{\raggedright
\noindent\textbf{Alt text:} Density of estimates at sample sizes $n=100$ and $n=1{,}000$, showing the concentration of the estimates around the true parameter values as $n$ grows.\par}
\label{fig:moving_average_estimate_densities_params0-3} 
\end{figure}

\textbf{Logistic map model.}~\cite{May_1976_logistic} popularized the logistic map as a population model that exhibits chaotic dynamics. For $t=1,\ldots,n$, define the latent state by $$Z_{t} = \rho Z_{t-1}(1-Z_{t-1}),$$ where the unknown initial value $Z_0$ is sampled from a $U[0,1]$ distribution. We observe $$X_t = Z_{t} + \sigma \varepsilon_t,$$ where each $\varepsilon_t \overset{\mathrm{iid}}{\sim} N(0,1)$.  Our goal is to estimate $\rho_0=3.9$ and $\sigma_0=0.1$. 

\begin{figure}[H]
  \centering \includegraphics[width=0.5\linewidth]{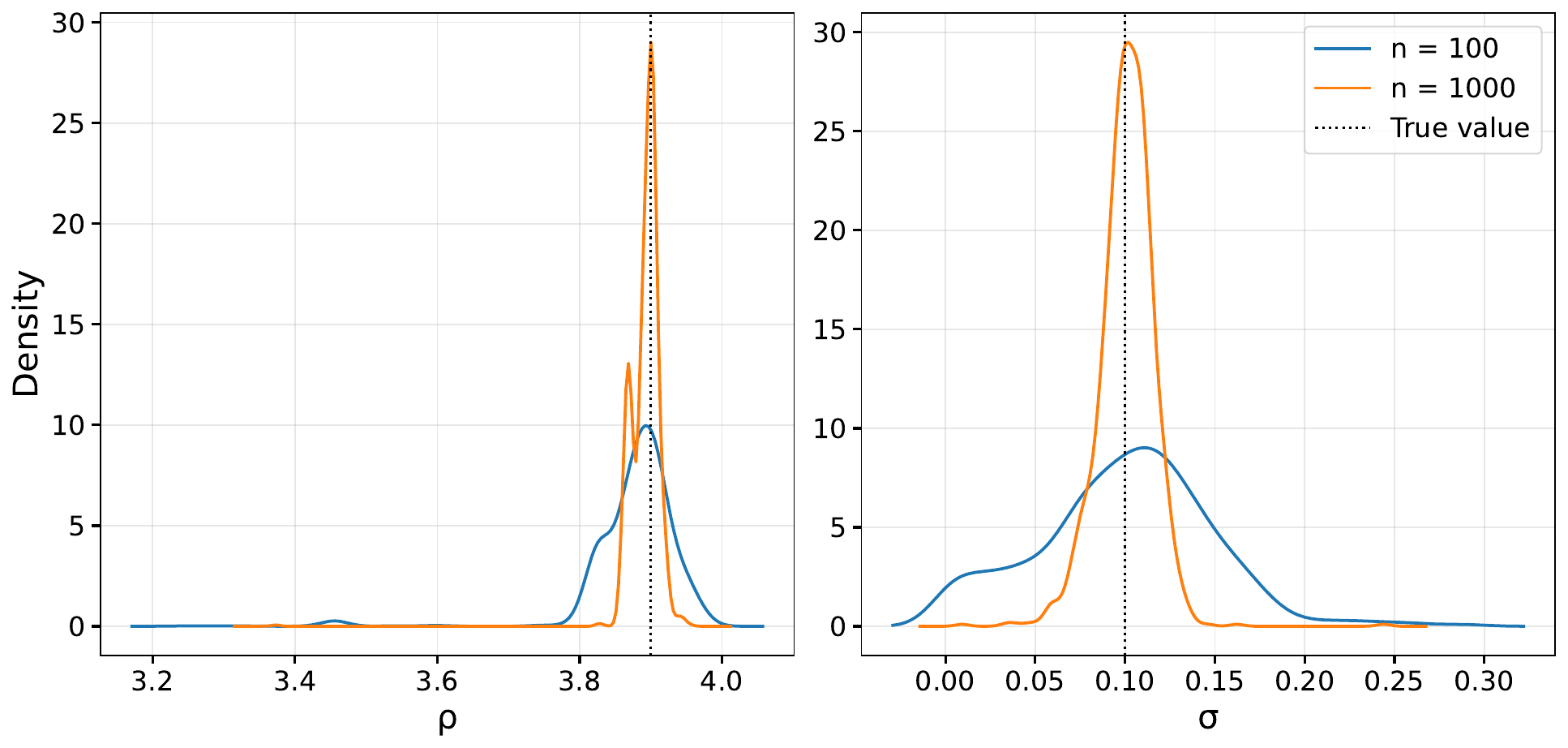}
  \caption{\textbf{Logistic map model.} Density of estimates using $5$ features.}
  \par\smallskip{\raggedright
\noindent\textbf{Alt text:} Density of estimates at sample sizes $n=100$ and $n=1{,}000$, showing the concentration of the estimates around the true parameter values as $n$ grows.\par}
\label{fig:logistic_map_estimate_densities_params0-1}  
\end{figure}

\textbf{State-space model.} For times $t=1,\ldots,n$, define the latent state by 
$$Z_{t} = Z_{t-1} + \psi (\mu - Z_{t-1}) + \sigma \varepsilon_t,$$ where each $\varepsilon_t \overset{\mathrm{iid}}{\sim} N(0,1)$. For dimensions $j=1,\ldots,d\equiv 25$, we observe \begin{align*}X_{t}^{(j)} &= Z_{t} + \lambda \xi_{t}^{(j)}, 
\end{align*} where $\xi_{t}^{(j)} \overset{\mathrm{iid}}{\sim} N(0,1)$. We aim to estimate $\psi_0=0.25$, $\sigma_0=0.08$, $\mu_0 = 2.5$, $\lambda_0=0.05$. 
%%% Note: To situate this model within our theoretical framework, for the process $Z_{t}$, we let the number of noise inputs in the burn-in period grow with the sample size. 

\begin{figure}[H]
%\captionsetup{width=\linewidth, margin=0pt}
  \centering \includegraphics[width=1\linewidth]{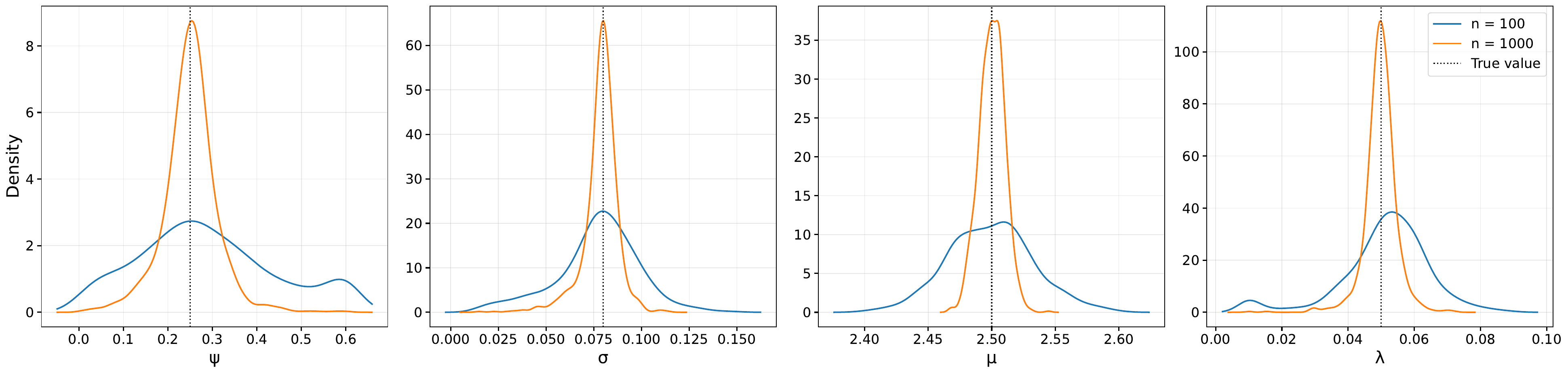}
  \caption{\textbf{State-space model.} Density of estimates using $9$ features.}
\par\smallskip{\raggedright
\noindent\textbf{Alt text:} Density of estimates at sample sizes $n=100$ and $n=1{,}000$, showing the concentration of the estimates around the true parameter values as $n$ grows.\par}
\label{fig:state_space_estimate_densities_params0-3} 
\end{figure}

\subsection{Experiments with rolling-window estimator}\label{subsection:experiments_rolling_window}

\textbf{SIR model.}~\cite{SIR_kermack_mckendrick_1927} introduced the susceptible-infected-recovered (SIR) model to describe infectious disease transmission dynamics. For time horizon $T=50$, let $(S(v), I(v), R(v))_{v\in [0,T]}$ be defined by the system of ODEs \begin{equation*}\label{eq:SIR_continuous} \begin{aligned}  \frac{d S}{dv}(v) &= - \frac{\beta}{N} S(v)\, I(v), \quad  \frac{d I}{dv}(v) &= \frac{\beta}{N} S(v)\, I(v) - \gamma I(v), \quad  \frac{d R}{dv}(v) &= \gamma I(v),\end{aligned} \end{equation*} where $S(v)/N$, $I(v)/N$, and $R(v)/N$ are the fractions of the population that are susceptible, infected, and recovered at time $v\in [0,T]$, respectively, $\beta>0$ is the transmission rate, and $\gamma>0$ is the recovery rate. Let $S(0)=980{,}000$, $I(0)=20{,}000$, $R(0)=0$ be the known initial values, and let $N=S(0)+I(0)+R(0)=1{,}000{,}000$. 

For $t=1,\ldots,n$, we test $\mathrm{TE}_t = \lfloor N/100 \rfloor$ people using a test with sensitivity $\mathrm{SE}_t = 0.95$ and specificity $\mathrm{SP}_t = 0.99$. We observe the proportion of positive tests \begin{align*} X_t=Z_t/\mathrm{TE}_t, \quad  Z_t \sim \mathrm{Binomial}\left(
\mathrm{TE}_t, \pi_{t}
\right),\end{align*} where $\pi_{t}=\mathrm{SE}_t I(Tt/n)/N + (1-\mathrm{SP}_t)\left(1-I(Tt/n)/N\right)$. We let $\mathrm{TE}_t$, $\mathrm{SE}_t$, $\mathrm{SP}_t$ change over time $t$, but we keep them constant. We aim to estimate $\beta_0 = 0.5$ and $\gamma_0 = 0.2$.

\begin{figure}[H]
  \centering \includegraphics[width=0.5\linewidth]{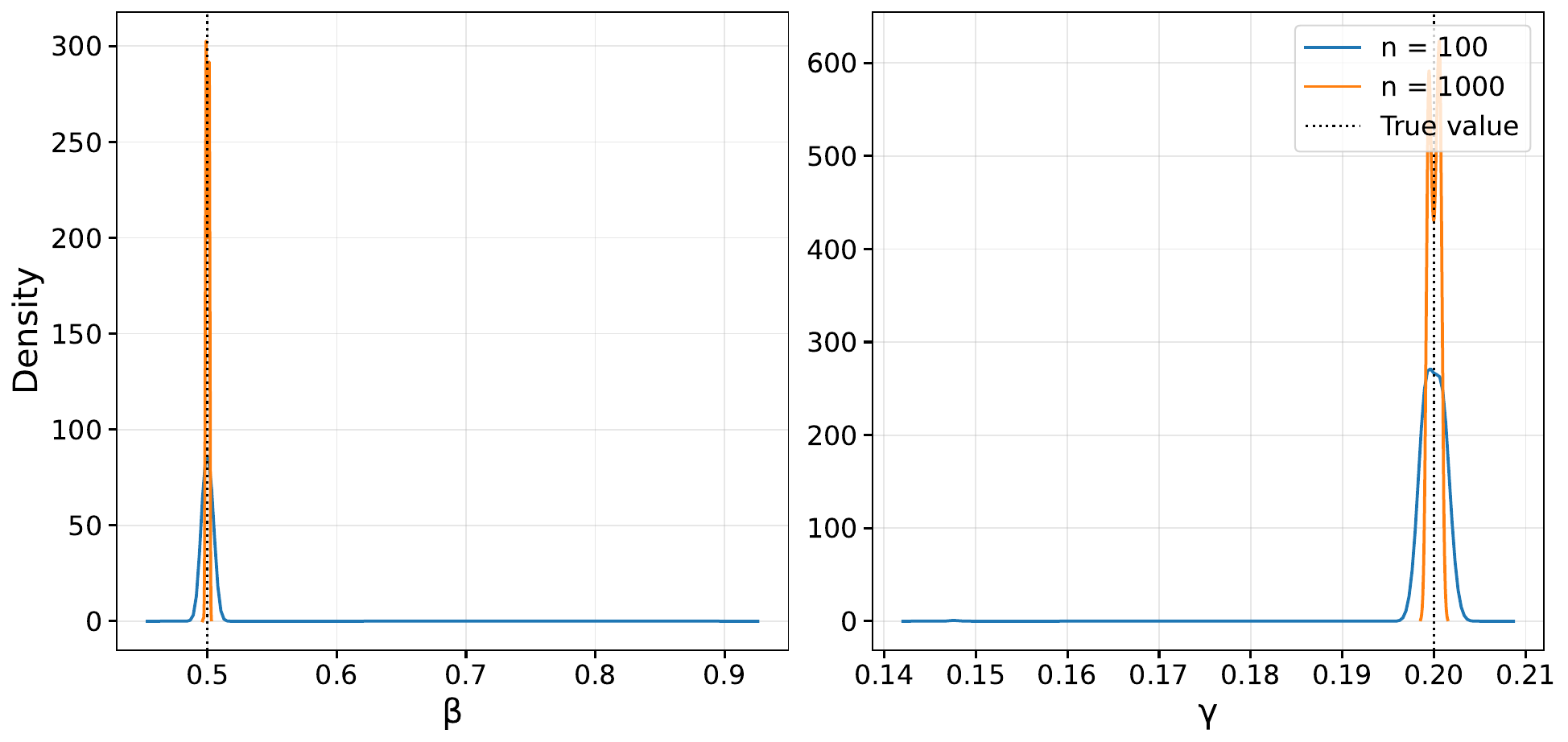}
  \caption{\textbf{SIR model.} Density of estimates using $5$ features.}
\par\smallskip{\raggedright
\noindent\textbf{Alt text:} Density of estimates at sample sizes $n=100$ and $n=1{,}000$, showing the concentration of the estimates around the true parameter values as $n$ grows.\par}  \label{fig:sir_binomial_testing_estimate_densities_params0-1} 
\end{figure}

\textbf{Structural time series model.} We consider a two-dimensional structural time series model with cycles, trends, and an abrupt change-point. For $u\in [0,1]$, define \begin{align*}  \label{eq:sine_cos_latent}Z^{(1)}(u)=\mu^{(1)}u+\alpha \mathbf{1}_{u \geq \tau}+\mathrm{cos}(2 \pi \beta^{(1)} u), \quad Z^{(2)}(u)=\mu^{(2)}u-\alpha \mathbf{1}_{u \geq \tau}+\mathrm{sin}(2 \pi \beta^{(2)} u).\end{align*} For $t=1,\ldots,n$, we observe \begin{equation*}\label{eq:sine_cos_obs}\begin{aligned}X_{t}^{(1)} &= Z^{(1)}(t/n) + \sigma^{(1)}  \varepsilon_{t}^{(1)},\quad X_{t}^{(2)} &= Z^{(2)}(t/n) + \sigma^{(2)}  \varepsilon_{t}^{(2)},\end{aligned}\end{equation*} where each $\varepsilon_{t}^{(1)},\varepsilon_{t}^{(2)} \overset{\mathrm{iid}}{\sim} N(0,1)$ and each $\sigma^{(1)}, \sigma^{(2)}> 0$. We aim to estimate 
\[ \begin{aligned} \alpha_0 &= 1.8, & \tau_0 &= 0.28, & \mu_0^{(1)}&=-1.5, & \mu_0^{(2)}&=0.4, \\ \beta_0^{(1)} &= 2.8,   & \beta_0^{(2)} &= 2.8, & \sigma_0^{(1)}&= 0.9, & \sigma_0^{(2)}&= 0.9. \end{aligned} \] 
This model is not covered by our theory because Assumptions~\ref{asmpt:reparametrization_continuous_time} and~\ref{asmpt:stochastic_equicontinuity_continuous_time} are violated due to the change-point time parameter $\tau$. Nevertheless, our method recovers $\tau$ well. We suspect these assumptions can be relaxed, but this requires a different proof technique. Our theory does cover unknown parameters for smooth change-points (e.g., parameterized by a logistic curve) and models with known change-point times.
\begin{figure}[H]
%\captionsetup{width=\linewidth, margin=0pt}
  \centering \includegraphics[width=1\linewidth]{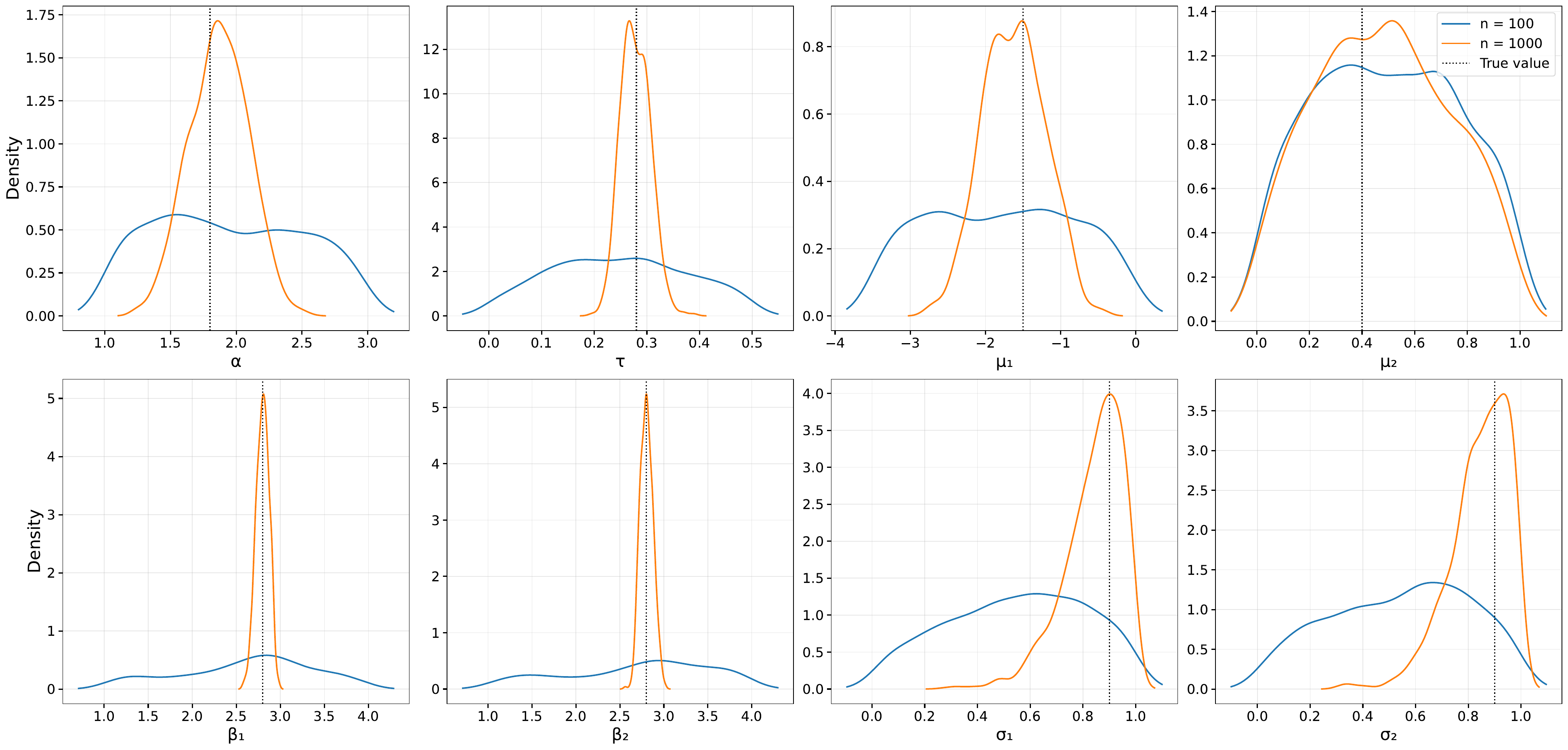}
  \caption{\textbf{Structural time series model.} Density of estimates using $17$ features.}
\par\smallskip{\raggedright
\noindent\textbf{Alt text:} Density of estimates at sample sizes $n=100$ and $n=1{,}000$, showing the concentration of the estimates around the true parameter values as $n$ grows.\par}  \label{fig:structural_timeseries_estimate_densities_params0-7} 
\end{figure}

%%% Note: we can handle change-points if we reparametrize e.g. instead of the parameter being the change point rescaled time, the parameter is the functional that experiences the change-point. so for any two different original parameters (rescaled times for the change point) we will have different distributions for a measurable amount of time due to differences in this functional due to the change-point 
%%% Note: that our theoretical guarantees only apply to models in which the distributions (at each rescaled time) change smoothly with the parameter. The abrupt change-point in this model violates this assumption. Nevertheless, in practice, our method still reliably estimates the location of the change-point. On the other hand, our theoretical framework allows for \textit{known} abrupt change-points and general forms of ``rough'' nonstationarity. We emphasize that our framework does provide guarantees for unknown \textit{smooth} change-points, which can be parametrized using logistic functions, for instance.

\section{Conclusion}
\label{section:conc}

We present an alternative to the existing paradigms for simulation-based estimation. Typically, simulation-based methods rely on user-chosen summary statistics or on neural network-based representations. Both approaches have large costs. In contrast, we reduce the problem of estimating a $p$-dimensional parameter to the problem of matching $2p+1$ generic random Fourier features of the observed and simulated data.

The principle of matching random features extends beyond models for time series. This suggests a broader research programme on random feature methods for various dependence structures, e.g., spatiotemporal and network data. Also, an inspection of our proofs reveals that random Fourier features are not necessary. In general, the random features need only form a function class that is rich enough to uniquely identify the distribution. This opens the door to alternative classes of random features for non-Euclidean data. We leave these directions for future work.
%%% Note: The key technical ingredients are concentration inequalities under dependence and control of the variation of the data-generating mechanism, e.g., over space and time.

% That is, one first estimates the parameter, then simulates datasets from the fitted model and re-estimates the parameter based on each simulated dataset, then finally calculates the empirical standard deviation or variance-covariance matrix of the re-estimates.
Confidence sets can be constructed by matching random features via a model-based bootstrap procedure. In future work, we will establish the asymptotic normality of both estimators under a stronger set of assumptions. All results concerning simulation-based inference will be reported in a separate manuscript.
%%% Note: We will also develop a method for estimating the corresponding asymptotic covariance matrix using either a simulation-based approach (estimate parameter, simulate from fitted model, Monte Carlo estimate of covariance) or time-series approach based on the observed time series.

%%%%%%%%%%%%%%

\section{Competing interests}
No competing interest is declared.

\section{Author contributions statement}
M.W.-S. and C.S. both contributed to the conceptualization, writing, and editing of all versions of this manuscript.

\section{Data availability statement} The code used to generate the data and reproduce the results of the experiments is publicly available at \url{https://github.com/mwiecksosa/random_features} and is archived on Zenodo at
\url{https://doi.org/10.5281/zenodo.23225189}.

\section{Acknowledgments} The authors gratefully acknowledge support from NSF grant DMS-2310834.
% The authors thank the anonymous reviewers for their valuable suggestions. 

\textit{AI use statement:} Starting from core ideas of the authors, we acknowledge that generative AI was used to: identify and search for relevant literature, brainstorm ideas involving the conceptual frameworks, suggest edits
to the paper to improve prose, suggest edits to the proofs to check correctness, shorten length, and improve readability, suggest edits to software code, experiments, and figures, identify possible gaps in the paper, and provide feedback on experiments.

Generative AI was not used to: draft parts of the paper, format references, formulate the presented mathematical claims, provide critical ingredients for proving mathematical claims, suggest a structure for the paper, assist with translation, interpret results, generate data sets, or propose a title or keywords for the paper.

The authors have reviewed all content suggested by generative AI and produced with the aid of generative AI. For example, the authors have reviewed all literature found by generative AI and checked the correctness of any suggestions made by generative AI. We take responsibility for the final content of this work, including text, claims or artifacts produced with the aid of generative AI.

%\bibliographystyle{abbrvnat}
%\bibliography{estimation_bibliography}

\putbib

\end{bibunit}

\pagebreak

\begin{appendices}

\begin{bibunit}

{\Large\textbf{Supplementary material for ``Estimating dynamic models by matching random features''}}
 
\makeatletter
\section*{Contents}
\@starttoc{toc}
\makeatother

\section{Proofs}\label{section:proofs}

\subsection{Preliminaries}\label{subsection:preliminaries}

% Note: we cannot directly use their "simple" finite witness theorem because it is only stated for countable unions of affine sets, so we'll just need to use the more general finite witness theorem. 

In this subsection, we state several definitions and theoretical results from different papers. We use the original notation from those papers.

First, we introduce the definitions of $\sigma$-subanalytic sets and $\sigma$-subanalytic functions, which require the definitions of globally subanalytic sets, globally semianalytic sets, and semianalytic sets. Figure~\ref{fig:diagram_of_subanalytic_sets_and_functions} (Figure 3 in \cite{amir_et_al_finite_witness_thm}) summarizes the relationships among the different classes of sets and functions that we discuss in this paper. For more information, see the appendix of \cite{amir_et_al_finite_witness_thm}.

\begin{figure}[h!]
%\captionsetup{width=\linewidth, margin=0pt}
  \centering \includegraphics[width=1\linewidth]{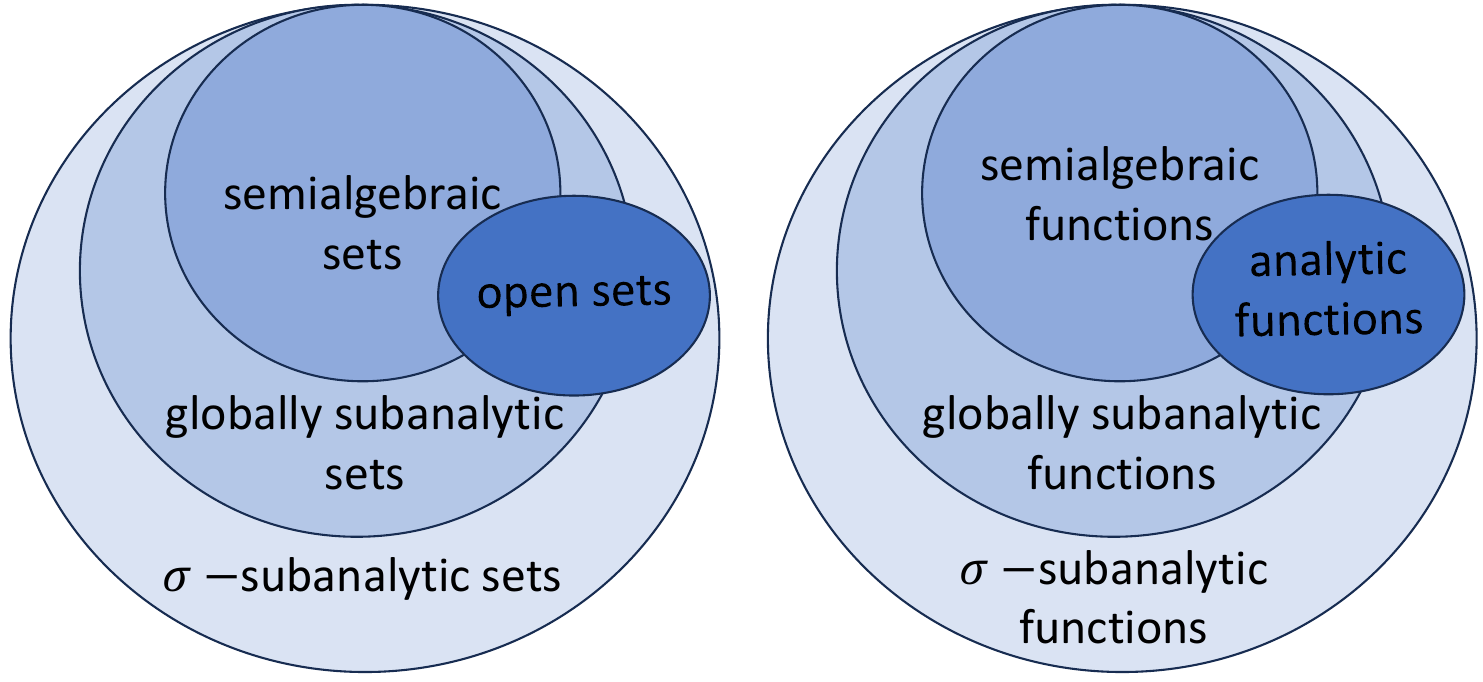}
  \caption{Relationships among classes (Figure 3 in \cite{amir_et_al_finite_witness_thm}).}
  \par\smallskip{\raggedright
\noindent\textbf{Alt text:} Diagram showing the relationships among classes of sets and functions. Semialgebraic classes are contained within globally subanalytic classes, which are contained within $\sigma$-subanalytic classes. Open sets and analytic functions are contained within the $\sigma$-subanalytic classes.\par}\label{fig:diagram_of_subanalytic_sets_and_functions} 
\end{figure}

\begin{definition}[Definition 1.1.1 in \cite{Valette_subanalytic}]
	A subset $E\subseteq \mathbb{R}^n$ is called \textit{semianalytic} if it is locally defined
	by finitely many real analytic equalities and inequalities. Namely, for each $a\in \mathbb{R}^n$,
	there is a neighbourhood $U$ of $a$, and real analytic functions $f_{ij}, g_{ij}$ on $U$, where
	$i=1,\ldots,r$ and $ j= 1,\ldots s_i$, such that
	\begin{equation}\label{eq:semi}
	E\cap U=\bigcup_{i=1}^r \bigcap_{j=1}^{s_i}\{\mathbf{z}\in U| \, g_{ij}(\mathbf{z})>0 \text{ and } f_{ij}(\mathbf{z})=0 \}. \end{equation}
\end{definition}

\begin{definition}[Definition 1.1.3 in \cite{Valette_subanalytic}]
	A subset $Z\subseteq \mathbb{R}^n$ is \textit{globally semianalytic} if $V_n(Z)$ is a semianalytic
	subset of $\mathbb{R}^n$, where $V_n : \mathbb{R}^n \to (-1,1)^n$ is the homeomorphism defined by
	$$V_n\left(\mathbf{z}=(z_1,\ldots,z_n)\right)=\left(\frac{z_1}{\sqrt{1+|\mathbf{z}|^2}},\ldots , \frac{z_n}{\sqrt{1+|\mathbf{z}|^2}} \right). $$
\end{definition}

\begin{definition}[Definition 1.1.6 in \cite{Valette_subanalytic}]
	A subset $E\subseteq \mathbb{R}^n$ is \textit{globally subanalytic} if it can be presented
	as a linear projection of a globally semianalytic set; more precisely, if there
	exists a globally semianalytic set $Z \subseteq \mathbb{R}^{n+p}$  such that $E=\pi(Z) $, with   $\pi:\mathbb{R}^{n+p}\to \mathbb{R}^n$ being the projection operator that omits the last $p$ coordinates while leaving the remaining coordinates unchanged.
\end{definition}

\begin{definition}[Definition A.10 in \cite{amir_et_al_finite_witness_thm}]\label{def:sigma_subanalytic_set}
We say that a subset $A\subseteq\mathbb{R}^D$ is $\sigma$\textit{-subanalytic} if it is a countable union of globally subanalytic subsets of $\mathbb{R}^D$.
\end{definition}

\begin{definition}[Definition A.13 in \cite{amir_et_al_finite_witness_thm}]\label{def:sigma_subanalytic_function}
Let $\mathbb{M}\subseteq \mathbb{R}^D$ be a $\sigma$-subanalytic set. We say that $f:\mathbb{M}\xrightarrow[]{}\mathbb{R}^L$ is a $\sigma$\textit{-subanalytic function} if its graph is a $\sigma$-subanalytic subset of $\mathbb{R}^{D+L}$.
\end{definition}

Second, we present some properties of $\sigma$-subanalytic sets.

\begin{theorem}[Proposition A.11 in \cite{amir_et_al_finite_witness_thm}]\label{thm:sigma_subanalytic_properties}
	Assume that $A,B \subseteq \mathbb{R}^D $ and $C \subseteq \mathbb{R}^M$ are $\sigma$-subanalytic sets. Then:
	{\renewcommand{\labelenumi}{(\alph{enumi})}\begin{enumerate}
		\item $A \cup B$ is $\sigma$-subanalytic. More generally, any countable union of $\sigma$-subanalytic sets is $\sigma$-subanalytic.
        \item $A\cap B$ is $\sigma$-subanalytic.
		\item $A\times C $ is $\sigma$-subanalytic.
		\item If $\pi:\mathbb{R}^D\to \mathbb{R}^L $ is a linear projection, then $\pi(A) $ is a $\sigma$-subanalytic set.
		\item $A $ is a \emph{countable} union of $C^\infty $ manifolds.
	\end{enumerate}}
\end{theorem}

Third, we state a corollary of the finite witness theorem (Theorem A.2) from \cite{amir_et_al_finite_witness_thm}. The notion of dimension is the Hausdorff dimension. When $\mathbb{W}=\mathbb{R}^q$ or is an open subset of $\mathbb{R}^q$, the result holds for almost every $\left( \vartheta_{1},\ldots,\vartheta_{2D+1}\right)$ w.r.t.\ Lebesgue measure.

\begin{theorem}[Corollary A.20 in \cite{amir_et_al_finite_witness_thm}]\label{thm:finite_witness}
Let $\mathbb{M} \subseteq \mathbb{R}^p$, $\mathbb{W} \subseteq \mathbb{R}^q$ be $\sigma$-subanalytic sets of dimension $D$ and $D_{\vartheta}$ respectively. Let $F:\mathbb{M} \times \mathbb{W} \to \mathbb{R}$ be a $\sigma$-subanalytic function. Define the set
$$\mathcal{N}=\{\left(\mathbf{x},\mathbf{y}\right) \in \mathbb{M} \times \mathbb{M} \mid F(\mathbf{x};\vartheta)=F(\mathbf{y};\vartheta),\ \forall \vartheta\in \mathbb{W} \}.$$
Suppose that for all $\left(\mathbf{x},\mathbf{y}\right) \in \mathbb{M} \times \mathbb{M} \setminus \mathcal{N}$
\begin{equation}\label{cond_dimension_defficiency_for_separation}
\dim\{\vartheta\in \mathbb{W}\mid F(\mathbf{x};\vartheta)=F(\mathbf{y};\vartheta) \}\leq D_{\vartheta}-1.\end{equation}
Then for generic $\left( \vartheta_{1},\ldots,\vartheta_{2D+1}\right) \in \mathbb{W}^{2D+1}$,
\begin{equation}\label{eqn:equality_of_functions_2D+1}
\mathcal{N}=\{\left(\mathbf{x},\mathbf{y}\right) \in \mathbb{M} \times \mathbb{M} \mid F(\mathbf{x};\vartheta_{i})=F(\mathbf{y};\vartheta_{i}),\ \forall i=1,\ldots 2D+1\} .
\end{equation}
Moreover, if $\mathbb{W}$ is an open and connected subset of $\mathbb{R}^q$, and $F(\mathbf{x};\vartheta)$ is analytic as a function of $\vartheta$ for all fixed $\mathbf{x} \in \mathbb{M}$, then condition~\eqref{cond_dimension_defficiency_for_separation} is not required, as it is automatically satisfied.
\end{theorem}

The next result is the measurable maximum theorem from \cite{aliprantis_border_inf_dim_analysis_2006}.

\begin{theorem}[Theorem 18.19 in \cite{aliprantis_border_inf_dim_analysis_2006}]\label{thm:mble_max_theorem}
Let $X$ be a separable metrizable space and $(S,\Sigma)$ a measurable space. Let
$\varphi:S\xrightarrow[]{} X$ be a weakly measurable correspondence with nonempty
compact values, and suppose $f:S\times X\to\mathbb{R}$ is a Carath\'eodory function.
Define the value function $m:S\to\mathbb{R}$ by
$$m(s)=\max_{x\in\varphi(s)} f(s,x),$$
and the correspondence $\mu:S\xrightarrow[]{} X$ of maximizers by $$\mu(s)=\{x\in\varphi(s):f(s,x)=m(s)\}.$$
Then:
{\renewcommand{\labelenumi}{(\alph{enumi})}\begin{enumerate}
    \item The value function $m$ is measurable.
    \item The argmax correspondence $\mu$ has nonempty and compact values.
    \item The argmax correspondence $\mu$ is measurable and admits a measurable selector.
\end{enumerate}}
\end{theorem}

The next result is the Tarski--Seidenberg theorem presented in \cite{bierstone_milman_subanalytic_set}.
\begin{theorem}[Theorem 1.5 in \cite{bierstone_milman_subanalytic_set}]\label{thm:tarski_seidenberg_theorem}
    The image of a semialgebraic set $X\subset \mathbb{R}^{n+1}$ by the projection $\pi:\mathbb{R}^{n+1}\xrightarrow[]{}\mathbb{R}^n$ is semialgebraic.
\end{theorem}

The following result is the uniqueness theorem for characteristic functions.

\begin{theorem}[Proposition 2.9 in \cite{multivariate_stats_book}]\label{thm:uniqueness_characteristic_function} The characteristic function $c:\mathbb{R}^n \to \mathbb{C}$ of $\mathbf{x}$, defined by
$c(\mathbf{t})
=
c_{\mathbf{x}}(\mathbf{t})
=
\mathbb{E}[\exp(i\mathbf{t}'\mathbf{x})]$,
is unique: $$\mathbf{x} \overset{d}{=} \mathbf{y}
\quad\iff\quad
c_{\mathbf{x}}(\mathbf{t})
=
c_{\mathbf{y}}(\mathbf{t}),
\quad
\forall\,\mathbf{t}\in\mathbb{R}^n.$$
\end{theorem}

\subsection{Proof of Theorem~\ref{thm:injective_discrete_time}}\label{subsection:proof_of_injective_discrete_time}

\paragraph*{Step 1: Limiting expectations.} We claim that under Assumptions~\ref{asmpt:algorithmic_dynamic_model_discrete_time}-\ref{asmpt:weak_convergence_of_time_avg_of_distributions_discrete_time}, $\forall \theta \in \Theta$, the limiting expectation values of the $k$ random features of the time series are given by \begin{equation*}    \Phi(\theta)=\int_{\mathbb{R}^{(m+1)\times d}}\varphi(x) \, dP_{\theta}(x).\end{equation*} That is, for all $\theta \in \Theta$ and all $r \in \{0,1,\ldots,s\}$, we have, as $n\xrightarrow[]{}\infty$, 
\begin{equation}\label{eqn:convergence_to_Phi_theta_discrete_time}\left\|\frac{1}{n-m}\sum_{t=m+1}^n \mathbb{E}_{\theta}\left[ \varphi\left([G_{t-m}^{(n,r)}(\symbfit{\varepsilon}_{t-m}^{(r)},\theta),\ldots,G_t^{(n,r)}(\symbfit{\varepsilon}_t^{(r)},\theta)]^{\top}\right)\right] - \Phi(\theta)\right\|=o(1).\end{equation}  Note that the time series have this causal representation by Assumption~\ref{asmpt:algorithmic_dynamic_model_discrete_time}.

Observe that \begin{align*} &
\left\|
\frac{1}{n-m}\sum_{t=m+1}^n 
\mathbb{E}_{\theta}\left[
\varphi\left(
\left[G_{t-m}^{(n,r)}(\symbfit{\varepsilon}_{t-m}^{(r)},\theta),\ldots,
G_t^{(n,r)}(\symbfit{\varepsilon}_t^{(r)},\theta)\right]^{\top}
\right)
\right] - \Phi(\theta)
\right\|
\\
&
\overset{(1)}{\leq} \left\|
\begin{aligned}
&\frac{1}{n-m}\sum_{t=m+1}^n 
\mathbb{E}_{\theta}\left[
\varphi\left(
\left[G_{t-m}^{(n,r)}(\symbfit{\varepsilon}_{t-m}^{(r)},\theta),\ldots,
G_t^{(n,r)}(\symbfit{\varepsilon}_t^{(r)},\theta)\right]^{\top}
\right)
\right] \\
&-
\frac{1}{n-m}\sum_{t=m+1}^n 
\mathbb{E}_{\theta}\left[
\varphi\left(
\left[G_{t-m}^{(r)}(\symbfit{\varepsilon}_{t-m}^{(r)},\theta),\ldots,
G_t^{(r)}(\symbfit{\varepsilon}_t^{(r)},\theta)\right]^{\top}
\right)
\right]
\end{aligned}
\right\| 
\\ &
+
\norm{
\frac{1}{n-m}\sum_{t=m+1}^n 
\mathbb{E}_{\theta}\left[
\varphi\left(
\left[G_{t-m}^{(r)}(\symbfit{\varepsilon}_{t-m}^{(r)},\theta),\ldots,
G_t^{(r)}(\symbfit{\varepsilon}_t^{(r)},\theta)\right]^{\top}
\right)
\right] - \Phi(\theta)
} 
\\
&
\overset{(2)}{=} o(1),
\end{align*} where (1) is from adding and subtracting the corresponding term with the limiting mapping and by applying the triangle inequality, and (2) holds because both terms are $o(1)$. Specifically, the second term is $o(1)$ by upper bounding the $\ell_2$ norm $\left\|\cdot\right\|$ by the $\ell_1$ norm (i.e., over the $k$ random features) using the monotonicity of norms, linearity, and then applying the weak convergence from Assumption~\ref{asmpt:weak_convergence_of_time_avg_of_distributions_discrete_time}. The first term is $o(1)$ because the approximations can be asymptotically replaced by the limiting mappings, i.e., because for all $\theta \in \Theta$ and all $r \in \{0,1,\ldots,s\}$, we have \begin{align*} &
\left\|
\begin{aligned}
&\frac{1}{n-m}\sum_{t=m+1}^n 
\mathbb{E}_{\theta}\left[
\varphi\left(
\left[G_{t-m}^{(n,r)}(\symbfit{\varepsilon}_{t-m}^{(r)},\theta),\ldots,
G_t^{(n,r)}(\symbfit{\varepsilon}_t^{(r)},\theta)\right]^{\top}
\right)
\right] \\
&-
\frac{1}{n-m}\sum_{t=m+1}^n 
\mathbb{E}_{\theta}\left[
\varphi\left(
\left[G_{t-m}^{(r)}(\symbfit{\varepsilon}_{t-m}^{(r)},\theta),\ldots,
G_t^{(r)}(\symbfit{\varepsilon}_t^{(r)},\theta)\right]^{\top}
\right)
\right]
\end{aligned}
\right\|
\\
&\overset{(1)}{\leq} 
\frac{1}{n-m}\sum_{t=m+1}^n  \left\|
\begin{aligned}
& 
\varphi\left(
\left[G_{t-m}^{(n,r)}(\symbfit{\varepsilon}_{t-m}^{(r)},\theta),\ldots,
G_t^{(n,r)}(\symbfit{\varepsilon}_t^{(r)},\theta)\right]^{\top}
\right)
 \\
& - 
\varphi\left(
\left[G_{t-m}^{(r)}(\symbfit{\varepsilon}_{t-m}^{(r)},\theta),\ldots,
G_t^{(r)}(\symbfit{\varepsilon}_t^{(r)},\theta)\right]^{\top}
\right)
\end{aligned}
\right\|_{\mathcal{L}^1(\theta)}
\\
& \overset{(2)}{\leq}  L_{\varphi} 
\frac{1}{n-m}\sum_{t=m+1}^n   \left\|
\begin{aligned}
&  
\mathrm{Vec}\left(G_{t-m}^{(n,r)}(\symbfit{\varepsilon}_{t-m}^{(r)},\theta),\ldots,
G_t^{(n,r)}(\symbfit{\varepsilon}_t^{(r)},\theta)\right)
 \\
&- \mathrm{Vec}
\left(G_{t-m}^{(r)}(\symbfit{\varepsilon}_{t-m}^{(r)},\theta),\ldots,
G_t^{(r)}(\symbfit{\varepsilon}_t^{(r)},\theta)\right)
\end{aligned}
\right\|_{\mathcal{L}^1(\theta)}
\\
& \overset{(3)}{\leq}  L_{\varphi} 
\frac{1}{n-m}\sum_{t=m+1}^n \sum_{j=0}^m   \left\|
G_{t-j}^{(n,r)}(\symbfit{\varepsilon}_{t-j}^{(r)},\theta) - 
G_{t-j}^{(r)}(\symbfit{\varepsilon}_{t-j}^{(r)},\theta)
\right\|_{\mathcal{L}^1(\theta)}
\\
& \overset{(4)}{\leq}
 L_{\varphi} 
\frac{1}{1-m/n} (m+1) \left(\frac{1}{n}\sum_{t=1}^n    \left\|G_{t}^{(n,r)}(\symbfit{\varepsilon}_{0}^{(r)},\theta) - 
G_{t}^{(r)}(\symbfit{\varepsilon}_{0}^{(r)},\theta)
\right\|_{\mathcal{L}^1(\theta)}\right)
\\
& \overset{(5)}{=}  o(1),
\end{align*}
where (1) holds by linearity of summation, the triangle inequality, linearity of expectation, and Jensen's inequality, (2) holds because the random Fourier features $\varphi$ are $L_{\varphi}$-Lipschitz with Lipschitz constant given by~\eqref{eqn:lipschitz_constant_rff}, (3) holds by the triangle inequality and Minkowski inequality, (4) follows by upper bounding the original quantity by $m+1$ times the full sum (i.e., by including more non-negative terms in the summation), factoring out $1/n$, and noting that each $(G_t^{(i,r)}(\symbfit{\varepsilon}_j,\theta))_{j\in\mathbb{Z}}$, $(G_t^{(r)}(\symbfit{\varepsilon}_j,\theta))_{j\in\mathbb{Z}}$ is a stationary ergodic process for each fixed $t\in \mathbb{N}$, $i\in \mathbb{N}$, and $r=0,1,\ldots,s$, which allows us to shift the noise inputs, and (5) is from applying the convergence condition from Assumption~\ref{asmpt:limiting_mapping_discrete_time}.

\paragraph*{Step 2: Injectivity.} The main ingredient is the finite witness theorem (Theorem~\ref{thm:finite_witness}), so we show that its assumptions are satisfied. Denote $$\mathbb{W}=\mathbb{R}^{d(m+1)}\times (-\pi,\pi).$$

By Assumption~\ref{asmpt:statistical_model_discrete_time}, $\Theta$ is a compact subset of $\mathbb{R}^p$ with non-empty interior, so the Hausdorff dimension is $\mathrm{dim}_H(\Theta)=p$. By Assumption~\ref{asmpt:statistical_model_discrete_time}, $\Theta$ is a semialgebraic subset of $\mathbb{R}^p$, hence it is $\sigma$-subanalytic; see Figure~\ref{fig:diagram_of_subanalytic_sets_and_functions}.

By construction, the random parameters $\vartheta_i$, $i=1,\ldots,k$ from~\eqref{eqn:parameters_for_the_random_features} for the random Fourier features $\varphi_i$, $i=1,\ldots,k$ from~\eqref{eqn:random_Fourier_features_order_m} are each random variables taking values in $\mathbb{W}$, which is semialgebraic, so, as explained in Figure~\ref{fig:diagram_of_subanalytic_sets_and_functions}, it is $\sigma$-subanalytic.

By Assumption~\ref{asmpt:statistical_model_discrete_time}, the real and imaginary parts of the map $(\theta,\omega)\mapsto \psi_{\theta}(\omega)$ admit jointly real analytic extensions to $O \times \mathbb{R}^{(m+1)\times d}$, so each $(\theta,\vartheta_i)\mapsto \int_{\mathbb{R}^{(m+1)\times d}}\varphi_i(x) \, dP_{\theta}(x)$, $i=1,\ldots,k$, admits a jointly real analytic extension to $O \times \mathbb{W}$ by using the definition of the characteristic function and each random Fourier feature. Note that the graph of this extension is $\sigma$-subanalytic, and that $\Theta$ is a $\sigma$-subanalytic set as discussed previously. Therefore, since $\sigma$-subanalytic sets are closed under finite intersections and Cartesian products by Theorem~\ref{thm:sigma_subanalytic_properties}, we have that the intersection of the graph of this extension with $\Theta \times \mathbb{W}\times \mathbb{R}$ is $\sigma$-subanalytic.

$\mathbb{W}$ is open and connected, and $(\theta,\vartheta_i)\mapsto \int_{\mathbb{R}^{(m+1)\times d}}\varphi_i(x) \, dP_{\theta}(x)$ is analytic as a function of $\vartheta_i$ for all fixed $\theta\in\Theta$. Hence, the condition~\eqref{cond_dimension_defficiency_for_separation} from the finite witness theorem (Theorem~\ref{thm:finite_witness}) is satisfied.

Let $x_1,\ldots,x_{m+1}$ be vectors in $\mathbb{R}^d$, and define $x=(x_1,\ldots,x_{m+1})^{\top}\in\mathbb{R}^{(m+1)\times d}$. Letting $a=(a_1^{\top},\ldots,a_{m+1}^{\top})^{\top}\in\mathbb{R}^{d(m+1)}$, $b\in (-\pi,\pi)$, and $c=(a,b)\in\mathbb{W}$, define $$\phi_c(x)=\cos\left(\sum_{j=1}^{m+1} a_j \cdot x_{j} + b\right).$$ Define $$\mathcal{N}=\left\{\left(\theta_1,\theta_2\right) \in \Theta^2 \mid \int_{\mathbb{R}^{(m+1)\times d}}\phi_c(x) \, dP_{\theta_1}(x) =\int_{\mathbb{R}^{(m+1)\times d}}\phi_c(x) \, dP_{\theta_2}(x),\ \forall c\in\mathbb{W} \right\}.$$ By the finite witness theorem (Theorem~\ref{thm:finite_witness}), for generic parameters $\left( \vartheta_1,\ldots,\vartheta_{k}\right) \in \mathbb{W}^k$ of the $k=2p+1$ random Fourier features $\left( \varphi_1,\ldots,\varphi_{k}\right)$, we have the set equality $$\mathcal{N}=\left\{\left(\theta_1,\theta_2\right) \in \Theta^2 \mid \int_{\mathbb{R}^{(m+1)\times d}}\varphi_i(x) \, dP_{\theta_1}(x) =\int_{\mathbb{R}^{(m+1)\times d}}\varphi_i(x) \, dP_{\theta_2}(x),\ \forall i \in [k]\right\}.$$ By the discussion directly above Theorem~\ref{thm:finite_witness}, generic here means $\Pi$-almost every $\vartheta$, where $\Pi$ is the joint distribution from~\eqref{eqn:parameters_for_the_random_features} which is absolutely continuous w.r.t.\ the Lebesgue measure.

Denoting the $\mathbb{R}^{(m+1)\times d}$-valued random matrix by $\Omega_i=(\Omega_{i,1},\ldots,\Omega_{i,m+1})^{\top}$, we have \begin{align*}&\int_{\mathbb{R}^{(m+1)\times d}}\varphi_i(x) \, dP_{\theta}(x)
\\& 
=\int_{\mathbb{R}^{(m+1)\times d}}\cos\left(\sum_{j=1}^{m+1} \Omega_{i,j} \cdot x_{j} + \alpha_i\right)\, dP_{\theta}(x)
\\&
=
\cos(\alpha_i)
\int_{\mathbb{R}^{(m+1)\times d}}\cos\left(\sum_{j=1}^{m+1} \Omega_{i,j} \cdot x_{j}\right)\, dP_{\theta}(x) 
\\&
- 
\sin(\alpha_i)
\int_{\mathbb{R}^{(m+1)\times d}}\sin\left(\sum_{j=1}^{m+1} \Omega_{i,j} \cdot x_{j}\right)\, dP_{\theta}(x) 
\\&
= \cos(\alpha_i) \  \mathrm{Re}( \psi_{\theta}(\Omega_i)) - \sin(\alpha_i) \ \mathrm{Im}(  \psi_{\theta}(\Omega_i))
\\&
= M_{\vartheta_i}(\theta),
\end{align*}
where $\psi_{\theta}$ is the characteristic function. If $(\theta_1,\theta_2)\in\mathcal{N}$, then $ M_{\vartheta_i}(\theta_1)=M_{\vartheta_i}(\theta_2)$ for every possible value of $\vartheta_i=(\mathrm{Vec}(\Omega_i),\alpha_i)\in\mathbb{W}$. In particular, the expectation values of the random Fourier features correspond to the real and imaginary parts of the characteristic function by taking $\alpha_i=0$ and $\alpha_i=-\pi/2$, respectively, so by the uniqueness theorem (Theorem~\ref{thm:uniqueness_characteristic_function}) the corresponding distributions must be equal, $P_{\theta_1}=P_{\theta_2}$, if $(\theta_1,\theta_2)\in\mathcal{N}$. However, by Assumption~\ref{asmpt:statistical_model_discrete_time}, the map $\theta\mapsto P_{\theta}$ is injective, i.e., if $\theta_1 \neq \theta_2$ then $P_{\theta_1}\neq P_{\theta_2}$, so we must have
$$(\theta_1,\theta_2)\in\mathcal{N}\implies (\theta_1,\theta_2)\in \{(\theta_1,\theta_2)\in \Theta^2 : \theta_1=\theta_2\}=\{(\theta,\theta):\theta\in\Theta\},$$ so $\mathcal{N}\subseteq \{(\theta,\theta):\theta\in\Theta\}$. Conversely, we have $\{(\theta,\theta):\theta\in\Theta\} \subseteq \mathcal{N}$ trivially. Therefore,
\begin{equation}\label{eqn:conclusion_mathcal_N_discrete_time}\mathcal{N}=\{(\theta,\theta):\theta\in\Theta\}.\end{equation}

Putting it all together, for $\Pi$-almost any random parameters $(\vartheta_1,\ldots,\vartheta_k)$ for the $k=2p+1$ random Fourier features $(\varphi_1,\ldots,\varphi_k)$, for every pair of parameter values $\theta_1, \theta_2\in\Theta$, if $\theta_1\neq \theta_2$, then the expectations of the random Fourier features are not equal, $\Phi(\theta_1)\neq \Phi(\theta_2)$. Precisely in this sense, the map $\theta\mapsto\Phi(\theta)$ is one-to-one for generic choices of $k=2p+1$ random Fourier features.

\hfill$\square$

\subsection{Proof of Theorem~\ref{thm:consistency_time_average_estimator}}\label{subsection:proof_of_consistency_time_average_estimator}

\paragraph*{Setup.} We will use Corollary 3.2 of~\cite{pakes_pollard_simulation_asymptotics} to conclude that $\hat{\theta}^{\mathrm{TA}}$ is consistent. The sample discrepancy function, from~\eqref{eqn:time_average_sample_discrepancy}, is given by
$$\hat{Q}_n^{\mathrm{TA}}(\theta) =F^{\mathrm{obs}} - \bar{F}^{\mathrm{sim}}(\theta),$$ and the population discrepancy function is given by $$Q^{\mathrm{TA}}(\theta) =\Phi(\theta_0) - \Phi(\theta).$$ 
To use Corollary 3.2 of~\cite{pakes_pollard_simulation_asymptotics}, it suffices to show that, under the assumptions of Theorem~\ref{thm:consistency_time_average_estimator}, the following conditions hold: 
{\renewcommand{\labelenumi}{(\alph{enumi})}
\begin{enumerate}
    \item $\norm{\hat{Q}_n^{\mathrm{TA}}(\hat{\theta}^{\mathrm{TA}})}\leq \inf_{\theta \in \Theta}\norm{\hat{Q}_n^{\mathrm{TA}}(\theta)} + o_p(1)$.
    \item $\inf_{\norm{\theta - \theta_0} > \delta} \norm{Q^{\mathrm{TA}}(\theta)} >0$ for each $\delta >0$. 
    \item $\sup_{\theta\in\Theta}\frac{\norm{\hat{Q}_n^{\mathrm{TA}}(\theta) - Q^{\mathrm{TA}}(\theta)}}{1+\norm{\hat{Q}_n^{\mathrm{TA}}(\theta)}+\norm{Q^{\mathrm{TA}}(\theta)}}=o_p(1)$.
\end{enumerate}}

\paragraph*{Condition (a).}

This condition is immediate from~\eqref{eqn:time_average_estimator}.

\paragraph*{Condition (b).}

We establish that the infimum is strictly positive through contradiction. Assume for contradiction that there exists a $\delta >0$ such that $$\inf_{\norm{\theta - \theta_0} > \delta} \norm{Q^{\mathrm{TA}}(\theta)} =0.$$ 
Then by applying the definition of infimum, there exists a sequence $(\theta_n)_{n\in\mathbb{N}}$, where $\norm{\theta_n - \theta_0} > \delta$ for all $n\in\mathbb{N}$, such that $\lim_{n\xrightarrow[]{}\infty} \norm{Q^{\mathrm{TA}}(\theta_n)} = 0$.

By Assumption~\ref{asmpt:statistical_model_discrete_time}, $\Theta$ is a compact subset of $\mathbb{R}^p$, hence $\Theta$ is also sequentially compact, i.e., every sequence of parameters in $\Theta$ has a convergent subsequence converging to a parameter in $\Theta$. Thus, there exists a subsequence $(\theta_{n_j})_{j\in\mathbb{N}}$, where $\norm{\theta_{n_j} - \theta_0} > \delta$ for all $j\in\mathbb{N}$, and a parameter $\theta^{\ast}\in \Theta$ satisfying $\norm{Q^{\mathrm{TA}}(\theta^{\ast})} = 0$, such that $\theta_{n_j}\xrightarrow[]{}\theta^{\ast}$ as $j\xrightarrow[]{}\infty$, which follows because $Q^{\mathrm{TA}}$ is continuous in $\theta$ since the real and imaginary parts of the map $(\theta,\omega)\mapsto \psi_{\theta}(\omega)$ each admit a jointly real analytic extension (Assumption~\ref{asmpt:statistical_model_discrete_time}) and by the definition of the characteristic function and the random Fourier features. By Assumption~\ref{asmpt:injective_discrete_time}, the function $\Phi$ is one-to-one, so the function $Q^{\mathrm{TA}}$ only takes on the value zero at the true parameter $\theta_0$, which means that $\theta^{\ast}=\theta_0$.

However, by the continuity of $\theta \mapsto \norm{\theta-\theta_0}$, we have $\norm{\theta_{n_j}-\theta_0}\xrightarrow[]{}\norm{\theta^{\ast}-\theta_0}$ as $j\xrightarrow[]{}\infty$. Therefore, $\norm{\theta^{\ast}-\theta_0} \geq \delta$ because $\norm{\theta_{n_j} - \theta_0} > \delta$ for all $j\in\mathbb{N}$, so we have $$0 = \norm{\theta_0-\theta_0} = \norm{\theta^{\ast}-\theta_0} \geq \delta > 0,$$ which is a contradiction. %%% Note: Otherwise, if theta^{\ast} was less than delta distance away from theta_0, we would need elements of the sequence theta_{n_j} to be to be within epsilon distance of theta^{\ast} (for all epsilon >0) which would mean that those elements theta_{n_j} would be less than delta distance away from theta_0 which is impossible.

\paragraph*{Condition (c).}

%%% Note: because 1+\norm{\hat{Q}_n^{\mathrm{TA}}(\theta)}+\norm{Q^{\mathrm{TA}}(\theta)} >= 1 and therefore the original term is less than or equal to \sup_{\theta\in\Theta}\norm{\hat{Q}_n^{\mathrm{TA}}(\theta) - Q^{\mathrm{TA}}(\theta)}. 

It suffices to show that 
$\sup_{\theta\in\Theta}\norm{\hat{Q}_n^{\mathrm{TA}}(\theta) - Q^{\mathrm{TA}}(\theta)}=o_p(1)$ because the denominator is always greater than or equal to one. For $\theta\in\Theta$ and $r\in\{0,1,\ldots,s\}$, denote \begin{equation}\label{eqn:random_feature_time_average_approx}F^{(r)}(\theta) = \frac{1}{n-m}\sum_{t=m+1}^n  \varphi\left(\left[G_{t-m}^{(n,r)}(\symbfit{\varepsilon}_{t-m}^{(r)},\theta),\ldots,G_t^{(n,r)}(\symbfit{\varepsilon}_t^{(r)},\theta)\right]^{\top}\right).\end{equation} We have \begin{align*}&\sup_{\theta\in\Theta}\norm{\hat{Q}_n^{\mathrm{TA}}(\theta) - Q^{\mathrm{TA}}(\theta)} \\& \overset{(1)}{\leq} \norm{F^{\mathrm{obs}} - \Phi(\theta_0)} + \sup_{\theta\in\Theta}\norm{\bar{F}^{\mathrm{sim}}(\theta) - \Phi(\theta)}\\ & \overset{(2)}{\leq} \ \norm{F^{\mathrm{obs}} - \Phi(\theta_0)} + \frac{1}{s}\sum_{r=1}^s \sup_{\theta\in\Theta}\norm{F^{(r)}(\theta) - \Phi(\theta)}
\\ & \overset{(3)}{\leq}  \norm{F^{(0)}(\theta_0) - \Phi(\theta_0)} 
+ \frac{1}{s}\sum_{r=1}^s \sup_{\theta\in\Theta}\norm{F^{(r)}(\theta) - \Phi(\theta)}
\\ & \overset{(4)}{\leq} \sup_{\theta\in\Theta}\norm{F^{(0)}(\theta) - \Phi(\theta)} + \frac{1}{s}\sum_{r=1}^s \sup_{\theta\in\Theta}\norm{F^{(r)}(\theta) - \Phi(\theta)},\end{align*} by (1) the triangle inequality and subadditivity of the supremum, (2) the triangle inequality and subadditivity of the supremum, (3) noting that $F^{\mathrm{obs}}=F^{(0)}(\theta_0)$ as in~\eqref{eqn:time_average_random_features_G_general} and~\eqref{eqn:random_feature_time_average_approx}, and (4) upper bounding by the supremum over $\theta\in\Theta$.

%%% Note that, by~\eqref{eqn:convergence_to_Phi_theta_discrete_time} from Theorem~\ref{thm:injective_discrete_time}, the limiting expectation values for all $\theta\in\Theta$ and $r\in\{0,1,\ldots,s\}$ are given by $\Phi(\theta)$.

Observe that, by the triangle inequality, for each $r\in \{0,1,\ldots,s\}$, we have
\begin{equation}\label{eqn:two_terms_to_show_converge}\begin{aligned}
\sup_{\theta\in\Theta}\norm{F^{(r)}(\theta) - \Phi(\theta)}
\leq &\sup_{\theta\in\Theta}\norm{F^{(r)}(\theta) - \mathbb{E}_{\theta}\left[F^{(r)}(\theta)\right]} \\ &+  \sup_{\theta\in\Theta}\norm{\mathbb{E}_{\theta}\left[F^{(r)}(\theta)\right] - \Phi(\theta)}.\end{aligned}
\end{equation} The next steps are to show that the first term is $o_p(1)$ and the second term is $o(1)$.

\textbf{Step 1.} First, we consider the term $\sup_{\theta\in\Theta}\norm{F^{(r)}(\theta) - \mathbb{E}_{\theta}\left[F^{(r)}(\theta)\right]}$. Since $\Theta\subset \mathbb{R}^p$ is compact by Assumption~\ref{asmpt:statistical_model_discrete_time}, it is totally bounded, so 
there exists an $\epsilon$-net, i.e., for every $\epsilon>0$, there exist finitely many points $\theta_1,\ldots,\theta_{N(\epsilon)}\in\Theta$ for some $N(\epsilon) \in \mathbb{N}$, such that $$\Theta\subset\bigcup_{i\in [N(\epsilon)]} B_{\epsilon}(\theta_i),$$ where $B_{\epsilon}(\theta_i)$ is a ball of radius $\epsilon$. Thus, $\forall \theta \in \Theta$, there exists an $i\in [N(\epsilon)]$ such that $\norm{\theta-\theta_i}\leq \epsilon$. By adding and subtracting the same term, the triangle inequality, the definition of the $\epsilon$-net for some $\epsilon>0$ and $N(\epsilon)\in\mathbb{N}$, and the definition of the supremum, we have  
\begin{align*}
&\underset{\theta\in\Theta}{\sup} \ \norm{F^{(r)}(\theta) - \mathbb{E}_{\theta}\left[F^{(r)}(\theta)\right]} 
\\&
\leq 
\underset{i \in [N(\epsilon)]}{\max} \norm{F^{(r)}(\theta_i) - \mathbb{E}_{\theta_i}\left[F^{(r)}(\theta_i)\right]} 
\\ &
+
\underset{ \norm{\theta-\theta'}\leq \epsilon}{\underset{\theta,\theta'\in\Theta}{\sup}}\norm{(F^{(r)}(\theta) - \mathbb{E}_{\theta}\left[F^{(r)}(\theta)\right])-(F^{(r)}(\theta') - \mathbb{E}_{\theta'}\left[F^{(r)}(\theta')\right])}.
\end{align*}

%%% Note: the expectation \mathbb{E} is with respect to the law of the noise input sequences eps_t used to construct the time series via the generative model X_t = G_t(eps_t), whereas \mathbb{E}_{\theta} is expectation with respect to the \theta-dependent distribution P_\theta 
Taking the expectation
$\mathbb{E}(\cdot)$ with respect to the law of the noise inputs, we have
\begin{align*}
&\mathbb{E}\left(\underset{\theta\in\Theta}{\sup} \ \norm{F^{(r)}(\theta) - \mathbb{E}_{\theta}\left[F^{(r)}(\theta)\right]}\right) 
\\&
\overset{(1)}{\leq} 
\mathbb{E}\left(
\underset{i \in [N(\epsilon)]}{\max} \norm{F^{(r)}(\theta_i) - \mathbb{E}_{\theta_i}\left[F^{(r)}(\theta_i)\right]} \right)
\\ &
+
\mathbb{E}\left(
\underset{ \norm{\theta-\theta'}\leq \epsilon}{\underset{\theta,\theta'\in\Theta}{\sup}}\norm{(F^{(r)}(\theta) - \mathbb{E}_{\theta}\left[F^{(r)}(\theta)\right])-(F^{(r)}(\theta') - \mathbb{E}_{\theta'}\left[F^{(r)}(\theta')\right])}\right)
\\&
\overset{(2)}{\leq} 
\sum_{i \in [N(\epsilon)]} \mathbb{E}_{\theta_i}\left(\norm{F^{(r)}(\theta_i) - \mathbb{E}_{\theta_i}\left[F^{(r)}(\theta_i)\right]} \right)
\\ &
+
\mathbb{E}\left(
\underset{ \norm{\theta-\theta'}\leq \epsilon}{\underset{\theta,\theta'\in\Theta}{\sup}}\norm{(F^{(r)}(\theta) - \mathbb{E}_{\theta}\left[F^{(r)}(\theta)\right])-(F^{(r)}(\theta') - \mathbb{E}_{\theta'}\left[F^{(r)}(\theta')\right])}\right)
\\&
\overset{(3)}{\leq} 
\sum_{i \in [N(\epsilon)]} \mathbb{E}_{\theta_i}\left(\norm{F^{(r)}(\theta_i) - \mathbb{E}_{\theta_i}\left[F^{(r)}(\theta_i)\right]} \right)
\\ &
+
\mathbb{E}\left(
\underset{ \norm{\theta-\theta'}\leq \epsilon}{\underset{\theta,\theta'\in\Theta}{\sup}}\norm{F^{(r)}(\theta) - F^{(r)}(\theta') }\right)
+
\underset{ \norm{\theta-\theta'}\leq \epsilon}{\underset{\theta,\theta'\in\Theta}{\sup}}\norm{ \mathbb{E}_{\theta}\left[F^{(r)}(\theta)\right]-\mathbb{E}_{\theta'}\left[F^{(r)}(\theta')\right]}
\\&
\overset{(4)}{\leq}
\sum_{i \in [N(\epsilon)]} \mathbb{E}_{\theta_i}\left(\norm{F^{(r)}(\theta_i) - \mathbb{E}_{\theta_i}\left[F^{(r)}(\theta_i)\right]} \right)
\\& +
2 \ \mathbb{E}\left(
\underset{ \norm{\theta-\theta'}\leq \epsilon}{\underset{\theta,\theta'\in\Theta}{\sup}}\norm{F^{(r)}(\theta) - F^{(r)}(\theta') }\right),
\end{align*} by (1) the previous $\epsilon$-net inequality, monotonicity of expectation, and linearity of expectation, (2) upper bounding the maximum by the sum and linearity of expectation, (3) the triangle inequality and subadditivity of the supremum, and (4) linearity of expectation, Jensen's inequality, and monotonicity of expectation.
%%% Note: last monotonicity of expectation is to bring the sup inside the expectation 

\textbf{Step 1.1.} We begin with the first term. For all $i\in [N(\epsilon)]$, we have \begin{align*}
& \mathbb{E}_{\theta_i}\left(\norm{F^{(r)}(\theta_i) - \mathbb{E}_{\theta_i}\left[F^{(r)}(\theta_i)\right]} \right)
\\& 
\overset{(1)}{\leq}
2 \ C \ \frac{1}{n-m} \ \sum_{j=0}^{\infty}
\left(
\sum_{t=m+1}^n (\eta_{t,j,2,2}^{(n,r)}(\theta_i))^2
\right)^{\frac{1}{2}}
\\&
\overset{(2)}{\leq}
2 \ C \ \frac{1}{n-m} \ \sum_{j=0}^{\infty}
(n-m)^{\frac{1}{2}} \underset{t\in\{m+1,\ldots,n\}}{\max} \ \eta_{t,j,2,2}^{(n,r)}(\theta_i)
\\&
\overset{(3)}{\leq} 
2 \ C \ \frac{1}{\sqrt{n-m}} \ \Psi^{\varphi} \ \sum_{j=0}^{\infty} (j\lor 1)^{-\beta^{\varphi}}
\\&
\overset{(4)}{\leq} 
2 \ C \ \frac{1}{\sqrt{n-m}} \ \Psi^{\varphi}  K, 
\end{align*}
by (1) factoring out $\frac{1}{n-m}$ then applying the second inequality from Lemma~\ref{lma:lp_lln_discrete_time} to the time series of random features and further upper bounding with infinitely many terms in the summation, (2) upper bounding the sum by $n-m$ times the maximum and taking the square root of both terms, (3) by bounds on the physical dependence measures of the original time series by
Assumption~\ref{asmpt:temporal_dependence_discrete_time} and of the time series of random features by Lemma~\ref{lma:temporal_dependence_of_random_features_discrete_time}, and (4) because $\beta^{\varphi}>2$ so the series converges. Hence, for each fixed $\epsilon>0$ and $N(\epsilon)\in\mathbb{N}$, for all $i\in [N(\epsilon)]$, we have
$$\mathbb{E}_{\theta_i}\left(\norm{F^{(r)}(\theta_i) - \mathbb{E}_{\theta_i}\left[F^{(r)}(\theta_i)\right]} \right)=O(n^{-\frac{1}{2}}).$$ 

\textbf{Step 1.2.} 
We now consider the second term. Observe that
\begin{align*}
& \mathbb{E}\left(
\underset{ \norm{\theta-\theta'}\leq \epsilon}{\underset{\theta,\theta'\in\Theta}{\sup}}\norm{F^{(r)}(\theta) - F^{(r)}(\theta') }\right)
\\ & 
\overset{(1)}{\leq}
L_{\varphi} 
\frac{1}{n-m}\sum_{t=m+1}^n \sum_{j=0}^m   \mathbb{E}\left(
\underset{ \norm{\theta-\theta'}\leq \epsilon}{\underset{\theta,\theta'\in\Theta}{\sup}} \norm{
G_{t-j}^{(n,r)}(\symbfit{\varepsilon}_{t-j}^{(r)},\theta) - 
G_{t-j}^{(n,r)}(\symbfit{\varepsilon}_{t-j}^{(r)},\theta')}\right) 
\\& 
\overset{(2)}{\leq} L_{\varphi}  (m+1)  \eta_G,
\end{align*} because (1) the random Fourier features $\varphi$ are $L_{\varphi}$-Lipschitz with Lipschitz constant given by~\eqref{eqn:lipschitz_constant_rff}, the triangle inequality, subadditivity of the supremum, and linearity of expectation, and (2)  by upper bounding the average by the supremum over $t$ and applying the stochastic equicontinuity-type condition from Assumption~\ref{asmpt:stochastic_equicontinuity_discrete_time}, noting that each mapping $(G_t^{(i,r)}(\symbfit{\varepsilon}_j,\theta))_{j\in\mathbb{Z}}$ is a stationary ergodic process for each fixed $t\in \mathbb{N}$, $i\in\mathbb{N}$, $r=0,1,\ldots,s$, which allows us to use this condition even with shifted noise inputs. By Assumption~\ref{asmpt:stochastic_equicontinuity_discrete_time}, $\eta_G \xrightarrow[]{} 0$ as $\epsilon \xrightarrow[]{} 0$. However, taking $\epsilon \xrightarrow[]{} 0$ makes $N(\epsilon)\xrightarrow[]{}\infty$, so we must specify the rates.

Previously, we showed that, for each fixed $\epsilon>0$ and $N(\epsilon)\in\mathbb{N}$, for all $i\in [N(\epsilon)]$, we have
$$\mathbb{E}_{\theta_i}\left(\norm{F^{(r)}(\theta_i) - \mathbb{E}_{\theta_i}\left[F^{(r)}(\theta_i)\right]} \right)=O(n^{-\frac{1}{2}}).$$ Therefore, for any sequence $(\epsilon_n)_{n\in\mathbb{N}}$, $\epsilon_n \xrightarrow[]{} 0$ as $n\xrightarrow[]{}\infty$, such that $N(\epsilon_n)=o(n^{\frac{1}{2}})$, we have 
\begin{align*}&\sum_{i \in [N(\epsilon_n)]} \mathbb{E}_{\theta_i}\left(\norm{F^{(r)}(\theta_i) - \mathbb{E}_{\theta_i}\left[F^{(r)}(\theta_i)\right]} \right)
\\& 
\leq 
N(\epsilon_n) \underset{i \in [N(\epsilon_n)]}{\max} \mathbb{E}_{\theta_i}\left(\norm{F^{(r)}(\theta_i) - \mathbb{E}_{\theta_i}\left[F^{(r)}(\theta_i)\right]} \right)
\\&
=
o(1),
\end{align*}
and $\eta_G = \eta_G(\epsilon_n) =o(1)$ so that the upper bound from (2) is $L_{\varphi} (m+1) \eta_G=o(1)$. See, e.g., Corollary 4.2.11 in \cite{Vershynin} for the explicit covering number bounds.

Putting everything together, by Markov's inequality, we have the desired result $$\sup_{\theta\in\Theta}\norm{F^{(r)}(\theta) - \mathbb{E}_{\theta}\left[F^{(r)}(\theta)\right]}=o_p(1).$$

\textbf{Step 2.} Second, we consider the term $\sup_{\theta\in\Theta}\norm{\mathbb{E}_{\theta}\left[F^{(r)}(\theta)\right] - \Phi(\theta)}$. For each $r\in\{0,1,\ldots,s\}$, \begin{align*}&\sup_{\theta\in\Theta}\norm{\mathbb{E}_{\theta}\left[F^{(r)}(\theta)\right] - \Phi(\theta)} \\
&\overset{(1)}{=} \underset{\theta\in\Theta}{\sup} \ \norm{\frac{1}{n-m}\sum_{t=m+1}^n \mathbb{E}_{\theta}\left[ \varphi\left(\left[G_{t-m}^{(n,r)}(\symbfit{\varepsilon}_{t-m}^{(r)},\theta),\ldots,G_t^{(n,r)}(\symbfit{\varepsilon}_t^{(r)},\theta)\right]^{\top}\right)\right] - \Phi(\theta)}
\\&\overset{(2)}{\leq} \sup_{\theta\in\Theta} 
\norm{
\begin{aligned}
&\frac{1}{n-m}\sum_{t=m+1}^n 
\mathbb{E}_{\theta}\left[
\varphi\left(
\left[G_{t-m}^{(n,r)}(\symbfit{\varepsilon}_{t-m}^{(r)},\theta),\ldots,
G_t^{(n,r)}(\symbfit{\varepsilon}_t^{(r)},\theta)\right]^{\top}
\right)
\right] \\
&-
\frac{1}{n-m}\sum_{t=m+1}^n 
\mathbb{E}_{\theta}\left[
\varphi\left(
\left[G_{t-m}^{(r)}(\symbfit{\varepsilon}_{t-m}^{(r)},\theta),\ldots,
G_t^{(r)}(\symbfit{\varepsilon}_t^{(r)},\theta)\right]^{\top}
\right)
\right]
\end{aligned}
} 
\\ &
+ \sup_{\theta\in\Theta}
\norm{
\frac{1}{n-m}\sum_{t=m+1}^n 
\mathbb{E}_{\theta}\left[
\varphi\left(
\left[G_{t-m}^{(r)}(\symbfit{\varepsilon}_{t-m}^{(r)},\theta),\ldots,
G_t^{(r)}(\symbfit{\varepsilon}_t^{(r)},\theta)\right]^{\top}
\right)
\right] - \Phi(\theta)
} 
\\ 
& \overset{(3)}{=} o(1),
\end{align*}
by (1) the definition of $F^{(r)}(\theta)$ from~\eqref{eqn:random_feature_time_average_approx} and linearity of expectation, (2) adding and subtracting the same term, the triangle inequality, and subadditivity of the supremum, (3) because both terms are $o(1)$. Specifically, the second term is $o(1)$ by upper bounding the $\ell_2$ norm $\norm{\cdot}$ by the $\ell_1$ norm (i.e., over the $k$ random features) using the monotonicity of norms, subadditivity of the supremum, linearity, and then applying the weak convergence from Assumption~\ref{asmpt:weak_convergence_of_time_avg_of_distributions_discrete_time}. The first term is $o(1)$ because the approximations can be asymptotically replaced by the limiting mappings, because for all $r \in \{0,1,\ldots,s\}$, we have \begin{align*} & \underset{\theta\in\Theta}{\sup} \ 
\norm{
\begin{aligned}
&\frac{1}{n-m}\sum_{t=m+1}^n 
\mathbb{E}_{\theta}\left[
\varphi\left(
\left[G_{t-m}^{(n,r)}(\symbfit{\varepsilon}_{t-m}^{(r)},\theta),\ldots,
G_t^{(n,r)}(\symbfit{\varepsilon}_t^{(r)},\theta)\right]^{\top}
\right)
\right] \\
&-
\frac{1}{n-m}\sum_{t=m+1}^n 
\mathbb{E}_{\theta}\left[
\varphi\left(
\left[G_{t-m}^{(r)}(\symbfit{\varepsilon}_{t-m}^{(r)},\theta),\ldots,
G_t^{(r)}(\symbfit{\varepsilon}_t^{(r)},\theta)\right]^{\top}
\right)
\right]
\end{aligned}
}
\\
&\overset{(1)}{\leq} 
\underset{\theta\in\Theta}{\sup} \left(\frac{1}{n-m}\sum_{t=m+1}^n   \norm{
\begin{aligned}
& 
\varphi\left(
\left[G_{t-m}^{(n,r)}(\symbfit{\varepsilon}_{t-m}^{(r)},\theta),\ldots,
G_t^{(n,r)}(\symbfit{\varepsilon}_t^{(r)},\theta)\right]^{\top}
\right)
 \\
& - 
\varphi\left(
\left[G_{t-m}^{(r)}(\symbfit{\varepsilon}_{t-m}^{(r)},\theta),\ldots,
G_t^{(r)}(\symbfit{\varepsilon}_t^{(r)},\theta)\right]^{\top}
\right)
\end{aligned}
}_{\mathcal{L}^1(\theta)}\right)
\\
& \overset{(2)}{\leq}  L_{\varphi} \ \underset{\theta\in\Theta}{\sup} \left(
\frac{1}{n-m}\sum_{t=m+1}^n    \norm{
\begin{aligned}
&  
\mathrm{Vec}\left(G_{t-m}^{(n,r)}(\symbfit{\varepsilon}_{t-m}^{(r)},\theta),\ldots,
G_t^{(n,r)}(\symbfit{\varepsilon}_t^{(r)},\theta)\right)
 \\
&- \mathrm{Vec}
\left(G_{t-m}^{(r)}(\symbfit{\varepsilon}_{t-m}^{(r)},\theta),\ldots,
G_t^{(r)}(\symbfit{\varepsilon}_t^{(r)},\theta)\right)
\end{aligned}
}_{\mathcal{L}^1(\theta)}\right)
\\
& \overset{(3)}{\leq}  L_{\varphi} \ \underset{\theta\in\Theta}{\sup} \left(
\frac{1}{n-m}\sum_{t=m+1}^n \sum_{j=0}^m    \norm{
G_{t-j}^{(n,r)}(\symbfit{\varepsilon}_{t-j}^{(r)},\theta) - 
G_{t-j}^{(r)}(\symbfit{\varepsilon}_{t-j}^{(r)},\theta)
}_{\mathcal{L}^1(\theta)}\right)
\\
& \overset{(4)}{\leq}
 L_{\varphi} 
\frac{1}{1-m/n} (m+1) \ \underset{\theta\in\Theta}{\sup} \left(\frac{1}{n}\sum_{t=1}^n    \norm{G_{t}^{(n,r)}(\symbfit{\varepsilon}_{0}^{(r)},\theta) - 
G_{t}^{(r)}(\symbfit{\varepsilon}_{0}^{(r)},\theta)
}_{\mathcal{L}^1(\theta)}\right)
\\
& \overset{(5)}{=}  o(1),
\end{align*}
where (1) holds by linearity of summation, the triangle inequality, linearity of expectation, and Jensen's inequality, (2) holds because the random Fourier features $\varphi$ are $L_{\varphi}$-Lipschitz with Lipschitz constant given by~\eqref{eqn:lipschitz_constant_rff}, (3) holds by the triangle inequality and Minkowski inequality, (4) follows by upper bounding the original quantity by $m+1$ times the full sum (i.e., by including more non-negative terms in the summation), factoring out $1/n$, and noting that each $(G_t^{(i,r)}(\symbfit{\varepsilon}_j,\theta))_{j\in\mathbb{Z}}$, $(G_t^{(r)}(\symbfit{\varepsilon}_j,\theta))_{j\in\mathbb{Z}}$ is a stationary ergodic process for each fixed $t\in \mathbb{N}$, $i\in \mathbb{N}$, and $r=0,1,\ldots,s$, which allows us to shift the noise inputs, and (5) is from applying the convergence condition from Assumption~\ref{asmpt:limiting_mapping_discrete_time}.

\textbf{Step 3.} We have shown that for each $r\in \{0,1,\ldots,s\}$, the terms from~\eqref{eqn:two_terms_to_show_converge} are $o_p(1)$ and $o(1)$, respectively. Putting it all together, we have the desired result that $$\sup_{\theta\in\Theta}\norm{\hat{Q}_n^{\mathrm{TA}}(\theta) - Q^{\mathrm{TA}}(\theta)}=o_p(1).$$

\hfill$\square$

\subsection{Proof of Theorem~\ref{thm:injective_continuous_time}}\label{subsection:proof_of_injective_continuous_time}

\paragraph*{Step 1: Limiting local expectations.} We claim that under
Assumptions~\ref{asmpt:algorithmic_dynamic_model_continuous_time}-\ref{asmpt:reparametrization_continuous_time}, $\forall \tilde{\theta}_u\in\tilde{\Theta}_u$, the limiting local expectation values of the $k$ random features of the time series are given by \begin{equation*}   \tilde{\Phi}_u(\tilde{\theta}_u)=\int_{\mathbb{R}^{(m+1)\times d}}\varphi(x) \, d\tilde{P}_{\tilde{\theta}_u,u}(x), \quad u\in [0,1].\end{equation*}  That is, for all $\theta \in \Theta$, all $u\in [0,1]$, and all $r \in \{0,1,\ldots,s\}$, we have, as $n\xrightarrow[]{}\infty$, 
\begin{equation}\label{eqn:convergence_to_tilde_Phi_theta_u_continuous_time} \norm{\mathbb{E}_{\theta}\left[ \varphi\left(\left[G_{u}^{(n,r)}(\symbfit{\varepsilon}_{-m}^{(r)},\theta),\ldots,G_{u}^{(n,r)}(\symbfit{\varepsilon}_0^{(r)},\theta)\right]^{\top}\right)\right]-\tilde{\Phi}_u(\tilde{\theta}_u)}=o(1).\end{equation} Note that the time series has this causal representation by Assumption~\ref{asmpt:algorithmic_dynamic_model_continuous_time}, and that we only observe the time series at rescaled times $t/n$ for $t=1,\ldots,n$.

The desired result from~\eqref{eqn:convergence_to_tilde_Phi_theta_u_continuous_time} follows because, for all $\theta \in \Theta$, all $u\in [0,1]$, and all $r \in \{0,1,\ldots,s\}$, we have 
\begin{align*}&\norm{\mathbb{E}_{\theta}\left[ \varphi\left(\left[G_{u}^{(n,r)}(\symbfit{\varepsilon}_{-m}^{(r)},\theta),\ldots,G_{u}^{(n,r)}(\symbfit{\varepsilon}_0^{(r)},\theta)\right]^{\top}\right)\right] - \tilde{\Phi}_u(\tilde{\theta}_u)}
\\ &
\overset{(1)}{\leq} 
\norm{
\begin{aligned}
& 
\mathbb{E}_{\theta}\left[ \varphi\left(\left[G_{u}^{(n,r)}(\symbfit{\varepsilon}_{-m}^{(r)},\theta),\ldots,G_{u}^{(n,r)}(\symbfit{\varepsilon}_0^{(r)},\theta)\right]^{\top}\right)\right] \\
&-
\mathbb{E}_{\theta}\left[ \varphi\left(\left[G_{u}(\symbfit{\varepsilon}_{-m}^{(r)},\theta),\ldots,G_{u}(\symbfit{\varepsilon}_0^{(r)},\theta)\right]^{\top}\right)\right]
\end{aligned}
}
\\ &
+ \norm{\mathbb{E}_{\theta}\left[ \varphi\left(\left[G_{u}(\symbfit{\varepsilon}_{-m}^{(r)},\theta),\ldots,G_{u}(\symbfit{\varepsilon}_0^{(r)},\theta)\right]^{\top}\right)\right] - \tilde{\Phi}_u(\tilde{\theta}_u)}
\\ &
\overset{(2)}{=}
o(1),
\end{align*}
where (1) is from adding and subtracting the corresponding term with the limiting mapping and by applying the triangle inequality, and (2) holds because the second term is equal to zero by using the equivalent causal representation from Assumption~\ref{asmpt:reparametrization_continuous_time} and because the first term is $o(1)$. The approximations can be asymptotically replaced by the limiting mappings because for all $\theta \in \Theta$, all $u\in [0,1]$, and all $r \in \{0,1,\ldots,s\}$, we have
\begin{align*} &
\norm{
\begin{aligned}
& 
\mathbb{E}_{\theta}\left[ \varphi\left(\left[G_{u}^{(n,r)}(\symbfit{\varepsilon}_{-m}^{(r)},\theta),\ldots,G_{u}^{(n,r)}(\symbfit{\varepsilon}_0^{(r)},\theta)\right]^{\top}\right)\right] \\
&-
\mathbb{E}_{\theta}\left[ \varphi\left(\left[G_{u}(\symbfit{\varepsilon}_{-m}^{(r)},\theta),\ldots,G_{u}(\symbfit{\varepsilon}_0^{(r)},\theta)\right]^{\top}\right)\right]
\end{aligned}
}
\\
&\overset{(1)}{\leq} 
\norm{
\begin{aligned}
& 
 \varphi\left(\left[G_{u}^{(n,r)}(\symbfit{\varepsilon}_{-m}^{(r)},\theta),\ldots,G_{u}^{(n,r)}(\symbfit{\varepsilon}_0^{(r)},\theta)\right]^{\top}\right) \\
&-
\varphi\left(\left[G_{u}(\symbfit{\varepsilon}_{-m}^{(r)},\theta),\ldots,G_{u}(\symbfit{\varepsilon}_0^{(r)},\theta)\right]^{\top}\right)
\end{aligned}
}_{\mathcal{L}^1(\theta)}
\\
& \overset{(2)}{\leq}  L_{\varphi} 
\norm{
\begin{aligned}
& 
 \mathrm{Vec}\left(G_{u}^{(n,r)}(\symbfit{\varepsilon}_{-m}^{(r)},\theta),\ldots,G_{u}^{(n,r)}(\symbfit{\varepsilon}_0^{(r)},\theta)\right) \\
&-
\mathrm{Vec}\left(G_{u}(\symbfit{\varepsilon}_{-m}^{(r)},\theta),\ldots,G_{u}(\symbfit{\varepsilon}_0^{(r)},\theta)\right)
\end{aligned}
}_{\mathcal{L}^1(\theta)}
\\
& \overset{(3)}{\leq}  L_{\varphi} 
 \sum_{j=0}^m   \norm{
G_{u}^{(n,r)}(\symbfit{\varepsilon}_{-j}^{(r)},\theta) - 
G_{u}(\symbfit{\varepsilon}_{-j}^{(r)},\theta)
}_{\mathcal{L}^1(\theta)}
\\
& \overset{(4)}{=}  o(1),
\end{align*} where (1) holds by linearity of expectation and Jensen's inequality, (2) holds because the random Fourier features $\varphi$ are $L_{\varphi}$-Lipschitz with Lipschitz constant given by~\eqref{eqn:lipschitz_constant_rff}, (3) holds by the triangle inequality and Minkowski inequality, and (4) follows from the convergence condition from Assumption~\ref{asmpt:limiting_mapping_continuous_time}, noting that each mapping $(G_u^{(i,r)}(\symbfit{\varepsilon}_j,\theta))_{j\in\mathbb{Z}}$, $(G_u(\symbfit{\varepsilon}_j,\theta))_{j\in\mathbb{Z}}$ is a stationary ergodic process for each fixed $u\in [0,1]$, $i\in\mathbb{N}$, $r\in\{0,1,\ldots,s\}$, which allows us to use the convergence condition regardless of the shifts in the noise inputs.

\paragraph*{Step 2: Injectivity.} The two key steps are: (1) the use of the finite witness theorem (Theorem~\ref{thm:finite_witness}), and (2) the extension to almost every rescaled time $u\in [0,1]$. Denote $$\mathbb{W}=\mathbb{R}^{d(m+1)}\times (-\pi,\pi).$$ 

\textbf{Step 2.1.} We first show that the assumptions of Theorem~\ref{thm:finite_witness} are satisfied at each fixed $u$.

For each $u \in [0,1]$, the Hausdorff dimension of $\tilde{\Theta}_u$ is less than or equal to $p$ because (1)  $\Theta$ has Hausdorff dimension $p$ as it is a compact subset of $\mathbb{R}^p$ with non-empty interior by Assumption~\ref{asmpt:statistical_model_continuous_time}, and (2) $\tilde{\Theta}_u=g_u(\Theta)$ and $g_u$ is Lipschitz by Assumption~\ref{asmpt:reparametrization_continuous_time} so $\mathrm{dim}_H(\tilde{\Theta}_u)=\mathrm{dim}_H(g_u(\Theta))\leq \mathrm{dim}_H(\Theta)=p$. Hence, for each $u\in [0,1]$, its Hausdorff dimension is $\mathrm{dim}_H(\tilde{\Theta}_u)\leq p$. Therefore, it suffices to use $k=2p+1$ random Fourier features with Theorem~\ref{thm:finite_witness} (any $k\geq 2p+1$ works, see the proof of Proposition 3.6 in \cite{amir_et_al_finite_witness_thm}).

%%% Hausdorff dim <= box counting dim (when it exists, of course), see page 15 of https://pi.math.cornell.edu/~fractals/6/slides/Lapidus.pdf

By Assumption~\ref{asmpt:statistical_model_continuous_time}, $\Theta$ is a semialgebraic subset of $\mathbb{R}^p$. By Assumption~\ref{asmpt:reparametrization_continuous_time}, $g_u$ is a semialgebraic map. Therefore, by the Tarski--Seidenberg theorem (Theorem~\ref{thm:tarski_seidenberg_theorem}), $\tilde{\Theta}_u=g_u(\Theta)$ is a semialgebraic set (i.e., by using the graph of $g_u$ and repeatedly applying Theorem~\ref{thm:tarski_seidenberg_theorem}), hence it is $\sigma$-subanalytic; see Figure~\ref{fig:diagram_of_subanalytic_sets_and_functions}.

By construction, the random parameters $\vartheta_i$, $i=1,\ldots,k$ from~\eqref{eqn:parameters_for_the_random_features} for the random Fourier features $\varphi_i$, $i=1,\ldots,k$ from~\eqref{eqn:random_Fourier_features_order_m} are each random variables taking values in $\mathbb{W}$, which is semialgebraic and hence $\sigma$-subanalytic.

By Assumption~\ref{asmpt:statistical_model_continuous_time}, the real and imaginary parts of the map $(\tilde{\theta}_u,\omega)\mapsto \tilde{\psi}_{\tilde{\theta}_u,u}(\omega)$ each admit jointly real analytic extensions to $\tilde{O}_u \times \mathbb{R}^{(m+1)\times d}$, so each 
$(\tilde{\theta}_u,\vartheta_i)\mapsto \int_{\mathbb{R}^{(m+1)\times d}}\varphi_i(x) \, d\tilde{P}_{\tilde{\theta}_u,u}(x)$, $i=1,\ldots,k$, admits a jointly real analytic extension to $\tilde{O}_u\times \mathbb{W}$ by using the definition of the characteristic function and each random Fourier feature. Note that the graph of this extension is $\sigma$-subanalytic, and that $\tilde{\Theta}_u$ is a $\sigma$-subanalytic set as discussed previously. Therefore, since $\sigma$-subanalytic sets are closed under finite intersections and Cartesian products by Theorem~\ref{thm:sigma_subanalytic_properties}, we have that the intersection of the graph of this extension with $\tilde{\Theta}_u \times \mathbb{W}\times\mathbb{R}$ is $\sigma$-subanalytic.

$\mathbb{W}$ is open and connected, and $(\tilde{\theta}_u,\vartheta_i)\mapsto \int_{\mathbb{R}^{(m+1)\times d}}\varphi_i(x) \, d\tilde{P}_{\tilde{\theta}_u,u}(x)$ is analytic as a function of $\vartheta_i$ for all fixed $\tilde{\theta}_u\in\tilde{\Theta}_u$. Hence, the condition~\eqref{cond_dimension_defficiency_for_separation} from the finite witness theorem (Theorem~\ref{thm:finite_witness}) is satisfied.

Let $x_1,\ldots,x_{m+1}$ be vectors in $\mathbb{R}^d$, and define $x=(x_1,\ldots,x_{m+1})^{\top}\in\mathbb{R}^{(m+1)\times d}$. Letting $a=(a_1^{\top},\ldots,a_{m+1}^{\top})^{\top}\in\mathbb{R}^{d(m+1)}$, $b\in (-\pi,\pi)$, and $c=(a,b)\in\mathbb{W}$, define $$\phi_c(x)=\cos\left(\sum_{j=1}^{m+1} a_j \cdot x_{j} + b\right).$$ For each fixed $u\in [0,1]$, define \begin{equation*}\mathcal{N}=\left\{(\tilde{\theta}_{u,1},\tilde{\theta}_{u,2}) \in \tilde{\Theta}_u^2 \Biggm| \begin{aligned}&\int_{\mathbb{R}^{(m+1)\times d}}\phi_c(x)  d\tilde{P}_{\tilde{\theta}_{u,1},u}(x) \\&=\int_{\mathbb{R}^{(m+1)\times d}}\phi_c(x) d\tilde{P}_{\tilde{\theta}_{u,2},u}(x), \forall c\in\mathbb{W} \end{aligned}\right\},\end{equation*} where we have suppressed the dependence of $\mathcal{N}$ on $u$. For each fixed $u\in [0,1]$, by the finite witness theorem (Theorem~\ref{thm:finite_witness}), for generic parameters $\left( \vartheta_1,\ldots,\vartheta_{k}\right) \in \mathbb{W}^k$ of the $k=2p+1$ random Fourier features $\left( \varphi_1,\ldots,\varphi_{k}\right)$, we have the set equality \begin{equation*}\mathcal{N}=\left\{(\tilde{\theta}_{u,1},\tilde{\theta}_{u,2}) \in \tilde{\Theta}_u^2 \Biggm|\begin{aligned}& \int_{\mathbb{R}^{(m+1)\times d}}\varphi_i(x) d\tilde{P}_{\tilde{\theta}_{u,1},u}(x) \\& =\int_{\mathbb{R}^{(m+1)\times d}}\varphi_i(x) d\tilde{P}_{\tilde{\theta}_{u,2},u}(x),  \forall i\in [k]\end{aligned} \right\}.\end{equation*}  By the discussion above Theorem~\ref{thm:finite_witness}, generic here means $\Pi$-almost every $\vartheta$, where $\Pi$ is the joint distribution from~\eqref{eqn:parameters_for_the_random_features} which is absolutely continuous w.r.t.\ the Lebesgue measure.

Denoting the $\mathbb{R}^{(m+1)\times d}$-valued random matrix by $\Omega_i=(\Omega_{i,1},\ldots,\Omega_{i,m+1})^{\top}$, we have \begin{align*}&\int_{\mathbb{R}^{(m+1)\times d}}\varphi_i(x) \, d\tilde{P}_{\tilde{\theta}_u,u}(x)
\\& 
=\int_{\mathbb{R}^{(m+1)\times d}}\cos\left(\sum_{j=1}^{m+1} \Omega_{i,j} \cdot x_{j} + \alpha_i\right)\, d\tilde{P}_{\tilde{\theta}_u,u}(x)
\\&
=
\cos(\alpha_i)
\int_{\mathbb{R}^{(m+1)\times d}}\cos\left(\sum_{j=1}^{m+1} \Omega_{i,j} \cdot x_{j}\right)\, d\tilde{P}_{\tilde{\theta}_u,u}(x) 
\\&
- 
\sin(\alpha_i)
\int_{\mathbb{R}^{(m+1)\times d}}\sin\left(\sum_{j=1}^{m+1} \Omega_{i,j} \cdot x_{j}\right)\, d\tilde{P}_{\tilde{\theta}_u,u}(x) 
\\&
= \cos(\alpha_i) \  \mathrm{Re}( \tilde{\psi}_{\tilde{\theta}_u,u}(\Omega_i)) - \sin(\alpha_i) \ \mathrm{Im}(  \tilde{\psi}_{\tilde{\theta}_u,u}(\Omega_i))
\\&
= M_{\vartheta_i}(\tilde{\theta}_u,u),
\end{align*}
where $\tilde{\psi}_{\tilde{\theta}_u,u}$ is the characteristic function. If $(\tilde{\theta}_{u,1},\tilde{\theta}_{u,2})\in\mathcal{N}$, then $ M_{\vartheta_i}(\tilde{\theta}_{u,1},u)= M_{\vartheta_i}(\tilde{\theta}_{u,2},u)$ for every possible value of $\vartheta_i=(\mathrm{Vec}(\Omega_i),\alpha_i)\in\mathbb{W}$. In particular, the expectation values of the random Fourier features correspond to the real and imaginary parts of the characteristic function by taking $\alpha_i=0$ and $\alpha_i=-\pi/2$, respectively, so by the uniqueness theorem (Theorem~\ref{thm:uniqueness_characteristic_function}) the corresponding distributions must be equal, $\tilde{P}_{\tilde{\theta}_{u,1},u}= \tilde{P}_{\tilde{\theta}_{u,2},u}$, if $(\tilde{\theta}_{u,1},\tilde{\theta}_{u,2})\in\mathcal{N}$. However, by Assumption~\ref{asmpt:statistical_model_continuous_time}, the map $\tilde{\theta}_u\mapsto \tilde{P}_{\tilde{\theta}_u,u}$ is injective, i.e., if $\tilde{\theta}_{u,1} \neq \tilde{\theta}_{u,2}$ then $\tilde{P}_{\tilde{\theta}_{u,1},u}\neq \tilde{P}_{\tilde{\theta}_{u,2},u}$, so we must have \begin{align*}&(\tilde{\theta}_{u,1},\tilde{\theta}_{u,2})\in\mathcal{N} \\ & \implies (\tilde{\theta}_{u,1},\tilde{\theta}_{u,2})\in \{(\tilde{\theta}_{u,1},\tilde{\theta}_{u,2})\in\tilde{\Theta}_u^2:\tilde{\theta}_{u,1}=\tilde{\theta}_{u,2}\}=\{(\tilde{\theta}_u,\tilde{\theta}_u):\tilde{\theta}_u\in\tilde{\Theta}_u\},\end{align*} so $\mathcal{N}\subseteq \{(\tilde{\theta}_u,\tilde{\theta}_u):\tilde{\theta}_u\in\tilde{\Theta}_u\}$. Conversely, we have $\{(\tilde{\theta}_u,\tilde{\theta}_u):\tilde{\theta}_u\in\tilde{\Theta}_u\} \subseteq \mathcal{N}$ trivially. Therefore,
\begin{equation}\label{eqn:conclusion_mathcal_N_continuous_time}\mathcal{N}=\{(\tilde{\theta}_u,\tilde{\theta}_u):\tilde{\theta}_u\in\tilde{\Theta}_u\}.\end{equation}

Hence, for any particular rescaled time $u\in [0,1]$, for almost any random parameters $(\vartheta_1,\ldots,\vartheta_k)$ for the $k=2p+1$ random Fourier features $(\varphi_1,\ldots,\varphi_k)$, we have that for any two distinct parameter values, $\tilde{\theta}_{u,1} \neq \tilde{\theta}_{u,2}$, the corresponding expectations of the $k=2p+1$ random Fourier features are not equal, $\tilde{\Phi}_u(\tilde{\theta}_{u,1})\neq \tilde{\Phi}_u(\tilde{\theta}_{u,2})$.

\textbf{Step 2.2.} Next, we extend the result to almost every
rescaled time $u\in[0,1]$.

For each $\ell\in\mathbb N$, define \begin{equation}\label{eqn:function_being_maximized}
b_{\ell}(u,\vartheta)
=
\max_{(\tilde{\theta}_{u,1},\tilde{\theta}_{u,2}) \in\tilde{\Theta}_u^2}
-(\norm{\tilde{\Phi}_u
(\tilde{\theta}_{u,1})
-
\tilde{\Phi}_u
(\tilde{\theta}_{u,2})} +
\max(
\frac1\ell-
\|\tilde{\theta}_{u,1}-\tilde{\theta}_{u,2}\|
, 0 )).
\end{equation}

The function $\bar{\psi}$ from Assumption~\ref{asmpt:measurability} induces an extension of $\tilde{\Phi}_u$ to $\mathbb{R}^{\tilde{p}}$ that is jointly Borel measurable and continuous in $\tilde{\theta}_u$. We continue to denote this extension by $\tilde{\Phi}_u$ to simplify the presentation. By Assumptions~\ref{asmpt:statistical_model_continuous_time} and~\ref{asmpt:measurability}, the set-valued map $$
(u,\vartheta)
\mapsto
\tilde{\Theta}_u^2,$$ % vartheta does nothing here, but it's fine
% cartesian product of compact sets is compact, yes
is measurable and has values that are nonempty compact sets, and the function being maximized in~\eqref{eqn:function_being_maximized},
\begin{equation}\label{eqn:map_param_to_objective}(\tilde{\theta}_{u,1},\tilde{\theta}_{u,2})\mapsto -\left(\norm{\tilde{\Phi}_u
(\tilde{\theta}_{u,1})
-
\tilde{\Phi}_u
(\tilde{\theta}_{u,2})} +
\max\left(
\frac1\ell-
\|\tilde{\theta}_{u,1}-\tilde{\theta}_{u,2}\|
, 0 \right)
\right),\end{equation} is jointly Borel measurable and continuous
in
$(\tilde{\theta}_{u,1},\tilde{\theta}_{u,2})$
for each fixed $(u,\vartheta)$. Hence, by the
measurable maximum theorem (Theorem~\ref{thm:mble_max_theorem}), the function $b_{\ell}$ from~\eqref{eqn:function_being_maximized} is Borel measurable. Then the set \begin{equation}\label{eqn:set_B}B
=
\bigcup_{\ell=1}^{\infty}
\left\{
(u,\vartheta):
b_{\ell}(u,\vartheta)=0
\right\},\end{equation} is a Borel set. This fact will be used later on, so that we may apply Tonelli's theorem.

Next, we show that $(u,\vartheta)\in B$ if and only if
$\tilde{\Phi}_u$ is not one-to-one. 

First, suppose that $\tilde{\Phi}_u$ is not one-to-one. Then there exist distinct parameters $\tilde{\theta}_{u,1},\tilde{\theta}_{u,2}\in\tilde{\Theta}_u$ and some $\ell\in\mathbb{N}$ such that $\norm{\tilde{\theta}_{u,1}-\tilde{\theta}_{u,2}}\geq \frac{1}{\ell}$ and $\tilde{\Phi}_u
(\tilde{\theta}_{u,1})
=
\tilde{\Phi}_u
(\tilde{\theta}_{u,2})$. Therefore, we have that $b_{\ell}(u,\vartheta)=0$ for some $\ell\in\mathbb{N}$, so $(u,\vartheta)\in B$ by their definitions~\eqref{eqn:function_being_maximized} and~\eqref{eqn:set_B}.

Second, conversely, suppose that $(u,\vartheta)\in B$. Then there exists some $\ell\in\mathbb{N}$ such that $b_{\ell}(u,\vartheta)=0$ by definition. By the compactness of $\tilde{\Theta}_u$ from Assumption~\ref{asmpt:statistical_model_continuous_time} and the continuity of the map from~\eqref{eqn:map_param_to_objective}, the extreme value theorem implies that the map from~\eqref{eqn:map_param_to_objective} must attain its maximum at some $(\tilde{\theta}_{u,1},\tilde{\theta}_{u,2})\in\tilde{\Theta}_u^2$. Therefore, both terms from~\eqref{eqn:function_being_maximized}, $$\norm{\tilde{\Phi}_u
(\tilde{\theta}_{u,1})
-
\tilde{\Phi}_u
(\tilde{\theta}_{u,2})},\quad \max\left(
\frac1\ell-
\|\tilde{\theta}_{u,1}-\tilde{\theta}_{u,2}\|
, 0 \right),$$ must be equal to zero for some $\tilde{\theta}_{u,1},\tilde{\theta}_{u,2}\in\tilde{\Theta}_u$. That is, there exist distinct parameter values $\tilde{\theta}_{u,1},\tilde{\theta}_{u,2}\in\tilde{\Theta}_u$ and some $\ell \in \mathbb{N}$ such that $\norm{\tilde{\theta}_{u,1}-\tilde{\theta}_{u,2}}\geq \frac{1}{\ell}$ and $\tilde{\Phi}_u
(\tilde{\theta}_{u,1})
=
\tilde{\Phi}_u
(\tilde{\theta}_{u,2})$. Therefore, $\tilde{\Phi}_u$ is not one-to-one.

Let $\Pi$ denote the joint distribution of the random parameters $\vartheta=(\vartheta_1,\ldots,\vartheta_k)\in
\mathbb{W}^k$ from~\eqref{eqn:parameters_for_the_random_features}. Note that $\Pi$ is absolutely continuous
with respect to Lebesgue measure on $\mathbb{W}^k$. We just established that $(u,\vartheta)\in B$ if and only if
$\tilde{\Phi}_u$ is not one-to-one. By the conclusion of Step 2.1, namely the set equality from~\eqref{eqn:conclusion_mathcal_N_continuous_time} due to the finite witness theorem, and the absolute continuity of $\Pi$ w.r.t.\ the Lebesgue measure, for every fixed $u\in [0,1]$, we have $$\Pi\left(\left\{\vartheta\in\mathbb{W}^k:
(u,\vartheta)\in B
\right\}\right) =0.$$ Therefore, we have
$$\int_{[0,1]}\Pi\left(\left\{\vartheta\in\mathbb{W}^k:
(u,\vartheta)\in B
\right\}\right) d\lambda(u)=0.$$

We previously showed that $B$ from~\eqref{eqn:set_B} is a Borel set, hence the indicator function $$1_B(u,\vartheta)=\begin{cases}
    1, \quad (u,\vartheta)\in B, \\0, \quad (u,\vartheta)\not\in B,
\end{cases}$$ is jointly Borel measurable. Thus, we may apply Tonelli's theorem. Observe that \begin{align*}\label{eqn:conclusion_tonelli}
&\int_{\mathbb{W}^k}
\lambda\left(
\left\{
u\in[0,1]:
(u,\vartheta)\in B
\right\}
\right) d\Pi(\vartheta)\\
&
=\int_{\mathbb{W}^k} \left(\int_{[0,1]}
1_B(u,\vartheta) d\lambda(u)\right)  d\Pi(\vartheta)\\
&
=\int_{[0,1]}\left(\int_{\mathbb{W}^k} 1_B(u,\vartheta) d\Pi(\vartheta) \right) d\lambda(u)\\
&
=
\int_{[0,1]}\Pi\left(\left\{\vartheta\in\mathbb{W}^k:
(u,\vartheta)\in B
\right\}\right) d\lambda(u)
\\&
=
0.
\end{align*} The integrand $\lambda\left(
\left\{
u\in[0,1]:
(u,\vartheta)\in B
\right\}
\right)$ is nonnegative, so for $\Pi$-almost every $\vartheta$, we have  
$$\lambda\left(
\left\{
u\in[0,1]:
(u,\vartheta)\in B
\right\}
\right)=0.$$
However, we previously showed that $(u,\vartheta)\in B$ if and only if
$\tilde{\Phi}_u$ is not one-to-one, so for $\Pi$-almost every $\vartheta$, we have
$$\lambda\left(
\left\{
u\in[0,1]:
\tilde{\Phi}_u \text{ is not one-to-one }
\right\}
\right)=0.$$
Putting it all together, for $\Pi$-almost any random parameters $\vartheta=(\vartheta_1,\ldots,\vartheta_k)$ for the $k=2p+1$ random Fourier features $(\varphi_1,\ldots,\varphi_k)$, there exists a subset of rescaled times with full measure with the property that $\tilde{\Phi}_u$ is one-to-one, i.e., for $\Pi$-almost every $\vartheta$, 
\begin{equation}\label{eqn:set_of_rescaled_times_one_to_one}\exists \ \mathcal{U}_{\vartheta}=\left\{
u\in[0,1]:
\tilde{\Phi}_u \text{ is one-to-one }
\right\}\subseteq [0,1] \text{ such that } \lambda\left(\mathcal{U}_{\vartheta}\right)=1.\end{equation}

\hfill$\square$

\subsection{Proof of Theorem~\ref{thm:consistency_rolling_window_estimator}}\label{subsection:proof_of_consistency_rolling_window_estimator}

\paragraph*{Setup.} We will use Corollary 3.2 of~\cite{pakes_pollard_simulation_asymptotics} to conclude that $\hat{\theta}^{\mathrm{RW}}$ is consistent. The sample discrepancy function, from~\eqref{eqn:rolling_window_sample_discrepancy}, is given by
\begin{equation*} \begin{aligned}\hat{Q}_n^{\mathrm{RW}}(\theta)&= \frac{1}{n-m}\sum_{t=m+\tau+L}^{n} \left[\norm{F_{t-L}^{\mathrm{obs}}-\bar{F}_{t-L}^{\mathrm{sim}}(\theta)}^2 + K_{t}(\theta)\right],\\  K_{t}(\theta)&=2\left(F_{t-L}^{\mathrm{obs}}-\bar{F}_{t-L}^{\mathrm{sim}}(\theta) \right)^{\top}\left(\left[f_t^{\mathrm{obs}}-\bar{f}_{t}^{\mathrm{sim}}(\theta)\right]-\left[F_{t-L}^{\mathrm{obs}}-\bar{F}_{t-L}^{\mathrm{sim}}(\theta)\right]\right),\end{aligned}\end{equation*} and the population discrepancy function is given by $$Q^{\mathrm{RW}}(\theta) =\int_0^1 \norm{\Phi_u(\theta_0) - \Phi_u(\theta)}^2 du.$$ To use Corollary 3.2 of~\cite{pakes_pollard_simulation_asymptotics}, it suffices to show that, under the assumptions of Theorem~\ref{thm:consistency_rolling_window_estimator}, the following conditions hold: 
{\renewcommand{\labelenumi}{(\alph{enumi})}
\begin{enumerate}
    \item $\left|\hat{Q}_n^{\mathrm{RW}}(\hat{\theta}^{\mathrm{RW}})\right|\leq \inf_{\theta \in \Theta}\left|\hat{Q}_n^{\mathrm{RW}}(\theta)\right| + o_p(1)$.
    \item $\inf_{\norm{\theta - \theta_0} > \delta} \abs{Q^{\mathrm{RW}}(\theta)} >0$ for each $\delta >0$. 
    \item $\sup_{\theta\in\Theta}\frac{\abs{\hat{Q}_n^{\mathrm{RW}}(\theta) - Q^{\mathrm{RW}}(\theta)}}{1+\abs{\hat{Q}_n^{\mathrm{RW}}(\theta)}+\abs{Q^{\mathrm{RW}}(\theta)}}=o_p(1)$.
\end{enumerate}}

\paragraph*{Condition (a).}

This condition is immediate from~\eqref{eqn:rolling_window_estimator}.

\paragraph*{Condition (b).}

%%% Note: for existence of sequence (theta_n)_n see infimum "Relation of limits to sequences" https://en.wikipedia.org/wiki/Infimum_and_supremum

%%% Note: compact implies sequentially compact in R^p https://en.wikipedia.org/wiki/Sequentially_compact_space

%%% Note: transfer limit of parameter subsequence to limit of distance of parameter subsequence and true parameter by using continuity of distance i.e. for all epsilon there existence delta=epsilon such that if d(x,y)<delta then |d(x,a)-d(y,a)|< d(x,y) < delta = epsilon by reverse triangle inequality, and limit definition of continuity https://en.wikipedia.org/wiki/Continuous_function

We establish that the infimum is strictly positive through contradiction. Assume for contradiction that there exists a $\delta >0$ such that $$\inf_{\norm{\theta - \theta_0} > \delta} \abs{Q^{\mathrm{RW}}(\theta)} =0.$$ 
By applying the definition of infimum and the continuity of $Q^{\mathrm{RW}}$ by Lemma~\ref{lma:continuity_of_theta_to_QRW}, there exists a sequence $(\theta_n)_{n\in\mathbb{N}}$, where $\norm{\theta_n - \theta_0} > \delta$ for all $n\in\mathbb{N}$, such that $\lim_{n\xrightarrow[]{}\infty} \abs{Q^{\mathrm{RW}}(\theta_n)} = 0$.

By Assumption~\ref{asmpt:statistical_model_continuous_time}, $\Theta$ is a compact subset of $\mathbb{R}^p$, hence $\Theta$ is also sequentially compact, i.e., every sequence of parameters in $\Theta$ has a convergent subsequence converging to a parameter in $\Theta$. Thus, there exists a subsequence $(\theta_{n_j})_{j\in\mathbb{N}}$, where $\norm{\theta_{n_j} - \theta_0} > \delta$ for all $j\in\mathbb{N}$, and a parameter $\theta^{\ast}\in \Theta$ satisfying $\abs{Q^{\mathrm{RW}}(\theta^{\ast})} = Q^{\mathrm{RW}}(\theta^{\ast}) =0$, such that $\theta_{n_j}\xrightarrow[]{}\theta^{\ast}$ as $j\xrightarrow[]{}\infty$.

By Assumption~\ref{asmpt:distinguish_paths_of_distributions_continuous_time}, for all $\theta_0,\theta^{\ast}\in\Theta$, if $\theta_0\neq\theta^{\ast}$, then there exists a subset $\mathcal{U}_{\theta_0,\theta^{\ast}}\subseteq [0,1]$ with positive Lebesgue measure such that $P_{\theta_0,u}\neq P_{\theta^{\ast},u}$ for all $u\in \mathcal{U}_{\theta_0,\theta^{\ast}}$. Since $P_{\theta_0,u}\neq P_{\theta^{\ast},u}$ for all $u\in\mathcal{U}_{\theta_0,\theta^{\ast}}$, we also have that $\tilde{P}_{g_u(\theta_0),u}\neq \tilde{P}_{g_u(\theta^{\ast}),u}$ for all $u\in\mathcal{U}_{\theta_0,\theta^{\ast}}$ because $\tilde{P}_{g_u(\theta),u}=P_{\theta,u}$ is the distribution of $$[\tilde{G}_u(\symbfit{\varepsilon}_{-m},g_u(\theta)),\ldots,\tilde{G}_u(\symbfit{\varepsilon}_0,g_u(\theta))]^{\top}=[G_u(\symbfit{\varepsilon}_{-m},\theta),\ldots,G_u(\symbfit{\varepsilon}_0,\theta)]^{\top}.$$ Therefore, for all $u\in\mathcal{U}_{\theta_0,\theta^{\ast}}$, we have that $g_u(\theta_0)\neq g_u(\theta^{\ast})$ because $\tilde{\theta}_u \mapsto \tilde{P}_{\tilde{\theta}_u,u}$ is injective for all $u\in [0,1]$ by Assumption~\ref{asmpt:statistical_model_continuous_time}. We also have that
$\tilde{\Phi}_u(g_u(\theta_0))\neq \tilde{\Phi}_u(g_u(\theta^{\ast}))$ for almost every $u\in\mathcal{U}_{\theta_0,\theta^{\ast}}$, because by Assumption~\ref{asmpt:injective_continuous_time}, $\tilde{\theta}_u\mapsto \tilde{\Phi}_u(\tilde{\theta}_u)$ is one-to-one for almost every $u\in[0,1]$. This also implies that $\Phi_u(\theta_0)\neq \Phi_u(\theta^{\ast})$ for almost every $u\in\mathcal{U}_{\theta_0,\theta^{\ast}}$ by the definitions of $\Phi_u$ and $\tilde{\Phi}_u$ from~\eqref{eqn:original_random_feature_expectation_cont_time} and~\eqref{eqn:reparam_random_feature_expectation}, respectively, because $\tilde{P}_{g_u(\theta),u}=P_{\theta,u}$ as discussed above. Putting it all together, we have that, if $\theta_0 \neq \theta^{\ast}$, then  \begin{equation*}
Q^{\mathrm{RW}}(\theta^{\ast}) =\int_0^1 \norm{\Phi_u(\theta_0) - \Phi_u(\theta^{\ast})}^2 du  \geq \int_{\mathcal{U}_{\theta_0,\theta^{\ast}}} \norm{\Phi_u(\theta_0) - \Phi_u(\theta^{\ast})}^2 du  >0.\end{equation*} In other words, if $Q^{\mathrm{RW}}(\theta^{\ast})=0$, then $\theta^{\ast}=\theta_0$.

However, by the continuity of $\theta \mapsto \norm{\theta-\theta_0}$, we have $\norm{\theta_{n_j}-\theta_0}\xrightarrow[]{}\norm{\theta^{\ast}-\theta_0}$ as $j\xrightarrow[]{}\infty$. Therefore, $\norm{\theta^{\ast}-\theta_0} \geq \delta$ because $\norm{\theta_{n_j} - \theta_0} > \delta$ for all $j\in\mathbb{N}$, so we have $$0 = \norm{\theta_0-\theta_0} = \norm{\theta^{\ast}-\theta_0}\geq\delta > 0,$$ which is a contradiction. %%% Note: Otherwise, if theta^{\ast} was less than delta distance away from theta_0, we would need elements of the sequence theta_{n_j} to be to be within epsilon distance of theta^{\ast} (for all epsilon >0) which would mean that those elements theta_{n_j} would be less than delta distance away from theta_0 which is impossible.

\paragraph*{Condition (c).}

It suffices to show that 
$\sup_{\theta\in\Theta}\abs{\hat{Q}_n^{\mathrm{RW}}(\theta) - Q^{\mathrm{RW}}(\theta)}=o_p(1)$ because the denominator is always greater than or equal to one. For $\theta\in\Theta$ and $r\in\{0,1,\ldots,s\}$, define \begin{equation}\label{eqn:random_feature_time_average_approx_one_realization}F_{t}^{(r)}(\theta) =  \frac{1}{N_t}\sum_{j=(t-w)\lor (m+1)}^t \varphi\left(\left[G_{(j-m)/n}^{(n,r)}(\symbfit{\varepsilon}_{j-m}^{(r)},\theta),\ldots,G_{j/n}^{(n,r)}(\symbfit{\varepsilon}_j^{(r)},\theta)\right]^{\top}\right),\end{equation} and
\begin{equation}\label{eqn:decomp_for_rolling_window}\begin{aligned}
I_n^1(\theta) &= \frac{1}{n-m}\sum_{t=m+\tau+L}^n \norm{\mu_{\frac{t}{n}}^n(\theta)}^2,  \\ 
I_n^2(\theta) &=  \frac{1}{n-m}\sum_{t=m+\tau+L}^n  \norm{\mu_{\frac{t}{n}}^n(\theta) - \hat{\mu}_{t-L,n}(\theta)}^2,\\ 
I_n^3(\theta) &= \frac{1}{n-m}\sum_{t=m+\tau+L}^n K_{t}^{\ast}(\theta),
\end{aligned}
\end{equation}
where \begin{equation}\label{eqn:auxiliary_terms_for_rolling_window}\begin{aligned}\mu_{\frac{t}{n}}^n(\theta)&=\mathbb{E}_{\theta_0}\left[f_t^{\mathrm{obs}}\right]-\frac{1}{s}\sum_{r=1}^s\mathbb{E}_{\theta}\left[f_t^{(r)}(\theta)\right]\\&=\mathbb{E}_{\theta_0}\left[\varphi\left(\left[G_{(t-m)/n}^{(n,0)}(\symbfit{\varepsilon}_{t-m}^{(0)},\theta_0),\ldots,G_{t/n}^{(n,0)}(\symbfit{\varepsilon}_t^{(0)},\theta_0)\right]^{\top}\right)\right]\\ &- \frac{1}{s}\sum_{r=1}^s \mathbb{E}_{\theta}\left[\varphi\left(\left[G_{(t-m)/n}^{(n,r)}(\symbfit{\varepsilon}_{t-m}^{(r)},\theta),\ldots,G_{t/n}^{(n,r)}(\symbfit{\varepsilon}_t^{(r)},\theta)\right]^{\top}\right)\right],\\
\hat{\mu}_{t,n}(\theta) &= F_t^{\mathrm{obs}}-\bar{F}_t^{\mathrm{sim}}(\theta),
\\
K_{t}(\theta)&=2\hat{\mu}_{t-L,n}(\theta)^{\top}\left(\left[f_t^{\mathrm{obs}}-\bar{f}_{t}^{\mathrm{sim}}(\theta)\right]-\hat{\mu}_{t-L,n}(\theta)\right),\\
K_{t}^{\ast}(\theta)&=2\hat{\mu}_{t-L,n}(\theta)^{\top}\left(\left[f_t^{\mathrm{obs}}-\bar{f}_{t}^{\mathrm{sim}}(\theta)\right]-\mu_{\frac{t}{n}}^n(\theta)\right),
\end{aligned}\end{equation} %%% Note: the second derivative in I_n^2 is just the 2 times the d x d identity matrix, so since it is divided by 2 in I_n^2 we can just ignore it... and note that the result is x^T x so we just simplify to ||x||_2^2
and
\begin{equation}\label{eqn:mu_u}\begin{aligned}\breve{\mu}_{u}^n(\theta)&=\mathbb{E}_{\theta_0}\left[\varphi\left(\left[G_{u}^{(n,0)}(\symbfit{\varepsilon}_{-m}^{(0)},\theta_0),\ldots,G_{u}^{(n,0)}(\symbfit{\varepsilon}_0^{(0)},\theta_0)\right]^{\top}\right)\right]\\ &- \frac{1}{s}\sum_{r=1}^s \mathbb{E}_{\theta}\left[\varphi\left(\left[G_{u}^{(n,r)}(\symbfit{\varepsilon}_{-m}^{(r)},\theta),\ldots,G_{u}^{(n,r)}(\symbfit{\varepsilon}_0^{(r)},\theta)\right]^{\top}\right)\right].
\end{aligned}
\end{equation}

We have
\begin{align*}&\sup_{\theta\in\Theta}\abs{\hat{Q}_n^{\mathrm{RW}}(\theta) - Q^{\mathrm{RW}}(\theta)} \\& \overset{(1)}{=} \sup_{\theta\in\Theta}\abs{\frac{1}{n-m}\sum_{t=m+\tau+L}^{n} \left[\norm{F_{t-L}^{\mathrm{obs}}-\bar{F}_{t-L}^{\mathrm{sim}}(\theta)}^2 + K_{t}(\theta)\right] - \int_0^1 \norm{\Phi_u(\theta_0) - \Phi_u(\theta)}^2 du}
\\& \overset{(2)}{=} \sup_{\theta\in\Theta}\abs{I_n^1(\theta) - I_n^2(\theta) + I_n^3(\theta) - \int_0^1 \norm{\Phi_u(\theta_0) - \Phi_u(\theta)}^2 du}
\\& \overset{(3)}{\leq} \sup_{\theta\in\Theta}\abs{I_n^1(\theta) - \int_0^1 \norm{\Phi_u(\theta_0) - \Phi_u(\theta)}^2 du}
+ \sup_{\theta\in\Theta}\abs{I_n^2(\theta)} + \sup_{\theta\in\Theta}\abs{I_n^3(\theta)},
\end{align*}
by (1) the definition of $\hat{Q}_n^{\mathrm{RW}}(\theta)$ and $Q^{\mathrm{RW}}(\theta)$, (2) a Taylor expansion (see below) and the definition of $I_n^1(\theta)$, $I_n^2(\theta)$, $I_n^3(\theta)$, and (3) the triangle inequality and subadditivity of the supremum. More specifically, regarding (2), at each time $t$, a Taylor expansion yields 
\begin{align*}
   & \norm{\mu_{\frac{t}{n}}^n(\theta)}^2 \\ &= \norm{\hat{\mu}_{t-L,n}(\theta)}^2+2\hat{\mu}_{t-L,n}(\theta)^{\top}\left(\mu_{\frac{t}{n}}^n(\theta)-\hat{\mu}_{t-L,n}(\theta)\right)+\norm{\mu_{\frac{t}{n}}^n(\theta)-\hat{\mu}_{t-L,n}(\theta)}^2,
\end{align*}%%% Note: the higher order terms vanish because the derivatives will be zero 
so rearranging this equation and plugging in for $\norm{\hat{\mu}_{t-L,n}(\theta)}^2=\norm{F_{t-L}^{\mathrm{obs}}-\bar{F}_{t-L}^{\mathrm{sim}}(\theta)}^2$ yields
\begin{align*}
& \norm{F_{t-L}^{\mathrm{obs}}-\bar{F}_{t-L}^{\mathrm{sim}}(\theta)}^2 + K_{t}(\theta)
\\ & = 
    \norm{\mu_{\frac{t}{n}}^n(\theta)}^2  - \norm{\mu_{\frac{t}{n}}^n(\theta)-\hat{\mu}_{t-L,n}(\theta)}^2 \\ & +\left[- 2\hat{\mu}_{t-L,n}(\theta)^{\top}\left(\mu_{\frac{t}{n}}^n(\theta)-\hat{\mu}_{t-L,n}(\theta)\right) + K_{t}(\theta)\right]
\\& = \norm{\mu_{\frac{t}{n}}^n(\theta)}^2 - \norm{\mu_{\frac{t}{n}}^n(\theta)-\hat{\mu}_{t-L,n}(\theta)}^2 + 2\hat{\mu}_{t-L,n}(\theta)^{\top}\left(\left[f_t^{\mathrm{obs}}-\bar{f}_{t}^{\mathrm{sim}}(\theta)\right]-\mu_{\frac{t}{n}}^n(\theta)\right),
\end{align*}
therefore averaging yields 
$$\frac{1}{n-m}\sum_{t=m+\tau+L}^{n} \left[\norm{F_{t-L}^{\mathrm{obs}}-\bar{F}_{t-L}^{\mathrm{sim}}(\theta)}^2 + K_{t}(\theta)\right]=I_n^1(\theta) - I_n^2(\theta) + I_n^3(\theta).$$
%%% Note: Taylor expansion of ||x||^2 at $\mu_{\frac{t}{n}}$$ 
%%% Note: as usual, see e.g. https://en.wikipedia.org/wiki/Taylor%27s_theorem and https://mathinsight.org/taylors_theorem_multivariable_introduction
%%% Note: this is exact, no higher order terms... can just derive this by adding and subtracting the same term and using linearity of inner product...

%%% Note: by union bounding the above derivation, as usual... 
It suffices to show that
\begin{align*}\sup_{\theta\in\Theta}\abs{I_n^1(\theta) - \int_0^1 \norm{\Phi_u(\theta_0) - \Phi_u(\theta)}^2 du}&=o(1),\\ 
\sup_{\theta\in\Theta}\abs{I_n^2(\theta)}&=o_p(1),  \\ \sup_{\theta\in\Theta}\abs{I_n^3(\theta)}&=o_p(1),\end{align*} so we break down the remainder of the proof into these three steps.

\textbf{Step 1.} We show that \begin{align*}\sup_{\theta\in\Theta}\abs{I_n^1(\theta) - \int_0^1 \norm{\Phi_u(\theta_0) - \Phi_u(\theta)}^2 du}=o(1).\end{align*}

Observe that \begin{align*}&\sup_{\theta\in\Theta}\abs{I_n^1(\theta) - \int_0^1 \norm{\Phi_u(\theta_0) - \Phi_u(\theta)}^2 du}\\
&
\overset{(1)}{=}  \sup_{\theta\in\Theta}\abs{\frac{1}{n-m}\sum_{t=m+\tau+L}^n \norm{\mu_{\frac{t}{n}}^n(\theta)}^2 - \int_0^1 \norm{\Phi_u(\theta_0) - \Phi_u(\theta)}^2 du}\\
&
\overset{(2)}{\leq}  \sup_{\theta\in\Theta}\abs{\frac{1}{n-m}\sum_{t=m+\tau+L}^n \norm{\mu_{\frac{t}{n}}^n(\theta)}^2 - \int_0^1 \norm{\breve{\mu}_{u}^n(\theta)}^2 du}\\
&
+
\sup_{\theta\in\Theta}\abs{\int_0^1 \norm{\breve{\mu}_{u}^n(\theta)}^2 du - \int_0^1 \norm{\Phi_u(\theta_0) - \Phi_u(\theta)}^2 du},
\end{align*} by
(1) the definition of $I_n^1(\theta)$, and (2) adding and subtracting the same term, the triangle inequality, and subadditivity of the supremum. We will show that both terms are $o(1)$. 

\textbf{Step 1.1.} Observe that \begin{align*}&\sup_{\theta\in\Theta}\abs{\frac{1}{n-m}\sum_{t=m+\tau+L}^n \norm{\mu_{\frac{t}{n}}^n(\theta)}^2 - \int_{\frac{m+\tau+L-1}{n}}^1 \norm{\breve{\mu}_{u}^n(\theta)}^2 du}\\
&
\overset{(1)}{\leq} 
\frac{4mk}{n} + \sup_{\theta\in\Theta}\abs{\frac{1}{n}\sum_{t=m+\tau+L}^n \norm{\mu_{\frac{t}{n}}^n(\theta)}^2 - \int_{\frac{m+\tau+L-1}{n}}^1 \norm{\breve{\mu}_{u}^n(\theta)}^2 du}\\
&
\overset{(2)}{\leq}  \frac{4mk}{n} +  \sup_{\theta\in\Theta}\int_{\frac{m+\tau+L-1}{n}}^1\abs{ \norm{\mu_{\frac{\lceil un \rceil}{n}}^n(\theta)}^2 - \norm{\breve{\mu}_{u}^n(\theta)}^2} du\\  
&
\overset{(3)}{\leq}  \frac{4mk}{n} + \frac{1}{n} \sup_{\theta\in\Theta} \sum_{t=m+\tau+L}^{n} \underset{u\in \left[\frac{t-1}{n},\frac{t}{n}\right]}{\sup} \abs{ \norm{\mu_{\frac{t}{n}}^n(\theta)}^2 - \norm{\breve{\mu}_{u}^n(\theta)}^2}\\  
&
\overset{(4)}{\leq}   \frac{4mk}{n} + \frac{1}{n} \sup_{\theta\in\Theta} \sum_{t=m+\tau+L}^{n} \underset{u\in \left[\frac{t-1}{n},\frac{t}{n}\right]}{\sup}  \left(\norm{\mu_{\frac{t}{n}}^n(\theta)}+\norm{\breve{\mu}_{u}^n(\theta)}\right)\norm{\mu_{\frac{t}{n}}^n(\theta)-\breve{\mu}_{u}^n(\theta)}\\
&
\overset{(5)}{\leq} \frac{4mk}{n} +  4\sqrt{k} \frac{1}{n}  \sup_{\theta\in\Theta} \sum_{t=m+\tau+L}^{n} \underset{u\in \left[\frac{t-1}{n},\frac{t}{n}\right]}{\sup}  \left(\norm{\breve{\mu}_{\frac{t}{n}}^n(\theta)-\breve{\mu}_{u}^n(\theta)}+\norm{\mu_{\frac{t}{n}}^n(\theta)-\breve{\mu}_{\frac{t}{n}}^n(\theta)}\right) \\
& 
\overset{(6)}{\leq} \frac{4mk}{n} \\& + 4\sqrt{k}L_{\varphi} (m+1) \frac{1}{n}   \sup_{\theta\in\Theta} \sum_{t=m+\tau+L}^{n} \underset{u\in \left[\frac{t-1}{n},\frac{t}{n}\right]}{\sup}  \norm{
G_{\frac{t}{n}}^{(n,0)}(\symbfit{\varepsilon}_{0}^{(0)},\theta) - 
G_{u}^{(n,0)}(\symbfit{\varepsilon}_{0}^{(0)},\theta)
}_{\mathcal{L}^1(\theta)}\\
& 
+ 4\sqrt{k}L_{\varphi}(m+1)  \frac{1}{n} \frac{1}{s}\sum_{r=1}^s \\& \sup_{\theta\in\Theta} \left(\sum_{t=m+\tau+L}^{n}  \underset{u\in \left[\frac{t-1}{n},\frac{t}{n}\right]}{\sup}  \norm{
G_{\frac{t}{n}}^{(n,r)}(\symbfit{\varepsilon}_{0}^{(r)},\theta) - 
G_{u}^{(n,r)}(\symbfit{\varepsilon}_{0}^{(r)},\theta)
}_{\mathcal{L}^1(\theta)}\right) \\
& + 4\sqrt{k}L_{\varphi} \frac{1}{n}   \sup_{\theta\in\Theta} \sum_{t=m+\tau+L}^{n}  \sum_{i=0}^m \norm{
G_{\frac{t-i}{n}}^{(n,0)}(\symbfit{\varepsilon}_{0}^{(0)},\theta) - 
G_{\frac{t}{n}}^{(n,0)}(\symbfit{\varepsilon}_{0}^{(0)},\theta)
}_{\mathcal{L}^1(\theta)}\\
& 
+ 4\sqrt{k}L_{\varphi} \frac{1}{n} \frac{1}{s}\sum_{r=1}^s \sup_{\theta\in\Theta} \left(\sum_{t=m+\tau+L}^{n} \sum_{i=0}^m \norm{
G_{\frac{t-i}{n}}^{(n,r)}(\symbfit{\varepsilon}_{0}^{(r)},\theta) - 
G_{\frac{t}{n}}^{(n,r)}(\symbfit{\varepsilon}_{0}^{(r)},\theta)
}_{\mathcal{L}^1(\theta)}\right) \\
& 
\overset{(7)}{\leq} \frac{4mk}{n} +  8\sqrt{k}L_{\varphi} (m+1)^2\frac{n^{1-1/\kappa}}{n}\Lambda,
\end{align*} 
by (1) adding and subtracting the same term (i.e., with a factor of $\frac{1}{n}$), the triangle inequality, subadditivity of the supremum, factoring out the sum of squared norms, the equality $\frac{1}{n-m}-\frac{1}{n}=\frac{m}{n(n-m)}$, upper bounding the sum by $n-m$ times the maximum over time $t$ of the squared norm, noting that the random Fourier features are bounded between $-1$ and $1$ so the corresponding expectations must be as well, then upper bounding the squared norm by $4k$, (2) the definition of the ceiling function, using the Riemann sum approximation of the integral, and the triangle inequality for integrals, (3) upper bounding the integral over each interval by the supremum over that interval multiplied by the interval length and summing over the intervals, (4) the equality $\abs{x^2-y^2}=\abs{x-y}\abs{x+y}$ applied to the norms and the triangle inequality applied to this absolute value of the sum of norms, (5) using the definition of $\mu_{\frac{t}{n}}^n(\theta)$ and $\breve{\mu}_{u}^n(\theta)$ and applying the triangle inequality to each of their norms, applying the triangle inequality and absolute homogeneity to the simulation-average, upper bounding each norm by $\sqrt{k}$ where $k=2p+1$, $\theta \in \mathbb{R}^p$, because each of the $k$ dimensions of the random Fourier features takes values in $[-1,1]$, so we can upper bound by four times this amount, and the distributive property, (6) the triangle inequality, subadditivity of the supremum, linearity of expectation, Jensen's inequality, the random Fourier features $\varphi$ are $L_{\varphi}$-Lipschitz with Lipschitz constant given by~\eqref{eqn:lipschitz_constant_rff}, the triangle inequality, the Minkowski inequality, shifting the noise inputs which is equivalent due to stationarity, and (7) H\"{o}lder's inequality applied to each sum over time $t$ (i.e., for each $r=0,1,\ldots,s$), by telescoping the sums in the third term and fourth term, upper bounding by the $\kappa$-variation, and applying the upper bound from Assumption~\ref{asmpt:nonstationarity_continuous_time}.

Lastly, observe that 
\begin{align*}&\sup_{\theta\in\Theta}\abs{\frac{1}{n-m}\sum_{t=m+\tau+L}^n \norm{\mu_{\frac{t}{n}}^n(\theta)}^2 - \int_{0}^1 \norm{\breve{\mu}_{u}^n(\theta)}^2 du}\\
& 
\overset{(1)}{\leq} 
\sup_{\theta\in\Theta}\abs{\frac{1}{n-m}\sum_{t=m+\tau+L}^n \norm{\mu_{\frac{t}{n}}^n(\theta)}^2 - \int_{\frac{m+\tau+L-1}{n}}^1 \norm{\breve{\mu}_{u}^n(\theta)}^2 du} \\
& + \sup_{\theta\in\Theta}\abs{\int_{\frac{m+\tau+L-1}{n}}^1 \norm{\breve{\mu}_{u}^n(\theta)}^2 du - \int_{0}^1 \norm{\breve{\mu}_{u}^n(\theta)}^2 du}
\\
&
\overset{(2)}{=}
\sup_{\theta\in\Theta}\abs{\frac{1}{n-m}\sum_{t=m+\tau+L}^n \norm{\mu_{\frac{t}{n}}^n(\theta)}^2 - \int_{\frac{m+\tau+L-1}{n}}^1 \norm{\breve{\mu}_{u}^n(\theta)}^2 du} \\&  + \sup_{\theta\in\Theta}\abs{\int_0^{\frac{m+\tau+L-1}{n}}\norm{\breve{\mu}_{u}^n(\theta)}^2 du}
\\&
\overset{(3)}{\leq}
\left(\frac{4mk}{n} +  8\sqrt{k}L_{\varphi} (m+1)^2\frac{n^{1-1/\kappa}}{n}\Lambda\right)  + \frac{m+\tau+L-1}{n} \ 4k 
\\&
\overset{(4)}{=} o(1),
\end{align*}
by (1) adding and subtracting the same term, the triangle inequality, and subadditivity of the supremum, (2) subtraction, (3) upper bounding the first term by applying the inequality derived previously, and upper bounding the second term by the length of the interval $[0,\frac{m+\tau+L-1}{n}]$ multiplied by $4k$ because the integrand can be upper bounded by $4k$ since the random Fourier features each take values in $[-1,1]$, and (4) noting that both of these terms converge to zero because the first term is $o(1)$ because $\kappa\in [1,4)$ by Assumption~\ref{asmpt:nonstationarity_continuous_time} and the second term is $o(1)$ because $\tau_n =o(n)$ by Assumption~\ref{asmpt:growth_of_offset_and_window_size_continuous_time}.

\textbf{Step 1.2.} Observe that \begin{align*}&\sup_{\theta\in\Theta}\abs{\int_0^1 \norm{\breve{\mu}_{u}^n(\theta)}^2 du - \int_0^1 \norm{\Phi_u(\theta_0) - \Phi_u(\theta)}^2 du}\\
& 
\overset{(1)}{\leq}
\sup_{\theta\in\Theta}\int_0^1 \abs{\norm{\breve{\mu}_{u}^n(\theta)}^2  -  \norm{\Phi_u(\theta_0) - \Phi_u(\theta)}^2} du
\\
&
\overset{(2)}{\leq} 4\sqrt{k}
\sup_{\theta\in\Theta}\int_0^1 \norm{\breve{\mu}_{u}^n(\theta)-(\Phi_u(\theta_0) - \Phi_u(\theta))}du
\\
&
\overset{(3)}{\leq} 4\sqrt{k}
\int_0^1 \norm{\mathbb{E}_{\theta_0}\left[\varphi\left(\left[G_{u}^{(n,0)}(\symbfit{\varepsilon}_{-m}^{(0)},\theta_0),\ldots,G_{u}^{(n,0)}(\symbfit{\varepsilon}_0^{(0)},\theta_0)\right]^{\top}\right)\right] -\Phi_u(\theta_0)}du
\\&
+
4\sqrt{k} \frac{1}{s}\sum_{r=1}^s
\sup_{\theta\in\Theta}\int_0^1 \norm{ \mathbb{E}_{\theta}\left[\varphi\left(\left[G_{u}^{(n,r)}(\symbfit{\varepsilon}_{-m}^{(r)},\theta),\ldots,G_{u}^{(n,r)}(\symbfit{\varepsilon}_0^{(r)},\theta)\right]^{\top}\right)\right]-\Phi_u(\theta)}du
\\& 
\overset{(4)}{\leq} 
4\sqrt{k}
\underset{u \in [0,1]}{\sup} \norm{\mathbb{E}_{\theta_0}\left[\varphi\left(\left[G_{u}^{(n,0)}(\symbfit{\varepsilon}_{-m}^{(0)},\theta_0),\ldots,G_{u}^{(n,0)}(\symbfit{\varepsilon}_0^{(0)},\theta_0)\right]^{\top}\right)\right] -\Phi_u(\theta_0)}
\\&
+ 4\sqrt{k} \frac{1}{s}\sum_{r=1}^s
\sup_{\theta\in\Theta}\underset{u\in [0,1]}{\sup} \norm{ \mathbb{E}_{\theta}\left[\varphi\left(\left[G_{u}^{(n,r)}(\symbfit{\varepsilon}_{-m}^{(r)},\theta),\ldots,G_{u}^{(n,r)}(\symbfit{\varepsilon}_0^{(r)},\theta)\right]^{\top}\right)\right]-\Phi_u(\theta)}\\&
\overset{(5)}{=} o(1),
\end{align*} by (1) the linearity of integrals and the triangle inequality for integrals, (2) the equality $\abs{x^2-y^2}=\abs{x-y}\abs{x+y}$ applied to the norms and the triangle inequality applied to this absolute value of the sum of norms, then using the definition of $\breve{\mu}_{u}^n(\theta)$ and applying the triangle inequality to each of the norms, applying the triangle inequality and absolute homogeneity to the simulation-average terms, upper bounding each norm by $\sqrt{k}$ where $k=2p+1$, $\theta \in \mathbb{R}^p$, because each of the $k$ dimensions of the random Fourier features takes values in $[-1,1]$, so we can upper bound by four times this amount, and the distributive property, (3) the triangle inequality, subadditivity, and the distributive property, (4) upper bounding the integral by the length of the interval multiplied by the supremum over the interval, and (5) the same $(u,\theta)$-pointwise inequalities used to show~\eqref{eqn:convergence_to_tilde_Phi_theta_u_continuous_time}, where $\tilde{\Phi}_u(\tilde{\theta}_u)=\Phi_u(\theta)$ by the definitions and Assumption~\ref{asmpt:reparametrization_continuous_time}, hold uniformly over $u\in [0,1]$ and over $\theta\in\Theta$ by using the subadditivity of the supremum and because Assumption~\ref{asmpt:limiting_mapping_continuous_time} is a $(u,\theta)$-uniform convergence condition.

Putting it all together, we have \begin{align*}&\sup_{\theta\in\Theta}\abs{I_n^1(\theta) - \int_0^1 \norm{\Phi_u(\theta_0) - \Phi_u(\theta)}^2 du}=o(1).
\end{align*}

\textbf{Step 2.} We show that \begin{align*} 
\sup_{\theta\in\Theta}\abs{I_n^2(\theta)}=o_p(1).\end{align*}

%%% Note that assumption A.4 from Mies 2023 is automatically satisfied because the "time series of random features" is bounded between -1 and 1 so the norms will be bounded... and the derivative is just 2 times the RW mean and the second derivative is just 2 
Observe that
\begin{align*}
&\sup_{\theta\in\Theta}\abs{I_n^2(\theta)}
\\ 
&
\overset{(1)}{=} \sup_{\theta\in\Theta} \ \frac{1}{n-m}\sum_{t=m+\tau+L}^n  \norm{\mu_{\frac{t}{n}}^n(\theta) - \hat{\mu}_{t-L,n}(\theta)}^2
\\ 
&
\overset{(2)}{\leq}  \sup_{\theta\in\Theta} \ \frac{1}{n-m}\sum_{t=m+\tau+L}^n  \left(\norm{\mu_{\frac{t}{n}}^n(\theta) - \mu_{\frac{t-L}{n}}^n(\theta)} + \norm{\mu_{\frac{t-L}{n}}^n(\theta)-\hat{\mu}_{t-L,n}(\theta)}\right)^2
\\ 
&
\overset{(3)}{\leq}  \sup_{\theta\in\Theta} \ \frac{2}{n-m}\sum_{t=m+\tau+L}^n  \norm{\mu_{\frac{t}{n}}^n(\theta) - \mu_{\frac{t-L}{n}}^n(\theta)}^2 
\\&
+\sup_{\theta\in\Theta} \ \frac{2}{n-m}\sum_{t=m+\tau+L}^n  \norm{\mu_{\frac{t-L}{n}}^n(\theta)-\hat{\mu}_{t-L,n}(\theta)}^2,
\end{align*}
by (1) the definition of $I_n^2(\theta)$, (2) adding and subtracting the same term then applying the triangle inequality, and (3) expanding the square then applying Young's inequality and subadditivity of the supremum. We will show the first term is $o(1)$ and the second term is $o_p(1)$.

\textbf{Step 2.1.} To begin, observe that, for every $\theta\in\Theta$, we have
\begin{align*}
& \frac{1}{2}\sum_{t=m+\tau+L}^n  \norm{\mu_{\frac{t}{n}}^n(\theta) - \mu_{\frac{t-L}{n}}^n(\theta)}^2 
\\&
\overset{(1)}{\leq}
 \sum_{t=m+\tau+L}^n  
\norm{\begin{aligned}&\mathbb{E}_{\theta_0}\varphi\left(\left[G_{(t-m)/n}^{(n,0)}(\symbfit{\varepsilon}_{-m}^{(0)},\theta_0),\ldots,G_{t/n}^{(n,0)}(\symbfit{\varepsilon}_0^{(0)},\theta_0)\right]^{\top}\right)\\&-\mathbb{E}_{\theta_0}\varphi\left(\left[G_{(t-L-m)/n}^{(n,0)}(\symbfit{\varepsilon}_{-m}^{(0)},\theta_0),\ldots,G_{(t-L)/n}^{(n,0)}(\symbfit{\varepsilon}_{0}^{(0)},\theta_0)\right]^{\top}\right)\end{aligned}}^2 
\\&
+
\frac{1}{s}\sum_{r=1}^s  \sum_{t=m+\tau+L}^n  
\norm{\begin{aligned}&\mathbb{E}_{\theta}\varphi\left(\left[G_{(t-m)/n}^{(n,r)}(\symbfit{\varepsilon}_{-m}^{(r)},\theta),\ldots,G_{t/n}^{(n,r)}(\symbfit{\varepsilon}_0^{(r)},\theta)\right]^{\top}\right)\\&-\mathbb{E}_{\theta}\varphi\left(\left[G_{(t-L-m)/n}^{(n,r)}(\symbfit{\varepsilon}_{-m}^{(r)},\theta),\ldots,G_{(t-L)/n}^{(n,r)}(\symbfit{\varepsilon}_{0}^{(r)},\theta)\right]^{\top}\right)\end{aligned}}^2 
\\&
\overset{(2)}{\leq} L^2 \sum_{t=m+1}^n 
\norm{\begin{aligned}&\mathbb{E}_{\theta_0}\varphi\left(\left[G_{(t-m)/n}^{(n,0)}(\symbfit{\varepsilon}_{-m}^{(0)},\theta_0),\ldots,G_{t/n}^{(n,0)}(\symbfit{\varepsilon}_{0}^{(0)},\theta_0)\right]^{\top}\right)\\&-\mathbb{E}_{\theta_0}\varphi\left(\left[G_{(t-m-1)/n}^{(n,0)}(\symbfit{\varepsilon}_{-m}^{(0)},\theta_0),\ldots,G_{(t-1)/n}^{(n,0)}(\symbfit{\varepsilon}_{0}^{(0)},\theta_0)\right]^{\top}\right)\end{aligned}}^2 
\\&
+
\frac{L^2}{s}\sum_{r=1}^s  \sum_{t=m+1}^n  
\norm{\begin{aligned}&\mathbb{E}_{\theta}\varphi\left(\left[G_{(t-m)/n}^{(n,r)}(\symbfit{\varepsilon}_{-m}^{(r)},\theta),\ldots,G_{t/n}^{(n,r)}(\symbfit{\varepsilon}_{0}^{(r)},\theta)\right]^{\top}\right)\\&-\mathbb{E}_{\theta}\varphi\left(\left[G_{(t-m-1)/n}^{(n,r)}(\symbfit{\varepsilon}_{-m}^{(r)},\theta),\ldots,G_{(t-1)/n}^{(n,r)}(\symbfit{\varepsilon}_{0}^{(r)},\theta)\right]^{\top}\right)\end{aligned}}^2
\\& 
\overset{(3)}{\leq}
L^2 L_{\varphi}^2 (m+1)^2 \sum_{t=1}^n 
\norm{G_{t/n}^{(n,0)}(\symbfit{\varepsilon}_{0}^{(0)},\theta_0)-G_{(t-1)/n}^{(n,0)}(\symbfit{\varepsilon}_{0}^{(0)},\theta_0)}_{\mathcal{L}^q(\theta_0)}^2 
\\&
+ L^2 L_{\varphi}^2 (m+1)^2
\frac{1}{s}\sum_{r=1}^s  \sum_{t=1}^n  
\norm{G_{t/n}^{(n,r)}(\symbfit{\varepsilon}_{0}^{(r)},\theta)-G_{(t-1)/n}^{(n,r)}(\symbfit{\varepsilon}_{0}^{(r)},\theta)}_{\mathcal{L}^q(\theta)}^2,
\end{align*}
by (1) the definition of $\mu_{\frac{t}{n}}^n(\theta)$ and $\mu_{\frac{t-L}{n}}^n(\theta)$, the triangle inequality, expanding the square, Young's inequality, and the triangle inequality, (2) adding and subtracting the same intermediate lag terms between $t$ and $t-L$ (i.e., $t-1,\ldots,t-L+1$), the triangle inequality, applying Young's inequality $L$ times, factoring out $L$, adding extra lag terms so that every lag term is repeated $L$ times and then factoring out this second $L$ factor, and (3) linearity of expectation, Jensen's inequality, because the random Fourier features $\varphi$ are $L_{\varphi}$-Lipschitz with Lipschitz constant given by~\eqref{eqn:lipschitz_constant_rff}, the triangle inequality, the Minkowski inequality, and Young's inequality.

Observe that \begin{align*}&\sup_{\theta\in\Theta} \frac{2}{n-m} \sum_{t=m+\tau+L}^n  \norm{\mu_{\frac{t}{n}}^n(\theta) - \mu_{\frac{t-L}{n}}^n(\theta)}^2 \\& 
\overset{(1)}{\leq}  \frac{4L^2 L_{\varphi}^2 (m+1)^2}{n-m} \sum_{r=0}^s \sup_{\theta\in\Theta}  \sum_{t=1}^n \norm{G_{t/n}^{(n,r)}(\symbfit{\varepsilon}_{0}^{(r)},\theta)-G_{(t-1)/n}^{(n,r)}(\symbfit{\varepsilon}_{0}^{(r)},\theta)}_{\mathcal{L}^q(\theta)}^2
\\&
\overset{(2)}{\leq}
\frac{4 L^2 L_{\varphi}^2 (m+1)^2}{n-m}(s+1)n^{\max(1-2/\kappa,0)}\Lambda^2 
\\&
\overset{(3)}{=} o(1), \end{align*} by (1) the previous arguments, linearity, and subadditivity of the supremum, (2) H\"{o}lder's inequality applied to each sum over time $t$ (i.e., for each $r=0,1,\ldots,s$) so we can upper bound by the squared $\kappa$-variation multiplied by $n^{\max(1-2/\kappa,0)}$, where the max is because we consider the cases $\kappa \in [1,2]$ and $\kappa \in (2,4)$ separately, and applying the upper bound $\Lambda>0$ from Assumption~\ref{asmpt:nonstationarity_continuous_time}, and (3) the rate of growth of $L$ from Assumption~\ref{asmpt:growth_of_offset_and_window_size_continuous_time}.

\textbf{Step 2.2.} Observe that, for every $\theta\in\Theta$, we have \begin{align}\label{eqn:step_2_2}
& \frac{1}{2}\sum_{t=m+\tau+L}^n  \norm{\mu_{\frac{t-L}{n}}^n(\theta)-\hat{\mu}_{t-L,n}(\theta)}^2
\\& \nonumber
\overset{(1)}{\leq} 
\sum_{t=m+\tau+L}^n \norm{\begin{aligned}&
\mathbb{E}_{\theta_0}\varphi\left(\left[G_{(t-L-m)/n}^{(n,0)}(\symbfit{\varepsilon}_{t-L-m}^{(0)},\theta_0),\ldots,G_{(t-L)/n}^{(n,0)}(\symbfit{\varepsilon}_{t-L}^{(0)},\theta_0)\right]^{\top}\right)\\& - F_{t-L}^{\mathrm{obs}}\end{aligned}}^2\\ &+ \nonumber \sum_{r=1}^s \sum_{t=m+\tau+L}^n \norm{ \begin{aligned}&\mathbb{E}_{\theta}\varphi\left(\left[G_{(t-L-m)/n}^{(n,r)}(\symbfit{\varepsilon}_{t-L-m}^{(r)},\theta),\ldots,G_{(t-L)/n}^{(n,r)}(\symbfit{\varepsilon}_{t-L}^{(r)},\theta)\right]^{\top}\right) \\ & - F_{t-L}^{(r)}(\theta)\end{aligned}}^2,
\end{align}
by (1) the definitions of $\mu_{\frac{t-L}{n}}^n(\theta)$ and $\hat{\mu}_{t-L,n}(\theta)$, the triangle inequality, and Young's inequality, where $F_{t}^{(r)}(\theta)$ is from~\eqref{eqn:random_feature_time_average_approx_one_realization}. We have
\begin{align*}
& \underset{\theta\in\Theta}{\sup} \ \frac{2}{n-m}\sum_{t=m+\tau+L}^n  \norm{\mu_{\frac{t-L}{n}}^n(\theta)-\hat{\mu}_{t-L,n}(\theta)}^2
\\&
\overset{(1)}{\leq} 
\sum_{r=0}^s \underset{\theta\in\Theta}{\sup} \ \frac{4}{n-m}\sum_{t=m+\tau+L}^n  \norm{ \mathbb{E}_{\theta}\left[f_{t-L}^{(r)}(\theta)\right]-F_{t-L}^{(r)}(\theta)}^2,
\end{align*} by (1) the previous arguments, the definition of the random features $f_{t-L}^{(r)}(\theta)$ at time $t-L$ in terms of the generative mappings, linearity, and subadditivity of the supremum. It suffices to show that, for all $r=0,1,\ldots,s$, we have
\begin{align*}\underset{\theta\in\Theta}{\sup} \sum_{t=m+\tau+L}^n  \norm{ \mathbb{E}_{\theta}\left[f_{t-L}^{(r)}(\theta)\right]-F_{t-L}^{(r)}(\theta)}^2=o_p(n).\end{align*} 
%%% Note: this follows by similar steps to Lemma C.3 in Mies 2023 which is used to prove the rate condition in Proposition 3.1

By adding and subtracting the same term, the triangle inequality, Young's inequality, and subadditivity of the supremum, half of the previous term is less than or equal to 
\begin{align*}&\underset{\theta\in\Theta}{\sup} \sum_{t=m+\tau+L}^n  \norm{ \mathbb{E}_{\theta}\left[f_{t-L}^{(r)}(\theta)\right] -\mathbb{E}_{\theta}F_{t-L}^{(r)}(\theta)}^2\\&+\underset{\theta\in\Theta}{\sup} \sum_{t=m+\tau+L}^n  \norm{ \mathbb{E}_{\theta}F_{t-L}^{(r)}(\theta)-F_{t-L}^{(r)}(\theta)}^2,\end{align*} so it suffices to show that the first term is $o(n)$ and the second term is $o_p(n)$.

\textbf{Step 2.2.1.} For the following, denote $\mathcal{U}_t=\left[\frac{(t-L-w)\lor (m+1)}{n},\frac{t-L}{n}\right]$. Observe that, for all $r=0,1,\ldots,s$, we have
\begin{align*}&\underset{\theta\in\Theta}{\sup} \sum_{t=m+\tau+L}^n  \norm{ \mathbb{E}_{\theta}\left[f_{t-L}^{(r)}(\theta)\right]-\mathbb{E}_{\theta}F_{t-L}^{(r)}(\theta)}^2
\\& 
\overset{(1)}{\leq} 
\underset{\theta\in\Theta}{\sup} \sum_{t=m+\tau+L}^n  \underset{u\in \mathcal{U}_t}{\sup}\norm{\begin{aligned}&\mathbb{E}_{\theta}\varphi\left(\left[G_{(t-L-m)/n}^{(n,r)}(\symbfit{\varepsilon}_{t-L-m}^{(r)},\theta),\ldots,G_{(t-L)/n}^{(n,r)}(\symbfit{\varepsilon}_{t-L}^{(r)},\theta)\right]^{\top}\right)\\ &- \mathbb{E}_{\theta}\varphi\left(\left[G_{u-m/n}^{(n,r)}(\symbfit{\varepsilon}_{-m}^{(r)},\theta),\ldots,G_{u}^{(n,r)}(\symbfit{\varepsilon}_0^{(r)},\theta)\right]^{\top}\right)\end{aligned}}^2
\\&
\overset{(2)}{\leq}  L_{\varphi}^2 (m+1) \ \underset{\theta\in\Theta}{\sup} \sum_{i=0}^m \sum_{t=m+\tau+L}^n  \underset{u\in \mathcal{U}_t}{\sup}\norm{G_{(t-L-i)/n}^{(n,r)}(\symbfit{\varepsilon}_{0}^{(r)},\theta)-G_{u-i/n}^{(n,r)}(\symbfit{\varepsilon}_0^{(r)},\theta)}_{\mathcal{L}^q(\theta)}^2 \\&
\overset{(3)}{\leq} w L_{\varphi}^2 (m+1)^2  n^{\max(1-2/\kappa,0)} \left(\underset{\theta\in\Theta}{\sup} \ \norm{\left(G_u^{(n,r)}(\symbfit{\varepsilon}_0,\theta)\right)_u}_{\kappa\text{-}\mathrm{var},\mathcal{L}^q(\theta)}\right)^2 \\&
\overset{(4)}{\leq} w L_{\varphi}^2 (m+1)^2 n^{\max(1-2/\kappa,0)}  \Lambda^2 \\& \overset{(5)}{=}o(n), 
\end{align*} by
(1) the definitions of  $F_{t-L}^{(r)}(\theta)$ and $f_{t-L}^{(r)}(\theta)$, linearity, the triangle inequality, and upper bounding the rolling-window average of the distances between the expectations by the supremum of the distance between the expectations over the corresponding rescaled time interval $\mathcal{U}_t=\left[\frac{(t-L-w)\lor (m+1)}{n},\frac{t-L}{n}\right]$, (2) linearity of expectation, Jensen's inequality, because the random Fourier features $\varphi$ are $L_{\varphi}$-Lipschitz with Lipschitz constant given by~\eqref{eqn:lipschitz_constant_rff}, the triangle inequality, the Minkowski inequality, and Young's inequality, (3) upper bounding by $w$ multiplied by the squared $\kappa$-variation multiplied by $n^{\max(1-2/\kappa,0)}$, where the max is because we consider the cases $\kappa \in [1,2]$ and $\kappa \in (2,4)$ separately, and where the $w$ factor is from the following construction: each interval $\mathcal{U}_t$ has length at most $w/n$, so we consider a family of $w$ subsequences of intervals of the form $\mathcal{U}_{j+w},\mathcal{U}_{j+2w},\ldots$ so they are disjoint (besides the endpoints), then upper bound the sum over each subsequence by the squared $\kappa$-variation multiplied by $n^{\max(1-2/\kappa,0)}$, then sum over all of the subsequences, which leads to the desired upper bound with the factor of $w$, (4) using the bound from Assumption~\ref{asmpt:nonstationarity_continuous_time}, and (5) because of the rates specified in Assumption~\ref{asmpt:growth_of_offset_and_window_size_continuous_time}.

 %%%%%% Note: we will use an argument similar to Step 1 in the proof of Theorem~\ref{thm:consistency_time_average_estimator}.
\textbf{Step 2.2.2.} By Markov's inequality, it suffices to show that, for all $r=0,1,\ldots,s$,  \begin{align*}\mathbb{E}\left(\underset{\theta\in\Theta}{\sup} \sum_{t=m+\tau+L}^n  \norm{ \mathbb{E}_{\theta}F_{t-L}^{(r)}(\theta)-F_{t-L}^{(r)}(\theta)}^2\right)=o(n),
\end{align*} where $\mathbb{E}(\cdot)$ denotes expectation with respect to the law of the noise inputs. By Assumption~\ref{asmpt:statistical_model_continuous_time}, $\Theta\subset\mathbb{R}^p$ is compact, so it is totally bounded and there exists an $\epsilon$-net, i.e., for every $\epsilon>0$, there exist finitely many points $\theta_1,\ldots,\theta_{N(\epsilon)}$ for some $N(\epsilon)\in\mathbb{N}$, such that $$\Theta\subset \bigcup_{i\in [N(\epsilon)]}B_{\epsilon}(\theta_i),$$ where $B_{\epsilon}(\theta_i)$ is a ball of radius $\epsilon$. By adding and subtracting the same term, the triangle inequality, Young's inequality, the definition of the $\epsilon$-net for some $\epsilon>0$ and $N(\epsilon)\in\mathbb{N}$, and the definition of the supremum, 
\begin{align*}
&\frac{1}{2} \ \underset{\theta\in\Theta}{\sup} \sum_{t=m+\tau+L}^n\norm{\mathbb{E}_{\theta}\left[F_{t-L}^{(r)}(\theta)\right]-F_{t-L}^{(r)}(\theta)}^2 
\\&
\leq 
\underset{i \in [N(\epsilon)]}{\max} \sum_{t=m+\tau+L}^n \norm{\mathbb{E}_{\theta_i}\left[F_{t-L}^{(r)}(\theta_i)\right]-F_{t-L}^{(r)}(\theta_i)}^2 
\\ &
+
\underset{ \norm{\theta-\theta'}\leq \epsilon}{\underset{\theta,\theta'\in\Theta}{\sup}}\sum_{t=m+\tau+L}^n\norm{[\mathbb{E}_{\theta}\left[F_{t-L}^{(r)}(\theta)\right]-F_{t-L}^{(r)}(\theta)]-[\mathbb{E}_{\theta'}\left[F_{t-L}^{(r)}(\theta')\right]-F_{t-L}^{(r)}(\theta')]}^2.
\end{align*}
Taking the expectation $\mathbb{E}(\cdot)$ with respect to the law of the noise inputs, we have
\begin{align*} & \frac{1}{2} \ \mathbb{E}\left(\underset{\theta\in\Theta}{\sup} \sum_{t=m+\tau+L}^n\norm{\mathbb{E}_{\theta}\left[F_{t-L}^{(r)}(\theta)\right]-F_{t-L}^{(r)}(\theta)}^2\right)
\\ &
\overset{(1)}{\leq} \mathbb{E}\left(\underset{i \in [N(\epsilon)]}{\max} \sum_{t=m+\tau+L}^n \norm{\mathbb{E}_{\theta_i}\left[F_{t-L}^{(r)}(\theta_i)\right]-F_{t-L}^{(r)}(\theta_i)}^2\right) 
\\ &
+ \mathbb{E}\left(
\underset{ \norm{\theta-\theta'}\leq \epsilon}{\underset{\theta,\theta'\in\Theta}{\sup}}\sum_{t=m+\tau+L}^n\norm{\begin{aligned}&[\mathbb{E}_{\theta}\left[F_{t-L}^{(r)}(\theta)\right]-F_{t-L}^{(r)}(\theta)] \\ & -[\mathbb{E}_{\theta'}\left[F_{t-L}^{(r)}(\theta')\right]-F_{t-L}^{(r)}(\theta')]\end{aligned}}^2\right)
\\ & 
\overset{(2)}{\leq}
\sum_{i \in [N(\epsilon)]}\sum_{t=m+\tau+L}^n \mathbb{E}_{\theta_i}\left( \norm{\mathbb{E}_{\theta_i}\left[F_{t-L}^{(r)}(\theta_i)\right]-F_{t-L}^{(r)}(\theta_i)}^2\right) 
\\ &
+ \mathbb{E}\left(
\underset{ \norm{\theta-\theta'}\leq \epsilon}{\underset{\theta,\theta'\in\Theta}{\sup}}\sum_{t=m+\tau+L}^n\norm{\begin{aligned}&[\mathbb{E}_{\theta}\left[F_{t-L}^{(r)}(\theta)\right]-F_{t-L}^{(r)}(\theta)]\\&-[\mathbb{E}_{\theta'}\left[F_{t-L}^{(r)}(\theta')\right]-F_{t-L}^{(r)}(\theta')]\end{aligned}}^2\right)
\\ & 
\overset{(3)}{\leq}
\sum_{i \in [N(\epsilon)]}\sum_{t=m+\tau+L}^n \mathbb{E}_{\theta_i}\left( \norm{\mathbb{E}_{\theta_i}\left[F_{t-L}^{(r)}(\theta_i)\right]-F_{t-L}^{(r)}(\theta_i)}^2\right) 
\\ &
+ 4 \sum_{t=m+\tau+L}^n \mathbb{E}\left(
\underset{ \norm{\theta-\theta'}\leq \epsilon}{\underset{\theta,\theta'\in\Theta}{\sup}}\norm{F_{t-L}^{(r)}(\theta)-F_{t-L}^{(r)}(\theta')}^2\right),
\end{align*}
by (1) the previous $\epsilon$-net inequality, monotonicity of expectation, and linearity of expectation, (2) upper bounding the maximum by the sum and linearity of expectation, and (3) the triangle inequality, Young's inequality, subadditivity of the supremum, linearity of expectation, Jensen's inequality, and monotonicity of expectation.
%%% Note: last monotonicity of expectation is to bring the sup inside the expectation 

%%% Note: this follows by similar arguments as the second part of Lemma C.3 in Mies 2023, but where we don't use the univariate arguments we use the multivariate arguments based on the trace of the covariance as in Mies 2025 random multipliers paper.
\textbf{Step 2.2.2.1.} We begin with the first term. For the following derivations, define $$\mathcal{T}_t=\{(t-L-w)\lor (m+1),\ldots,t-L-1,t-L\}.$$ For all $i\in [N(\epsilon)]$, we have \begin{align*}
& \mathbb{E}_{\theta_i}\left( \norm{\mathbb{E}_{\theta_i}\left[F_{t-L}^{(r)}(\theta_i)\right]-F_{t-L}^{(r)}(\theta_i)}^2\right) 
\\ & \overset{(1)}{=}\mathrm{tr}\left[\mathrm{Cov}_{\theta_i}\left(F_{t-L}^{(r)}(\theta_i)\right)\right]
\\ & \overset{(2)} = \mathrm{tr}\left[\mathrm{Var}_{\theta_i}\left(\begin{aligned}&\frac{1}{N_{t-L}}\sum_{j=(t-L-w)\lor (m+1)}^{t-L} \\&\varphi\left(\left[G_{(j-m)/n}^{(n,r)}(\symbfit{\varepsilon}_{j-m}^{(r)},\theta_i),\ldots,G_{j/n}^{(n,r)}(\symbfit{\varepsilon}_j^{(r)},\theta_i)\right]^{\top}\right)\end{aligned}\right)\right]
\\ & \overset{(3)} = \frac{1}{N_{t-L}^2}\sum_{\ell,j \in \mathcal{T}_t} 
\mathrm{tr}\left[\mathrm{Cov}_{\theta_i}\left(\begin{aligned}&\varphi\left(\left[G_{(\ell-m)/n}^{(n,r)}(\symbfit{\varepsilon}_{\ell-m}^{(r)},\theta_i),\ldots,G_{\ell/n}^{(n,r)}(\symbfit{\varepsilon}_{\ell}^{(r)},\theta_i)\right]^{\top}\right), \\& \varphi\left(\left[G_{(j-m)/n}^{(n,r)}(\symbfit{\varepsilon}_{j-m}^{(r)},\theta_i),\ldots,G_{j/n}^{(n,r)}(\symbfit{\varepsilon}_j^{(r)},\theta_i)\right]^{\top}\right)\end{aligned}\right)\right]
\\&
\overset{(4)}{\leq} \frac{1}{N_{t-L}^2}\sum_{\ell,j\in\mathcal{T}_t} C^{\varphi}(|\ell-j|+1)^{-K^{\varphi}}
\\&
\overset{(5)}{\leq} \frac{1}{N_{t-L}}\sum_{h=-\infty}^{\infty} C^{\varphi}(|h|+1)^{-K^{\varphi}},
\end{align*} by (1) the equality of the inner product with the trace of the outer product, and the definition of the covariance of this quantity as this outer product, (2) the definition of $F_{t-L}^{(r)}(\theta_i)$ from~\eqref{eqn:random_feature_time_average_approx_one_realization}, (3) expanding the trace of the covariance of the sum of random vectors by using the linearity of covariance and trace, (4) upper bounding the trace of the covariance by its trace norm and applying the trace norm inequality with constants $C^{\varphi}>0$, $K^{\varphi}>1$ from Lemma~\ref{lma:temporal_dependence_of_random_features_continuous_time_implies_covariance_decay}, (5) by enlarging the inner summation to be over all integers and upper bounding by this quantity multiplied by the number of terms in the outer summation, noting that the series converges because $K^{\varphi}>1$ by Lemma~\ref{lma:temporal_dependence_of_random_features_continuous_time_implies_covariance_decay}.   
%%% Note on (1): because inner product is trace of outer product e.g. see https://en.wikipedia.org/wiki/Trace_(linear_algebra)
%%% Note on (4): The Lemma uses trace norm of autocovariance with random features in the concise notation, which upper bounds the trace of autocovariance with random features in the verbose notation,in particular because tr(cov) \leq tr((cov^T cov)^1/2)=tr((cov^2)^1/2)=tr(cov) because the autocovariance is not positive semitifinite (in the positive semidefinite case you get equality, because singular values equal eigenvalues). this is because eigenvalues can be negative while singular values are nonnegative so sum of eigenvalues (trace) is less than sum of singular values (trace norm) 

Denoting $\sum_{h=-\infty}^{\infty} C^{\varphi}(|h|+1)^{-K^{\varphi}}$ by the constant $C'>0$, since the series converges, summing over $t$ yields
\begin{align*}
    & \sum_{t=m+\tau+L}^n \mathbb{E}_{\theta_i}\left( \norm{\mathbb{E}_{\theta_i}\left[F_{t-L}^{(r)}(\theta_i)\right]-F_{t-L}^{(r)}(\theta_i)}^2\right) 
    \\ & 
    \overset{(1)}{\leq}
    C'\sum_{t=m+\tau+L}^n \frac{1}{N_{t-L}}  
    \\ & 
    \overset{(2)}{\leq}C'\frac{n}{w+1}  + C'
    \sum_{t=1}^{w+1} \frac{1}{t}
    \\ & 
    \overset{(3)}{\leq}C'\frac{n}{w}  + C'
    (1+\log(w+1))
    \\& \overset{(4)}{=} O(\max(n/w,\log(w+1))),
\end{align*} by (1) applying the previous inequality, (2) expanding the sum with the minimum, where the inequality covers both cases, (3) upper bounding the harmonic series and further upper bounding by dropping some addition and subtraction terms, (4) because $w_n=o(n)$ by Assumption~\ref{asmpt:growth_of_offset_and_window_size_continuous_time}.

\textbf{Step 2.2.2.2.} We now consider the second term. Observe that  \begin{align*}
& \sum_{t=m+\tau+L}^n \mathbb{E}\left(
\underset{ \norm{\theta-\theta'}\leq \epsilon}{\underset{\theta,\theta'\in\Theta}{\sup}}\norm{F_{t-L}^{(r)}(\theta)-F_{t-L}^{(r)}(\theta')}^2\right)
\\&
\overset{(1)}{\leq} 2\sum_{t=m+\tau+L}^n
\frac{1}{N_{t-L}}\sum_{\ell=(t-L-w)\lor (m+1)}^{t-L} \\& \mathbb{E}\left(
\underset{ \norm{\theta-\theta'}\leq \epsilon}{\underset{\theta,\theta'\in\Theta}{\sup}}\norm{\begin{aligned}&\varphi\left(\left[G_{(\ell-m)/n}^{(n,r)}(\symbfit{\varepsilon}_{\ell-m}^{(r)},\theta),\ldots,G_{\ell/n}^{(n,r)}(\symbfit{\varepsilon}_{\ell}^{(r)},\theta)\right]^{\top}\right)\\&- \varphi\left(\left[G_{(\ell-m)/n}^{(n,r)}(\symbfit{\varepsilon}_{\ell-m}^{(r)},\theta'),\ldots,G_{\ell/n}^{(n,r)}(\symbfit{\varepsilon}_{\ell}^{(r)},\theta')\right]^{\top}\right)\end{aligned}}^2\right)
\\& \overset{(2)}{\leq}
2 (m+1) L_{\varphi}^2 
\sum_{t=m+\tau+L}^n \frac{1}{N_{t-L}}\sum_{\ell=(t-L-w)\lor (m+1)}^{t-L}  \sum_{j=0}^m  \\&  \mathbb{E}\left(
\underset{ \norm{\theta-\theta'}\leq \epsilon}{\underset{\theta,\theta'\in\Theta}{\sup}} \norm{
G_{(\ell-j)/n}^{(n,r)}(\symbfit{\varepsilon}_{\ell-j}^{(r)},\theta) - 
G_{(\ell-j)/n}^{(n,r)}(\symbfit{\varepsilon}_{\ell-j}^{(r)},\theta')}^2\right)
\\& \overset{(3)}{\leq}
2 n (m+1)^2 L_{\varphi}^2 
     \eta_G,
\end{align*} because (1) the definition of $F_{t-L}^{(r)}(\theta)$ from~\eqref{eqn:random_feature_time_average_approx_one_realization}, the Cauchy--Schwarz inequality, the subadditivity of the supremum, and linearity of expectation, (2) the random Fourier features $\varphi$ are $L_{\varphi}$-Lipschitz with Lipschitz constant given by~\eqref{eqn:lipschitz_constant_rff}, the triangle inequality, Young's inequality $m+1$ times, subadditivity of the supremum, and linearity of expectation, (3) applying the stochastic equicontinuity-type condition from Assumption~\ref{asmpt:stochastic_equicontinuity_continuous_time} and noting that each rolling window contains exactly $N_{t-L}$ terms. By Assumption~\ref{asmpt:stochastic_equicontinuity_continuous_time}, $\eta_G \xrightarrow[]{} 0$ as $\epsilon \xrightarrow[]{} 0$. However, taking $\epsilon \xrightarrow[]{} 0$ makes $N(\epsilon)\xrightarrow[]{}\infty$, so we must specify the rates.

Previously, in Step 2.2.2.1 we showed that, for each fixed $\epsilon>0$ and $N(\epsilon)\in\mathbb{N}$, for all $i\in [N(\epsilon)]$, we have $$\sum_{t=m+\tau+L}^n \mathbb{E}_{\theta_i}\left( \norm{\mathbb{E}_{\theta_i}\left[F_{t-L}^{(r)}(\theta_i)\right]-F_{t-L}^{(r)}(\theta_i)}^2\right)=O(\max(n/w,\log(w+1))).$$ Therefore, for any sequence $(\epsilon_n)_{n\in\mathbb{N}}$, $\epsilon_n\xrightarrow[]{}0$ as $n\xrightarrow[]{}\infty$, such that $$N(\epsilon_n)=o\left(\frac{n}{\max(n/w,\log(w+1))}\right),$$ we have
\begin{align*}&\sum_{i\in [N(\epsilon_n)]} \sum_{t=m+\tau+L}^n \mathbb{E}_{\theta_i}\left( \norm{\mathbb{E}_{\theta_i}\left[F_{t-L}^{(r)}(\theta_i)\right]-F_{t-L}^{(r)}(\theta_i)}^2\right) 
\\&
\leq N(\epsilon_n) \underset{i \in [N(\epsilon_n)]}{\max} \sum_{t=m+\tau+L}^n \mathbb{E}_{\theta_i}\left( \norm{\mathbb{E}_{\theta_i}\left[F_{t-L}^{(r)}(\theta_i)\right]-F_{t-L}^{(r)}(\theta_i)}^2\right)
\\& 
=o(n),
\end{align*} and $\eta_G = \eta_G(\epsilon_n)=o(1)$ so that the upper bound from (3) is $2 n (m+1)^2 L_{\varphi}^2 \eta_G=o(n)$. See, e.g., Corollary 4.2.11 in \cite{Vershynin} for the explicit covering number bounds.

Putting everything together, by Markov's inequality, we have the desired result \begin{align*}\underset{\theta\in\Theta}{\sup} \sum_{t=m+\tau+L}^n  \norm{ \mathbb{E}_{\theta}F_{t-L}^{(r)}(\theta)-F_{t-L}^{(r)}(\theta)}^2=o_p(n).
\end{align*}

\textbf{Step 3.} We show that \begin{align*}\sup_{\theta\in\Theta}\abs{I_n^3(\theta)}=o_p(1).\end{align*}

Observe that
\begin{align*}
&\sup_{\theta\in\Theta}\abs{I_n^3(\theta)} 
\\& 
\overset{(1)}{=} \sup_{\theta\in\Theta}\abs{\frac{2}{n-m}\sum_{t=m+\tau+L}^n \hat{\mu}_{t-L,n}(\theta)^{\top}\left(\left[f_t^{\mathrm{obs}}-\bar{f}_{t}^{\mathrm{sim}}(\theta)\right]-\mu_{\frac{t}{n}}^n(\theta)\right)}
\\& 
\overset{(2)}{\leq} \sup_{\theta\in\Theta}\abs{\frac{2}{n-m}\sum_{t=m+\tau+L}^n \mu_{\frac{t-L}{n}}^n(\theta)^{\top}\left(\left[f_t^{\mathrm{obs}}-\bar{f}_{t}^{\mathrm{sim}}(\theta)\right]-\mu_{\frac{t}{n}}^n(\theta)\right)}
\\&
+
\sup_{\theta\in\Theta}\abs{\frac{2}{n-m}\sum_{t=m+\tau+L}^n \left(\hat{\mu}_{t-L,n}(\theta)-\mu_{\frac{t-L}{n}}^n(\theta)\right)^{\top}\left(\left[f_t^{\mathrm{obs}}-\bar{f}_{t}^{\mathrm{sim}}(\theta)\right]-\mu_{\frac{t}{n}}^n(\theta)\right)},
\end{align*}
by (1) the definition of $I_n^3(\theta)$, and (2) adding and subtracting the same term, linearity in the first argument of the inner product, the triangle inequality, and the subadditivity of the supremum. We will show that both terms are $o_p(1)$.

\textbf{Step 3.1.} By Markov's inequality, it suffices to show that
$$\mathbb{E}\left(\sup_{\theta\in\Theta}\abs{\frac{2}{n-m}\sum_{t=m+\tau+L}^n \mu_{\frac{t-L}{n}}^n(\theta)^{\top}\left(\left[f_t^{\mathrm{obs}}-\bar{f}_{t}^{\mathrm{sim}}(\theta)\right]-\mu_{\frac{t}{n}}^n(\theta)\right)}\right)=o(1),$$ where $\mathbb{E}(\cdot)$ denotes expectation with respect to the law of the noise inputs.  We have
\begin{align*} 
&\sup_{\theta\in\Theta}\abs{\frac{2}{n-m}\sum_{t=m+\tau+L}^n \mu_{\frac{t-L}{n}}^n(\theta)^{\top}\left(\left[f_t^{\mathrm{obs}}-\bar{f}_{t}^{\mathrm{sim}}(\theta)\right]-\mu_{\frac{t}{n}}^n(\theta)\right)}
\\&
\leq \frac{2}{n-m}
\sup_{\theta\in\Theta} \left|\sum_{t=m+\tau+L}^n  \mu_{\frac{t-L}{n}}^n(\theta)^{\top}\left[f_t^{\mathrm{obs}} -\mathbb{E}_{\theta_0}[f_t^{\mathrm{obs}}]\right]\right|
\\&
+ \frac{2}{s(n-m)}\sum_{r=1}^s 
\sup_{\theta\in\Theta} \left|\sum_{t=m+\tau+L}^n  \mu_{\frac{t-L}{n}}^n(\theta)^{\top}\left[f_t^{(r)}(\theta)-\mathbb{E}_{\theta}[f_t^{(r)}(\theta)]\right]\right|,
\end{align*} by linearity, the definition of the terms, the triangle inequality, and the subadditivity of the supremum.

\textbf{Step 3.1.1.} We begin with the second term. By linearity and monotonicity of expectation, it suffices to show that, for all $r=1,\ldots,s$, we have
$$\mathbb{E}\left(\sup_{\theta\in\Theta} \left|\sum_{t=m+\tau+L}^n  \mu_{\frac{t-L}{n}}^n(\theta)^{\top}\left[f_t^{(r)}(\theta)-\mathbb{E}_{\theta}[f_t^{(r)}(\theta)]\right]\right|\right)=o(n).$$

By Assumption~\ref{asmpt:statistical_model_continuous_time}, $\Theta\subset\mathbb{R}^p$ is compact, so it is totally bounded and there exists an $\epsilon$-net, i.e., for every $\epsilon>0$, there exist finitely many points $\theta_1,\ldots,\theta_{N(\epsilon)}$ for some $N(\epsilon)\in\mathbb{N}$, such that $$\Theta\subset \bigcup_{i\in [N(\epsilon)]}B_{\epsilon}(\theta_i),$$ where $B_{\epsilon}(\theta_i)$ is a ball of radius $\epsilon$. By adding and subtracting the same term, the triangle inequality, the definition of the $\epsilon$-net for some $\epsilon>0$ and $N(\epsilon)\in\mathbb{N}$, and the definition of the supremum, 
\begin{align*}
    &  \sup_{\theta\in\Theta} \left|\sum_{t=m+\tau+L}^n  \mu_{\frac{t-L}{n}}^n(\theta)^{\top}\left[f_t^{(r)}(\theta)-\mathbb{E}_{\theta}[f_t^{(r)}(\theta)]\right]\right|
    \\& 
    \overset{(1)}{\leq}
    \max_{i \in [N(\epsilon)]} \left|\sum_{t=m+\tau+L}^n  \mu_{\frac{t-L}{n}}^n(\theta_i)^{\top}\left[f_t^{(r)}(\theta_i)-\mathbb{E}_{\theta_i}[f_t^{(r)}(\theta_i)]\right]\right|
    \\ &
    + \underset{ \norm{\theta-\theta'}\leq \epsilon}{\underset{\theta,\theta'\in\Theta}{\sup}}\abs{\sum_{t=m+\tau+L}^n \left[\begin{aligned}&\mu_{\frac{t-L}{n}}^n(\theta)^{\top}\left[f_t^{(r)}(\theta)-\mathbb{E}_{\theta}[f_t^{(r)}(\theta)]\right]\\& - \mu_{\frac{t-L}{n}}^n(\theta')^{\top}\left[f_t^{(r)}(\theta')-\mathbb{E}_{\theta'}[f_t^{(r)}(\theta')]\right]\end{aligned}\right]}.
\end{align*}
Taking the expectation $\mathbb{E}(\cdot)$ with respect to the law of the noise inputs, we have
\begin{align*}
    & \mathbb{E}\left( \sup_{\theta\in\Theta} \left|\sum_{t=m+\tau+L}^n  \mu_{\frac{t-L}{n}}^n(\theta)^{\top}\left[f_t^{(r)}(\theta)-\mathbb{E}_{\theta}[f_t^{(r)}(\theta)]\right]\right|\right)
    \\& 
    \overset{(1)}{\leq}
    \mathbb{E}\left(\max_{i \in [N(\epsilon)]} \left|\sum_{t=m+\tau+L}^n  \mu_{\frac{t-L}{n}}^n(\theta_i)^{\top}\left[f_t^{(r)}(\theta_i)-\mathbb{E}_{\theta_i}[f_t^{(r)}(\theta_i)]\right]\right|\right)
    \\ &
    + \mathbb{E}\left[\underset{ \norm{\theta-\theta'}\leq \epsilon}{\underset{\theta,\theta'\in\Theta}{\sup}}\abs{\sum_{t=m+\tau+L}^n \left[\begin{aligned}&\mu_{\frac{t-L}{n}}^n(\theta)^{\top}\left[f_t^{(r)}(\theta)-\mathbb{E}_{\theta}[f_t^{(r)}(\theta)]\right]\\& - \mu_{\frac{t-L}{n}}^n(\theta')^{\top}\left[f_t^{(r)}(\theta')-\mathbb{E}_{\theta'}[f_t^{(r)}(\theta')]\right]\end{aligned}\right]}\right]
        \\& 
    \overset{(2)}{\leq}
    \sum_{i \in [N(\epsilon)]} \mathbb{E}_{\theta_i}\left|\sum_{t=m+\tau+L}^n  \mu_{\frac{t-L}{n}}^n(\theta_i)^{\top}\left[f_t^{(r)}(\theta_i)-\mathbb{E}_{\theta_i}[f_t^{(r)}(\theta_i)]\right]\right|
    \\ &
    + \mathbb{E}\left[\underset{ \norm{\theta-\theta'}\leq \epsilon}{\underset{\theta,\theta'\in\Theta}{\sup}}\abs{\sum_{t=m+\tau+L}^n \left[\begin{aligned}&\mu_{\frac{t-L}{n}}^n(\theta)^{\top}\left[f_t^{(r)}(\theta)-\mathbb{E}_{\theta}[f_t^{(r)}(\theta)]\right]\\& - \mu_{\frac{t-L}{n}}^n(\theta')^{\top}\left[f_t^{(r)}(\theta')-\mathbb{E}_{\theta'}[f_t^{(r)}(\theta')]\right]\end{aligned}\right]}\right]
       \\& 
    \overset{(3)}{\leq}
    \sum_{i \in [N(\epsilon)]} \mathbb{E}_{\theta_i}\left|\sum_{t=m+\tau+L}^n  \mu_{\frac{t-L}{n}}^n(\theta_i)^{\top}\left[f_t^{(r)}(\theta_i)-\mathbb{E}_{\theta_i}[f_t^{(r)}(\theta_i)]\right]\right|
    \\ &
    + \mathbb{E}\left[\underset{ \norm{\theta-\theta'}\leq \epsilon}{\underset{\theta,\theta'\in\Theta}{\sup}}\sum_{t=m+\tau+L}^n \norm{\mu_{\frac{t-L}{n}}^n(\theta)-\mu_{\frac{t-L}{n}}^n(\theta')} \norm{f_t^{(r)}(\theta)-\mathbb{E}_{\theta}[f_t^{(r)}(\theta)]}\right]   
    \\ &
    + \mathbb{E}\left[\underset{ \norm{\theta-\theta'}\leq \epsilon}{\underset{\theta,\theta'\in\Theta}{\sup}}\sum_{t=m+\tau+L}^n \norm{\mu_{\frac{t-L}{n}}^n(\theta')} \norm{\begin{aligned}&(f_t^{(r)}(\theta)-f_t^{(r)}(\theta')) \\& -(\mathbb{E}_{\theta}[f_t^{(r)}(\theta)]-\mathbb{E}_{\theta'}[f_t^{(r)}(\theta')])\end{aligned}}\right],       
\end{align*} by (1) the previous $\epsilon$-net inequality, monotonicity of expectation, and linearity of expectation, (2) upper bounding the maximum by the sum and linearity of expectation, (3)  adding and subtracting the cross-term with $\theta$ and $\theta'$, the triangle inequality, the Cauchy-Schwarz inequality, and the subadditivity of the supremum. Continuing on from the quantity from (3), we have
\begin{align*}
& \overset{(4)}{\leq}
    \sum_{i \in [N(\epsilon)]} \mathbb{E}_{\theta_i}\left|\sum_{t=m+\tau+L}^n  \mu_{\frac{t-L}{n}}^n(\theta_i)^{\top}\left[f_t^{(r)}(\theta_i)-\mathbb{E}_{\theta_i}[f_t^{(r)}(\theta_i)]\right]\right|
    \\ &
    + 2\sqrt{k}\underset{ \norm{\theta-\theta'}\leq \epsilon}{\underset{\theta,\theta'\in\Theta}{\sup}}\sum_{t=m+\tau+L}^n \norm{\mu_{\frac{t-L}{n}}^n(\theta)-\mu_{\frac{t-L}{n}}^n(\theta')}    
    \\ &
    + 4\sqrt{k}\mathbb{E}\left[\underset{ \norm{\theta-\theta'}\leq \epsilon}{\underset{\theta,\theta'\in\Theta}{\sup}}\sum_{t=m+\tau+L}^n \norm{f_t^{(r)}(\theta)-f_t^{(r)}(\theta')}\right]    
       \\& 
    \overset{(5)}{\leq}
    \sum_{i \in [N(\epsilon)]} \mathbb{E}_{\theta_i}\left|\sum_{t=m+\tau+L}^n  \mu_{\frac{t-L}{n}}^n(\theta_i)^{\top}\left[f_t^{(r)}(\theta_i)-\mathbb{E}_{\theta_i}[f_t^{(r)}(\theta_i)]\right]\right|
    \\& + 6\sqrt{k}\underset{1\leq r \leq s}{\max}\mathbb{E}\left[\underset{ \norm{\theta-\theta'}\leq \epsilon}{\underset{\theta,\theta'\in\Theta}{\sup}}\sum_{t=m+\tau}^n \norm{f_t^{(r)}(\theta)-f_t^{(r)}(\theta')}\right],
\end{align*} by (4) upper bounding because the random features are bounded between $-1$ and $1$ so the expectations must be as well, and for the last term: the triangle inequality, linearity, subadditivity of the supremum, linearity of expectation, Jensen's inequality, and monotonicity of expectation, and (5) linearity of expectation, Jensen's inequality, the triangle inequality, the Minkowski inequality, and subadditivity of the supremum. Note that in the second term, in the sum over $t=m+\tau,\ldots,n$, we sum over more terms than in the previous line. Next, we show both terms are $o(n)$.
%%% Note: last monotonicity of expectation is to bring the sup inside the expectation 

\textbf{Step 3.1.1.1.} We begin with the first term. For each $r=1,\ldots,s$, and $\theta\in\Theta$, denote $$Y_t^{(r)}(\theta)=\mu_{\frac{t-L}{n}}^n(\theta)^{\top}\left[f_t^{(r)}(\theta)-\mathbb{E}_{\theta}[f_t^{(r)}(\theta)]\right],\quad t=m+\tau+L,\ldots,n,$$
which admits the representation
$$Y_t^{(r)}(\theta)=H_{t/n}^{(n,r)}(\symbfit{\varepsilon}_t^{(r)},\theta),$$ where the mapping $H_{t/n}^{(n,r)}(\symbfit{\varepsilon}_t^{(r)},\theta)$ is defined in terms of the original causal representations from Assumption~\ref{asmpt:algorithmic_dynamic_model_continuous_time}, $$H_{t/n}^{(n,r)}(\symbfit{\varepsilon}_t^{(r)},\theta)=\mu_{\frac{t-L}{n}}^n(\theta)^{\top}\left[\begin{aligned}&\varphi\left(\left[G_{(t-m)/n}^{(n,r)}(\symbfit{\varepsilon}_{t-m}^{(r)},\theta),\ldots,G_{t/n}^{(n,r)}(\symbfit{\varepsilon}_t^{(r)},\theta)\right]^{\top}\right)\\&-\mathbb{E}_{\theta}[f_t^{(r)}(\theta)]\end{aligned}\right].$$
Analogous to Assumption~\ref{asmpt:temporal_dependence_continuous_time}, denote the version of $Y_t^{(r)}(\theta)$ with the $j$-th noise input in the past replaced with its iid copy by
\begin{align*}\tilde{Y}_{t,j}^{(r)}(\theta)&=H_{t/n}^{(n,r)}(\tilde{\symbfit{\varepsilon}}_{t,j}^{(r)},\theta)\\&=\mu_{\frac{t-L}{n}}^n(\theta)^{\top}\left[\begin{aligned}&\varphi\left(\left[G_{(t-m)/n}^{(n,r)}(\tilde{\symbfit{\varepsilon}}_{t-m,j-m}^{(r)},\theta),\ldots,G_{t/n}^{(n,r)}(\tilde{\symbfit{\varepsilon}}_{t,j}^{(r)},\theta)\right]^{\top}\right)\\&-\mathbb{E}_{\theta}[f_t^{(r)}(\theta)]\end{aligned}\right],\end{align*}
where for any $i\in\mathbb{N}$ we define $\tilde{\symbfit{\varepsilon}}_{t-i,j-i}^{(r)}$ as $\symbfit{\varepsilon}_{t-i}^{(r)}$ if $j-i < 0$. Observe that the time series $Y_t^{(r)}(\theta)$ satisfies a decaying physical dependence measure condition
\begin{align*}
&\norm{Y_t^{(r)}(\theta)-\tilde{Y}_{t,j}^{(r)}(\theta)}_{\mathcal{L}^q(\theta)}
\\ &
\overset{(1)}{=}
\norm{H_{t/n}^{(n,r)}(\symbfit{\varepsilon}_{t}^{(r)},\theta)-H_{t/n}^{(n,r)}(\tilde{\symbfit{\varepsilon}}_{t,j}^{(r)},\theta)}_{\mathcal{L}^q(\theta)} 
\\&
\overset{(2)}{=}
\norm{\mu_{\frac{t-L}{n}}^n(\theta)^{\top}\left[\begin{aligned}
    & \varphi\left(\left[G_{(t-m)/n}^{(n,r)}(\symbfit{\varepsilon}_{t-m}^{(r)},\theta),\ldots,G_{t/n}^{(n,r)}(\symbfit{\varepsilon}_{t}^{(r)},\theta)\right]^{\top}\right) \\ & - \varphi\left(\left[G_{(t-m)/n}^{(n,r)}(\tilde{\symbfit{\varepsilon}}_{t-m,j-m}^{(r)},\theta),\ldots,G_{t/n}^{(n,r)}(\tilde{\symbfit{\varepsilon}}_{t,j}^{(r)},\theta)\right]^{\top}\right)\end{aligned}\right]
}_{\mathcal{L}^q(\theta)}
\\&
\overset{(3)}{\leq}
\norm{\mu_{\frac{t-L}{n}}^n(\theta)}\norm{\begin{aligned}
    & \varphi\left(\left[G_{(t-m)/n}^{(n,r)}(\symbfit{\varepsilon}_{t-m}^{(r)},\theta),\ldots,G_{t/n}^{(n,r)}(\symbfit{\varepsilon}_{t}^{(r)},\theta)\right]^{\top}\right) \\ & - \varphi\left(\left[G_{(t-m)/n}^{(n,r)}(\tilde{\symbfit{\varepsilon}}_{t-m,j-m}^{(r)},\theta),\ldots,G_{t/n}^{(n,r)}(\tilde{\symbfit{\varepsilon}}_{t,j}^{(r)},\theta)\right]^{\top}\right)\end{aligned}}_{\mathcal{L}^q(\theta)}
\\&
\overset{(4)}{\leq} 2\sqrt{k} L_{\varphi} \sum_{i=0}^m
\norm{
G_{(t-i)/n}^{(n,r)}(\symbfit{\varepsilon}_{t-i}^{(r)},\theta)  - G_{(t-i)/n}^{(n,r)}(\tilde{\symbfit{\varepsilon}}_{t-i,j-i}^{(r)},\theta) }_{\mathcal{L}^q(\theta)}
\\&
\overset{(5)}{=} 
2\sqrt{k} L_{\varphi} \sum_{i=0}^{\min(m,j)}
\norm{
G_{(t-i)/n}^{(n,r)}(\symbfit{\varepsilon}_{t-i}^{(r)},\theta)  - G_{(t-i)/n}^{(n,r)}(\tilde{\symbfit{\varepsilon}}_{t-i,j-i}^{(r)},\theta) }_{\mathcal{L}^q(\theta)}
\\&
\overset{(6)}{\leq} 2\sqrt{k} L_{\varphi} \sum_{i=0}^{\min(m,j)}
\Lambda \rho^{j-i}
\\&
\overset{(7)}{\leq} 2\sqrt{k} L_{\varphi} (m+1)
\Lambda \rho^{j-\min(m,j)}
\\&
\overset{(8)}{\leq} \Lambda^H (\rho^H)^j,
\end{align*} by (1) definition, (2) cancellation of terms and the linearity of the inner product, (3) the Cauchy-Schwarz inequality and linearity of expectation, (4) because the random Fourier features are bounded between $-1$ and $1$ so the expectations must be as well, and because the random Fourier features $\varphi$ are $L_{\varphi}$-Lipschitz with Lipschitz constant given by~\eqref{eqn:lipschitz_constant_rff}, the triangle inequality, and the Minkowski inequality, (5) because the distances are zero for $j-i< 0$, so it suffices to consider the terms where $j-i \geq 0 \iff i \leq j$, (6) using the bound on the physical dependence measure from Assumption~\ref{asmpt:temporal_dependence_continuous_time}, (7) upper bounding by $m+1$ times the largest term which is the ``earliest'' term because $\rho\in (0,1)$, and (8) upper bounding with an expression that satisfies exponential decay condition analogous to Assumption~\ref{asmpt:temporal_dependence_continuous_time}, where
$\Lambda^H = 2\sqrt{k} L_{\varphi} (m+1)
\Lambda \rho^{-m} >0$ and $\rho^H = \rho\in (0,1)$. Lemma~\ref{lma:lp_lln_continuous_time} is therefore applicable to the time series $Y_t^{(r)}(\theta)$.

Let us return to the term of interest. For all $i\in [N(\epsilon)]$, we have 
\begin{align*}
&\mathbb{E}_{\theta_i}\left|\sum_{t=m+\tau+L}^n  \mu_{\frac{t-L}{n}}^n(\theta_i)^{\top}\left[f_t^{(r)}(\theta_i)-\mathbb{E}_{\theta_i}[f_t^{(r)}(\theta_i)]\right]\right|
\\& 
\overset{(1)}{=}
\mathbb{E}_{\theta_i}\left|\sum_{t=m+\tau+L}^n  Y_t^{(r)}(\theta_i)\right|
\\&
\overset{(2)}{\leq}
C n^{\frac{1}{2}} \Lambda^H \sum_{j=0}^{\infty}  (\rho^H)^j
\\&
\overset{(3)}{=}
C n^{\frac{1}{2}} \Lambda^H K^H,
\end{align*} by (1) definition, (2) the second inequality from Lemma~\ref{lma:lp_lln_continuous_time}, noting that $\mathbb{E}_{\theta}[Y_t^{(r)}(\theta)]=0$ for all $t$, $r$, $\theta$, by construction, applying the inequality on the physical dependence measure derived above, and further upper bounding with infinitely many terms in the summation, (3) noting that the geometric series converges and the constant $K^H=\frac{1}{1-\rho^H}>0$ since $\rho^H=\rho\in (0,1)$ by Assumption~\ref{asmpt:temporal_dependence_continuous_time}. Hence, for each fixed $\epsilon>0$ and $N(\epsilon)\in\mathbb{N}$, for all $i\in [N(\epsilon)]$, we have
$$\mathbb{E}_{\theta_i}\left|\sum_{t=m+\tau+L}^n  \mu_{\frac{t-L}{n}}^n(\theta_i)^{\top}\left[f_t^{(r)}(\theta_i)-\mathbb{E}_{\theta_i}[f_t^{(r)}(\theta_i)]\right]\right|=O(n^{\frac{1}{2}}).$$

\textbf{Step 3.1.1.2.} We now consider the second term. Observe that
\begin{align*}
    & \underset{1\leq r \leq s}{\max}\mathbb{E}\left[\underset{ \norm{\theta-\theta'}\leq \epsilon}{\underset{\theta,\theta'\in\Theta}{\sup}}\sum_{t=m+\tau}^n \norm{f_t^{(r)}(\theta)-f_t^{(r)}(\theta')}\right]
\\ & 
\overset{(1)}{\leq}
L_{\varphi} 
\sum_{t=m+\tau}^n \sum_{j=0}^m   \underset{1\leq r \leq s}{\max} \mathbb{E}\left(
\underset{ \norm{\theta-\theta'}\leq \epsilon}{\underset{\theta,\theta'\in\Theta}{\sup}} \norm{
G_{(t-j)/n}^{(n,r)}(\symbfit{\varepsilon}_{t-j}^{(r)},\theta) - 
G_{(t-j)/n}^{(n,r)}(\symbfit{\varepsilon}_{t-j}^{(r)},\theta')}\right) 
\\& 
\overset{(2)}{\leq} L_{\varphi} n (m+1)  \sqrt{\eta_G},
\end{align*} because (1) the random Fourier features $\varphi$ are $L_{\varphi}$-Lipschitz with Lipschitz constant given by~\eqref{eqn:lipschitz_constant_rff}, the triangle inequality, subadditivity of the supremum, and linearity of expectation, and (2)  by upper bounding by $n$ times the supremum over $u\in [0,1]$, the Cauchy-Schwarz inequality, and applying the stochastic equicontinuity-type condition from Assumption~\ref{asmpt:stochastic_equicontinuity_continuous_time}. By Assumption~\ref{asmpt:stochastic_equicontinuity_continuous_time}, $\eta_G \xrightarrow[]{} 0$ as $\epsilon \xrightarrow[]{} 0$. However, taking $\epsilon \xrightarrow[]{} 0$ makes $N(\epsilon)\xrightarrow[]{}\infty$, so we must specify the rates.

Previously, we showed that, for each fixed $\epsilon>0$ and $N(\epsilon)\in\mathbb{N}$, for all $i\in [N(\epsilon)]$, we have
$$\mathbb{E}_{\theta_i}\left|\sum_{t=m+\tau+L}^n  \mu_{\frac{t-L}{n}}^n(\theta_i)^{\top}\left[f_t^{(r)}(\theta_i)-\mathbb{E}_{\theta_i}[f_t^{(r)}(\theta_i)]\right]\right|=O(n^{\frac{1}{2}}).$$ 
Therefore, for any sequence $(\epsilon_n)_{n\in\mathbb{N}}$, $\epsilon_n \xrightarrow[]{} 0$ as $n\xrightarrow[]{}\infty$, such that $N(\epsilon_n)=o(n^{\frac{1}{2}})$, we have
\begin{align*}
&\sum_{i \in [N(\epsilon_n)]} \mathbb{E}_{\theta_i}\left|\sum_{t=m+\tau+L}^n  \mu_{\frac{t-L}{n}}^n(\theta_i)^{\top}\left[f_t^{(r)}(\theta_i)-\mathbb{E}_{\theta_i}[f_t^{(r)}(\theta_i)]\right]\right|
\\ &
\leq
N(\epsilon_n) \underset{i\in [N(\epsilon_n)]}{\max} \mathbb{E}_{\theta_i}\left|\sum_{t=m+\tau+L}^n  \mu_{\frac{t-L}{n}}^n(\theta_i)^{\top}\left[f_t^{(r)}(\theta_i)-\mathbb{E}_{\theta_i}[f_t^{(r)}(\theta_i)]\right]\right|
\\ &
=
o(n),
\end{align*}
and $\eta_G=\eta_G(\epsilon_n)=o(1)$ so that the upper bound from (2) is $L_{\varphi}n(m+1)\sqrt{\eta_G}=o(n)$. See, e.g., Corollary 4.2.11 in \cite{Vershynin} for the explicit covering number bounds.

\textbf{Step 3.1.2.} We now turn to the first term. By linearity and monotonicity of expectation, it suffices to show that
$$\mathbb{E}_{\theta_0}\left(\sup_{\theta\in\Theta} \left|\sum_{t=m+\tau+L}^n  \mu_{\frac{t-L}{n}}^n(\theta)^{\top}\left[f_t^{\mathrm{obs}} -\mathbb{E}_{\theta_0}[f_t^{\mathrm{obs}}]\right]\right|\right)=o(n).$$

By Assumption~\ref{asmpt:statistical_model_continuous_time}, $\Theta\subset\mathbb{R}^p$ is compact, so it is totally bounded and there exists an $\epsilon$-net, i.e., for every $\epsilon>0$, there exist finitely many points $\theta_1,\ldots,\theta_{N(\epsilon)}$ for some $N(\epsilon)\in\mathbb{N}$, such that $$\Theta\subset \bigcup_{i\in [N(\epsilon)]}B_{\epsilon}(\theta_i),$$ where $B_{\epsilon}(\theta_i)$ is a ball of radius $\epsilon$. By adding and subtracting the same term, the triangle inequality, the definition of the $\epsilon$-net for some $\epsilon>0$ and $N(\epsilon)\in\mathbb{N}$, and the definition of the supremum, 
\begin{align*}
    &  \sup_{\theta\in\Theta} \left|\sum_{t=m+\tau+L}^n  \mu_{\frac{t-L}{n}}^n(\theta)^{\top}\left[f_t^{\mathrm{obs}}-\mathbb{E}_{\theta_0}[f_t^{\mathrm{obs}}]\right]\right|
    \\& 
    \overset{(1)}{\leq}
    \max_{i \in [N(\epsilon)]} \left|\sum_{t=m+\tau+L}^n  \mu_{\frac{t-L}{n}}^n(\theta_i)^{\top}\left[f_t^{\mathrm{obs}} -\mathbb{E}_{\theta_0}[f_t^{\mathrm{obs}}]\right]\right|
    \\ &
    + \underset{ \norm{\theta-\theta'}\leq \epsilon}{\underset{\theta,\theta'\in\Theta}{\sup}}\abs{\sum_{t=m+\tau+L}^n \left[\begin{aligned}&\mu_{\frac{t-L}{n}}^n(\theta)^{\top}\left[f_t^{\mathrm{obs}}-\mathbb{E}_{\theta_0}[f_t^{\mathrm{obs}}]\right]\\& - \mu_{\frac{t-L}{n}}^n(\theta')^{\top}\left[f_t^{\mathrm{obs}}-\mathbb{E}_{\theta_0}[f_t^{\mathrm{obs}}]\right]\end{aligned}\right]}.
\end{align*}

The arguments for handling the second term are nearly identical to those used in Step 3.1.1, so we omit the details due to the high similarity. The main idea is to use the stochastic equicontinuity-type condition from Assumption~\ref{asmpt:stochastic_equicontinuity_continuous_time} with any sequence $(\epsilon_n)_{n\in\mathbb{N}}$, $\epsilon_n \xrightarrow[]{} 0$ as $n\xrightarrow[]{}\infty$, such that $N(\epsilon_n)=o(n^{\frac{1}{2}})$ and $\eta_G=\eta_G(\epsilon_n)=o(1)$. 

We explain how to handle the first term. By arguments similar to those used in Step 3.1.1 it suffices to show that, for each $i\in [N(\epsilon)]$, we have
$$\mathbb{E}_{\theta_0}\left|\sum_{t=m+\tau+L}^n  \mu_{\frac{t-L}{n}}^n(\theta_i)^{\top}\left[f_t^{\mathrm{obs}}-\mathbb{E}_{\theta_0}(f_t^{\mathrm{obs}})\right]\right|=O(n^{\frac{1}{2}}).$$

For each $\theta\in\Theta$, denote $$Y_t^{\mathrm{obs}}(\theta)=\mu_{\frac{t-L}{n}}^n(\theta)^{\top}\left[f_t^{\mathrm{obs}}-\mathbb{E}_{\theta_0}(f_t^{\mathrm{obs}})\right],\quad t=m+\tau+L,\ldots,n,$$
which admits the representation
$$Y_t^{\mathrm{obs}}(\theta)=H_{t/n}^{(n,0)}(\symbfit{\varepsilon}_t^{(0)},\theta),$$ where the mapping $H_{t/n}^{(n,0)}(\symbfit{\varepsilon}_t^{(0)},\theta)$ is defined in terms of the original causal representations from Assumption~\ref{asmpt:algorithmic_dynamic_model_continuous_time}, $$H_{t/n}^{(n,0)}(\symbfit{\varepsilon}_t^{(0)},\theta)=\mu_{\frac{t-L}{n}}^n(\theta)^{\top}\left[\begin{aligned}&\varphi\left(\left[G_{(t-m)/n}^{(n,0)}(\symbfit{\varepsilon}_{t-m}^{(0)},\theta_0),\ldots,G_{t/n}^{(n,0)}(\symbfit{\varepsilon}_t^{(0)},\theta_0)\right]^{\top}\right)\\&-\mathbb{E}_{\theta_0}(f_t^{\mathrm{obs}})\end{aligned}\right].$$
Analogous to Assumption~\ref{asmpt:temporal_dependence_continuous_time}, denote the version of $Y_t^{\mathrm{obs}}(\theta)$ with the $j$-th noise input in the past replaced with its iid copy by
\begin{align*}\tilde{Y}_{t,j}^{\mathrm{obs}}(\theta)&=H_{t/n}^{(n,0)}(\tilde{\symbfit{\varepsilon}}_{t,j}^{(0)},\theta)\\&=\mu_{\frac{t-L}{n}}^n(\theta)^{\top}\left[\begin{aligned}&\varphi\left(\left[G_{(t-m)/n}^{(n,0)}(\tilde{\symbfit{\varepsilon}}_{t-m,j-m}^{(0)},\theta_0),\ldots,G_{t/n}^{(n,0)}(\tilde{\symbfit{\varepsilon}}_{t,j}^{(0)},\theta_0)\right]^{\top}\right)\\&-\mathbb{E}_{\theta_0}(f_t^{\mathrm{obs}})\end{aligned}\right],\end{align*}
where for any $i\in\mathbb{N}$ we define $\tilde{\symbfit{\varepsilon}}_{t-i,j-i}^{(0)}$ as $\symbfit{\varepsilon}_{t-i}^{(0)}$ if $j-i < 0$. Observe that the time series $Y_t^{\mathrm{obs}}(\theta)$ satisfies a decaying physical dependence measure condition
\begin{align*}
&\norm{Y_t^{\mathrm{obs}}(\theta)-\tilde{Y}_{t,j}^{\mathrm{obs}}(\theta)}_{\mathcal{L}^q(\theta_0)}
\\ &
\overset{(1)}{=}
\norm{H_{t/n}^{(n,0)}(\symbfit{\varepsilon}_{t}^{(0)},\theta)-H_{t/n}^{(n,0)}(\tilde{\symbfit{\varepsilon}}_{t,j}^{(0)},\theta)}_{\mathcal{L}^q(\theta_0)} 
\\&
\overset{(2)}{=}
\norm{\mu_{\frac{t-L}{n}}^n(\theta)^{\top}\left[\begin{aligned}
    & \varphi\left(\left[G_{(t-m)/n}^{(n,0)}(\symbfit{\varepsilon}_{t-m}^{(0)},\theta_0),\ldots,G_{t/n}^{(n,0)}(\symbfit{\varepsilon}_{t}^{(0)},\theta_0)\right]^{\top}\right) \\ & - \varphi\left(\left[G_{(t-m)/n}^{(n,0)}(\tilde{\symbfit{\varepsilon}}_{t-m,j-m}^{(0)},\theta_0),\ldots,G_{t/n}^{(n,0)}(\tilde{\symbfit{\varepsilon}}_{t,j}^{(0)},\theta_0)\right]^{\top}\right)\end{aligned}\right]
}_{\mathcal{L}^q(\theta_0)}
\\&
\overset{(3)}{\leq}
\norm{\mu_{\frac{t-L}{n}}^n(\theta)}\norm{\begin{aligned}
    & \varphi\left(\left[G_{(t-m)/n}^{(n,0)}(\symbfit{\varepsilon}_{t-m}^{(0)},\theta_0),\ldots,G_{t/n}^{(n,0)}(\symbfit{\varepsilon}_{t}^{(0)},\theta_0)\right]^{\top}\right) \\ & - \varphi\left(\left[G_{(t-m)/n}^{(n,0)}(\tilde{\symbfit{\varepsilon}}_{t-m,j-m}^{(0)},\theta_0),\ldots,G_{t/n}^{(n,0)}(\tilde{\symbfit{\varepsilon}}_{t,j}^{(0)},\theta_0)\right]^{\top}\right)\end{aligned}}_{\mathcal{L}^q(\theta_0)}
\\&
\overset{(4)}{\leq} 2\sqrt{k} L_{\varphi} \sum_{i=0}^m
\norm{
G_{(t-i)/n}^{(n,0)}(\symbfit{\varepsilon}_{t-i}^{(0)},\theta_0)  - G_{(t-i)/n}^{(n,0)}(\tilde{\symbfit{\varepsilon}}_{t-i,j-i}^{(0)},\theta_0) }_{\mathcal{L}^q(\theta_0)}
\\&
\overset{(5)}{=} 
2\sqrt{k} L_{\varphi} \sum_{i=0}^{\min(m,j)}
\norm{
G_{(t-i)/n}^{(n,0)}(\symbfit{\varepsilon}_{t-i}^{(0)},\theta_0)  - G_{(t-i)/n}^{(n,0)}(\tilde{\symbfit{\varepsilon}}_{t-i,j-i}^{(0)},\theta_0) }_{\mathcal{L}^q(\theta_0)}
\\&
\overset{(6)}{\leq} 2\sqrt{k} L_{\varphi} \sum_{i=0}^{\min(m,j)}
\Lambda \rho^{j-i}
\\&
\overset{(7)}{\leq} 2\sqrt{k} L_{\varphi} (m+1)
\Lambda \rho^{j-\min(m,j)}
\\&
\overset{(8)}{\leq} \Lambda^H (\rho^H)^j,
\end{align*} by (1) definition, (2) cancellation of terms and the linearity of the inner product, (3) the Cauchy-Schwarz inequality and linearity of expectation, (4) because the random Fourier features are bounded between $-1$ and $1$ so the expectations must be as well, and because the random Fourier features $\varphi$ are $L_{\varphi}$-Lipschitz with Lipschitz constant given by~\eqref{eqn:lipschitz_constant_rff}, the triangle inequality, and the Minkowski inequality, (5) because the distances are zero for $j-i< 0$, so it suffices to consider the terms where $j-i \geq 0 \iff i \leq j$, (6) using the bound on the physical dependence measure from Assumption~\ref{asmpt:temporal_dependence_continuous_time}, (7) upper bounding by $m+1$ times the largest term which is the ``earliest'' term because $\rho\in (0,1)$, and (8) upper bounding with an expression that satisfies exponential decay condition analogous to Assumption~\ref{asmpt:temporal_dependence_continuous_time}, where
$\Lambda^H = 2\sqrt{k} L_{\varphi} (m+1)
\Lambda \rho^{-m} >0$ and $\rho^H = \rho\in (0,1)$. Lemma~\ref{lma:lp_lln_continuous_time} is therefore applicable to the time series $Y_t^{\mathrm{obs}}(\theta)$.

Let us return to the term of interest. For all $i\in [N(\epsilon)]$, we have 
\begin{align*}
&\mathbb{E}_{\theta_0}\left|\sum_{t=m+\tau+L}^n  \mu_{\frac{t-L}{n}}^n(\theta_i)^{\top}\left[f_t^{\mathrm{obs}}-\mathbb{E}_{\theta_0}(f_t^{\mathrm{obs}})\right]\right|
\\& 
\overset{(1)}{=}
\mathbb{E}_{\theta_0}\left|\sum_{t=m+\tau+L}^n  Y_t^{\mathrm{obs}}(\theta_i)\right|
\\&
\overset{(2)}{\leq}
C n^{\frac{1}{2}} \Lambda^H \sum_{j=0}^{\infty}  (\rho^H)^j
\\&
\overset{(3)}{=}
C n^{\frac{1}{2}} \Lambda^H K^H,
\end{align*} by (1) definition, (2) the second inequality from Lemma~\ref{lma:lp_lln_continuous_time}, noting that $\mathbb{E}_{\theta_0}[Y_t^{\mathrm{obs}}(\theta)]=0$ for all $t$, $\theta$, by construction, applying the inequality on the physical dependence measure derived above, and further upper bounding with infinitely many terms in the summation, (3) noting that the geometric series converges and the constant $K^H=\frac{1}{1-\rho^H}>0$ since $\rho^H=\rho\in (0,1)$ by Assumption~\ref{asmpt:temporal_dependence_continuous_time}. Hence, for each fixed $\epsilon>0$ and $N(\epsilon)\in\mathbb{N}$, for all $i\in [N(\epsilon)]$, we have
$$\mathbb{E}_{\theta_0}\left|\sum_{t=m+\tau+L}^n  \mu_{\frac{t-L}{n}}^n(\theta_i)^{\top}\left[f_t^{\mathrm{obs}}-\mathbb{E}_{\theta_0}(f_t^{\mathrm{obs}})\right]\right|=O(n^{\frac{1}{2}}).$$

Putting everything together, by Markov's inequality, we have the desired result
$$\sup_{\theta\in\Theta}\abs{\frac{2}{n-m}\sum_{t=m+\tau+L}^n \mu_{\frac{t-L}{n}}^n(\theta)^{\top}\left(\left[f_t^{\mathrm{obs}}-\bar{f}_{t}^{\mathrm{sim}}(\theta)\right]-\mu_{\frac{t}{n}}^n(\theta)\right)}=o_p(1).$$

\textbf{Step 3.2.} Next, we show that $$\sup_{\theta\in\Theta}\abs{\frac{2}{n-m}\sum_{t=m+\tau+L}^n \left(\hat{\mu}_{t-L,n}(\theta)-\mu_{\frac{t-L}{n}}^n(\theta)\right)^{\top}\left(\left[f_t^{\mathrm{obs}}-\bar{f}_{t}^{\mathrm{sim}}(\theta)\right]-\mu_{\frac{t}{n}}^n(\theta)\right)}=o_p(1).$$
Observe that
\begin{align*}
&\sup_{\theta\in\Theta}\abs{\frac{2}{n-m}\sum_{t=m+\tau+L}^n \left(\hat{\mu}_{t-L,n}(\theta)-\mu_{\frac{t-L}{n}}^n(\theta)\right)^{\top}\left(\left[f_t^{\mathrm{obs}}-\bar{f}_{t}^{\mathrm{sim}}(\theta)\right]-\mu_{\frac{t}{n}}^n(\theta)\right)}
\\&
\overset{(1)}{\leq}
\sup_{\theta\in\Theta}\frac{2}{n-m}\sum_{t=m+\tau+L}^n \norm{\hat{\mu}_{t-L,n}(\theta)-\mu_{\frac{t-L}{n}}^n(\theta)}\norm{\left(\left[f_t^{\mathrm{obs}}-\bar{f}_{t}^{\mathrm{sim}}(\theta)\right]-\mu_{\frac{t}{n}}^n(\theta)\right)}
\\&
\overset{(2)}{\leq}
\left(\sup_{\theta\in\Theta}\frac{2}{n-m}\sum_{t=m+\tau+L}^n \norm{\hat{\mu}_{t-L,n}(\theta)-\mu_{\frac{t-L}{n}}^n(\theta)}^2\right)^{\frac{1}{2}} \\&   \left( \sup_{\theta\in\Theta}\frac{2}{n-m}\sum_{t=m+\tau+L}^n\norm{\left(\left[f_t^{\mathrm{obs}}-\bar{f}_{t}^{\mathrm{sim}}(\theta)\right]-\mu_{\frac{t}{n}}^n(\theta)\right)}^2\right)^{\frac{1}{2}}
\\&
\overset{(3)}{\leq}
C  \left(\sup_{\theta\in\Theta}\frac{2}{n-m}\sum_{t=m+\tau+L}^n \norm{\hat{\mu}_{t-L,n}(\theta)-\mu_{\frac{t-L}{n}}^n(\theta)}^2\right)^{\frac{1}{2}}
\\&
\overset{(4)}{=}
o_p(1),
\end{align*} by (1) the triangle inequality and the Cauchy-Schwarz inequality, (2) the Cauchy-Schwarz inequality, (3) the triangle inequality and the fact that the random features are bounded between $-1$ and $1$, so we can upper bound this average by some constant $C>0$, and (4) because this term was previously shown to be $o_p(1)$ in Step 2.2 of the proof~\eqref{eqn:step_2_2}.
\hfill$\square$

\subsection{Auxiliary results}\label{subsection:auxiliary_lemmas}

To begin, we derive the Lipschitz constant for the random Fourier features.

\begin{lemma} The function $\varphi=(\varphi_1,\ldots,\varphi_k)$ from~\eqref{eqn:all_k_random_Fourier_features_order_m} is Lipschitz with Lipschitz constant  \begin{equation}\label{eqn:lipschitz_constant_rff}L_{\varphi}=\left(\sum_{i=1}^k\sum_{j=1}^{m+1}\norm{\Omega_{i,j}}^2\right)^{1/2}.\end{equation}
\end{lemma}
\textit{Proof:}
 For each coordinate $i\in [k]$ of $\varphi=(\varphi_1,\ldots,\varphi_k)$, we have, for $x,y \in \mathbb{R}^{(m+1)\times d}$, \begin{align*}
&\left|\varphi_i(x)-\varphi_i(y)\right| 
\\ 
& \overset{(1)}{=} \left| \cos\left(\sum_{j=1}^{m+1} \Omega_{i,j} \cdot x_{j} + \alpha_i\right)
-
\cos\left(\sum_{j=1}^{m+1} \Omega_{i,j} \cdot y_{j} + \alpha_i\right)
\right| 
\\ 
& \overset{(2)}{\leq} \left| \sum_{j=1}^{m+1} \Omega_{i,j} \cdot \left(x_{j} - y_{j}\right) 
\right|
\\ 
& \overset{(3)}{\leq}  \sum_{j=1}^{m+1} \norm{\Omega_{i,j}} \norm{x_{j} - y_{j}} 
\\ 
& \overset{(4)}{\leq}  \left(\sum_{j=1}^{m+1} \norm{\Omega_{i,j}}^2\right)^{1/2} \left(\sum_{j=1}^{m+1}\norm{x_{j} - y_{j}}^2 \right)^{1/2} 
\\ 
& \overset{(5)}{\leq}  \left(\sum_{j=1}^{m+1} \norm{\Omega_{i,j}}^2\right)^{1/2} \norm{\mathrm{Vec}(x) - \mathrm{Vec}(y)}, 
\end{align*} because (1) the definition of the random Fourier features $\varphi$, (2) cosine is $1$-Lipschitz, cancelling the phase $\alpha_i$, and linearity of sums and inner products, (3) the triangle inequality and the Cauchy-Schwarz inequality applied to each inner product, (4) the Cauchy-Schwarz inequality applied to the sum of products, and (5) the definition of the Euclidean norm, which yields $$\left|\varphi_i(x)-\varphi_i(y)\right|\leq \left(\sum_{j=1}^{m+1} \norm{\Omega_{i,j}}^2\right)^{1/2} \norm{\mathrm{Vec}(x) - \mathrm{Vec}(y)},$$ so squaring both sides, summing over $i=1,\ldots,k$, and taking square roots, yields $$\norm{\varphi(x)-\varphi(y)}\leq \left(\sum_{i=1}^k\sum_{j=1}^{m+1} \norm{\Omega_{i,j}}^2\right)^{1/2} \norm{\mathrm{Vec}(x) - \mathrm{Vec}(y)}.$$

\hfill$\square$

The next result establishes the continuity of $\theta\mapsto Q^{\mathrm{RW}}(\theta)$.

\begin{lemma}\label{lma:continuity_of_theta_to_QRW}
Suppose the assumptions of Theorem~\ref{thm:consistency_rolling_window_estimator} hold. Then the population discrepancy function
$$ Q^{\mathrm{RW}}(\theta)
= \int_0^1
\|\Phi_u(\theta_0) - \Phi_u(\theta)\|^2\,du,
$$ is continuous in $\theta$ on $\Theta$.
\end{lemma}

\textit{Proof:} For all $\theta\in\Theta$, and for all $c > 0$, there exists an $\epsilon >0$ such that, for all $\theta' \in\Theta$, if $\norm{\theta-\theta'}<\epsilon$, then
\begin{align*}
& \abs{Q^{\mathrm{RW}}(\theta)-Q^{\mathrm{RW}}(\theta')}
\\&
\overset{(1)}{=} \abs{\int_0^1
\|\Phi_u(\theta_0) - \Phi_u(\theta)\|^2\,du - \int_0^1
\|\Phi_u(\theta_0) - \Phi_u(\theta')\|^2\,du}
\\& 
\overset{(2)}{\leq} \int_0^1 \abs{
\|\Phi_u(\theta_0) - \Phi_u(\theta)\|^2 - 
\|\Phi_u(\theta_0) - \Phi_u(\theta')\|^2 }\,du
\\& 
\overset{(3)}{\leq} \int_0^1 \left(
\|\Phi_u(\theta_0) - \Phi_u(\theta)\| + \|\Phi_u(\theta_0) - \Phi_u(\theta')\| \right) \|\Phi_u(\theta) - \Phi_u(\theta')\|\,du
\\& 
\overset{(4)}{\leq} 4\sqrt{k} \underset{u\in [0,1]}{\sup} \norm{\Phi_u(\theta)-\Phi_u(\theta')} 
\\& 
\overset{(5)}{\leq} 4\sqrt{k} [e(\epsilon) \\& + L_{\varphi} (m+1) \ \underset{i \in\mathbb{N}}{\sup} \ \underset{u \in [0,1]}{\sup} \ \underset{0\leq r\leq s}{\max} \ \mathbb{E}(
\underset{ \left\|\theta-\theta'\right\|\leq \epsilon}{\underset{\theta,\theta'\in\Theta}{\sup}} \left\|
G_{u}^{(i,r)}(\symbfit{\varepsilon}_{0},\theta) - 
G_{u}^{(i,r)}(\symbfit{\varepsilon}_{0},\theta')\right\|^2)^{1/2}]
\\& 
\overset{(6)}{\leq} 4\sqrt{k} [e(\epsilon) +  L_{\varphi} (m+1) \sqrt{\eta_G(\epsilon)}]
\\&
\overset{(7)}{\leq}
c,
\end{align*}
by (1) definition, (2) the triangle inequality, (3) the equality $\abs{a^2-b^2}=(a+b)\abs{a-b}$ and the reverse triangle inequality applied to $\abs{a-b}$, (4) upper bounding by two times $M=\underset{u\in [0,1]}{\sup} \ \underset{\theta \in \Theta}{\sup} \norm{\Phi_u(\theta_0)-\Phi_u(\theta)}\leq 2\sqrt{k} < \infty$, i.e., by taking the supremum over rescaled time and parameter values, and multiplying by the length of the interval $[0,1]$ (i.e., $1$), (5) adding and subtracting the expectations at approximation level $i$ and the triangle inequality, where we can let $i=i(\epsilon)$ grow as $\epsilon$ decreases so that the approximation error, denoted by $e(\epsilon)$, decreases to $0$ as in~\eqref{eqn:convergence_to_tilde_Phi_theta_u_continuous_time}, then using the same arguments from Theorems~\ref{thm:injective_continuous_time} and~\ref{thm:consistency_rolling_window_estimator} to upper bound the Euclidean distance between the expectations of the features (i.e., by linearity of expectation, Jensen's inequality, because the random Fourier features $\varphi$ are $L_{\varphi}$-Lipschitz with Lipschitz constant given by~\eqref{eqn:lipschitz_constant_rff}, the triangle inequality, and Minkowski inequality), (6) upper bounding by $\eta_G(\epsilon)$ using the stochastic equicontinuity condition from Assumption~\ref{asmpt:stochastic_equicontinuity_continuous_time}, (7) follows by choosing $\epsilon>0$ sufficiently small, since both $e(\epsilon)\xrightarrow[]{} 0$ and $\eta_G(\epsilon)\xrightarrow[]{} 0$ as $\epsilon \xrightarrow[]{} 0$.

\hfill$\square$

For $t\in \mathbb{N}$, $i\in\mathbb{N}$, $j\in\mathbb{N}_0$, $r=0,1,\ldots,s$, define
\begin{align*}Z_{t}^{(i,r)}(\theta)&=\left[G_{t-m}^{(i,r)}(\symbfit{\varepsilon}_{-m}^{(r)},\theta),\ldots,G_t^{(i,r)}(\symbfit{\varepsilon}_0^{(r)},\theta)\right]^{\top}, 
\\ 
\tilde{Z}_{t,j}^{(i,r)}(\theta) &=\left[G_{t-m}^{(i,r)}(\bar{\symbfit{\varepsilon}}_{-m,j}^{(r)},\theta),\ldots,G_t^{(i,r)}(\bar{\symbfit{\varepsilon}}_{0,j}^{(r)},\theta)\right]^{\top},\end{align*}
where for $\ell=0,\ldots,m$ we define $$\bar{\symbfit{\varepsilon}}_{-\ell,j}^{(r)}=\begin{cases} \tilde{\symbfit{\varepsilon}}_{-\ell,j-\ell}^{(r)} \text{ for } \ell \leq j,
\\ \symbfit{\varepsilon}_{-\ell}^{(r)} \text{ for } \ell > j,
\end{cases}$$
where $\tilde{\symbfit{\varepsilon}}_{-\ell,j-\ell}$ is from~\eqref{eqn:replaced_j_past_noise_input_sequence_up_to_time_t}.

The following result establishes that the time series of random features satisfies a physical dependence measure condition, analogous to Assumption~\ref{asmpt:temporal_dependence_discrete_time}. 
\begin{lemma}\label{lma:temporal_dependence_of_random_features_discrete_time}
There exist $\Psi^{\varphi}>0$, $\beta^{\varphi}>2$ such that, for some $q^{\varphi} > 2$ and all $j\in\mathbb{N}_0$,   \begin{align*}  \underset{i \in\mathbb{N}}{\sup} \ \underset{\theta\in\Theta}{\sup} \ \underset{t\in\mathbb{N}}{\sup}   \  \underset{0\leq r \leq s}{\max} \   \norm{\varphi\left(Z_{t}^{(i,r)}(\theta)\right) - \varphi\left(
\tilde{Z}_{t,j}^{(i,r)}(\theta)\right)}_{\mathcal{L}^{q^{\varphi}}(\theta)}   & \leq   \Psi^{\varphi} (j \lor 1)^{-\beta^{\varphi}},\\  \underset{i \in\mathbb{N}}{\sup} \ \underset{\theta\in\Theta}{\sup} \ \underset{t\in\mathbb{N}}{\sup} \ \underset{0\leq r \leq s}{\max} \ \norm{\varphi\left(Z_{t}^{(i,r)}(\theta)\right)}_{\mathcal{L}^{q^{\varphi}}(\theta)}  &\leq  \Psi^{\varphi}. \end{align*} 
\end{lemma}
\textit{Proof:} The result follows from the fact that $\varphi$ is Lipschitz with Lipschitz constant $L_{\varphi}$ from~\eqref{eqn:lipschitz_constant_rff}, the triangle inequality and Minkowski inequality, and applying the bounds from Assumption~\ref{asmpt:temporal_dependence_discrete_time}. The same line of reasoning is used in the proofs of Theorems~\ref{thm:injective_discrete_time} and~\ref{thm:consistency_time_average_estimator}, so the details are omitted.

\hfill$\square$

The following result presents Theorem 3.2 from~\cite{seq_gauss_approx2022} using our notation.

\begin{lemma}\label{lma:lp_lln_discrete_time} 
For $r=0,1,\ldots,s$ and $\theta\in\Theta$, suppose a time series $X_t^{(r)}(\theta)$, $t=1,\ldots,n$, is generated as $$X_t^{(r)}(\theta)=G_{t}^{(n,r)}(\symbfit{\varepsilon}_{t},\theta).$$ Suppose there exists a constant $q\geq 2$, such that $$\underset{i\in\mathbb{N}}{\sup} \ \underset{\theta\in\Theta}{\sup} \  \underset{t\in \mathbb{N}}{\sup} \ \underset{0\leq r \leq s}{\max} \  \norm{G_t^{(i,r)}(\symbfit{\varepsilon}_{0},\theta)}_{\mathcal{L}^q(\theta)}<\infty.$$ For some $2\leq h \leq q < \infty$, define the physical dependence measure at time $t$ and sample size $n$ as $$\eta_{t,j,q,h}^{(n,r)}(\theta)=  \mathbb{E}_{\theta}\left(\norm{G_t^{(n,r)}(\symbfit{\varepsilon}_{t},\theta) - G_t^{(n,r)}(\tilde{\symbfit{\varepsilon}}_{t,j},\theta)}_h^q\right)^{1/q}.$$ There exists a constant $C = C(q,h)$, such that for all $r=0,1,\ldots,s$, $\theta\in\Theta$, and $n \in \mathbb{N}$, we have
\begin{align*}&
 \left(
\mathbb{E}_{\theta} \left[\max_{\ell \leq n}
\left\|
\sum_{t=1}^{\ell} \left(X_t^{(r)}(\theta) - \mathbb{E}_{\theta}\left[X_t^{(r)}(\theta)\right]\right)
\right\|_h^q\right]
\right)^{\frac{1}{q}}\\
&=
 \left(
\mathbb{E}_{\theta} \left[\max_{\ell \leq n}
\left\|
\sum_{t=1}^{\ell} \left(G_t^{(n,r)}(\symbfit{\varepsilon}_{t},\theta) - \mathbb{E}_{\theta}\left[G_t^{(n,r)}(\symbfit{\varepsilon}_{t},\theta)\right]\right)
\right\|_h^q\right]
\right)^{\frac{1}{q}}\\
&\leq
C n^{\frac{1}{2} - \frac{1}{q}}
\sum_{j=0}^{\infty}
\left(
\sum_{t=1}^n (\eta_{t,j,q,h}^{(n,r)}(\theta))^q
\right)^{\frac{1}{q}} \\ &\leq
C n^{\frac{1}{2}}
\sum_{j=0}^{\infty}
\max_{t \leq n} \eta_{t,j,q,h}^{(n,r)}(\theta).
\end{align*} % and in the special case of $h = 2$, the inequality may be improved to \begin{align*}& \left(\mathbb{E}_{\theta}\left[ \max_{\ell \leq n}\left\|\sum_{t=1}^{\ell} \left(X_t^{(r)}(\theta) - \mathbb{E}_{\theta}\left[X_t^{(r)}(\theta)\right]\right)\right\|_2^q\right]\right)^{\frac{1}{q}}\\&=  \left(\mathbb{E}_{\theta}\left[ \max_{\ell \leq n}\left\|\sum_{t=1}^{\ell} \left(G_t^{(n,r)}(\symbfit{\varepsilon}_{t},\theta) - \mathbb{E}_{\theta}\left[G_t^{(n,r)}(\symbfit{\varepsilon}_{t},\theta)\right]\right)\right\|_2^q\right]\right)^{\frac{1}{q}}\\&\leq C \sum_{j=0}^{\infty} (j \wedge n)^{\frac{1}{2} - \frac{1}{q}}\left(\sum_{t=1}^n (\eta_{t,j,q,2}^{(n,r)}(\theta))^q \right)^{\frac{1}{q}}+ C \sum_{j=0}^n\left(\sum_{t=1}^n (\eta_{t,j,2,2}^{(n,r)}(\theta))^2\right)^{\frac{1}{2}}.\end{align*}
\end{lemma}

\textit{Proof:} See the proof of Theorem 3.2 in~\cite{seq_gauss_approx2022}.
\hfill$\square$

For $u\in [0,1]$, $i\in\mathbb{N}$, $j\in\mathbb{N}_0$, $r=0,1,\ldots,s$, define
\begin{align*}Z_{u}^{(i,r)}(\theta)&=\left[G_{u}^{(i,r)}(\symbfit{\varepsilon}_{-m}^{(r)},\theta),\ldots,G_u^{(i,r)}(\symbfit{\varepsilon}_0^{(r)},\theta)\right]^{\top}, 
\\ 
\tilde{Z}_{u,j}^{(i,r)}(\theta) &=\left[G_{u}^{(i,r)}(\bar{\symbfit{\varepsilon}}_{-m,j}^{(r)},\theta),\ldots,G_u^{(i,r)}(\bar{\symbfit{\varepsilon}}_{0,j}^{(r)},\theta)\right]^{\top},\end{align*} where for $\ell=0,\ldots,m$ we define $$\bar{\symbfit{\varepsilon}}_{-\ell,j}^{(r)}=\begin{cases} \tilde{\symbfit{\varepsilon}}_{-\ell,j-\ell}^{(r)} \text{ for } \ell \leq j,
\\ \symbfit{\varepsilon}_{-\ell}^{(r)} \text{ for } \ell > j,
\end{cases}$$
where $\tilde{\symbfit{\varepsilon}}_{-\ell,j-\ell}^{(r)}$ is from~\eqref{eqn:replaced_j_past_noise_input_sequence_up_to_time_t}.

The next result shows that the time series of random features satisfies a physical dependence measure condition, analogous to Assumption~\ref{asmpt:temporal_dependence_continuous_time}.
\begin{lemma}\label{lma:temporal_dependence_of_random_features_continuous_time}
 There exist constants $\rho^{\varphi}\in (0,1)$, $\Lambda^{\varphi} >0$, such that, for some $q^{\varphi} > 2$ and all $j\in\mathbb{N}_0$, we have \begin{align*}  \underset{i\in\mathbb{N}}{\sup} \ \underset{\theta\in\Theta}{\sup} \ \underset{u\in [0,1]}{\sup}  \ \underset{0\leq r \leq s}{\max} \   \norm{\varphi\left(Z_{u}^{(i,r)}(\theta)\right) - \varphi\left(\tilde{Z}_{u,j}^{(i,r)}(\theta)\right)}_{\mathcal{L}^{q^{\varphi}}(\theta)}   &\leq \Lambda^{\varphi} (\rho^{\varphi})^j, \\  \underset{i\in\mathbb{N}}{\sup} \ \underset{\theta\in\Theta}{\sup} \  \underset{u\in [0,1]}{\sup}  \ \underset{0\leq r \leq s}{\max} \    \norm{\varphi\left(Z_{u}^{(i,r)}(\theta)\right)}_{\mathcal{L}^{q^{\varphi}}(\theta)}   &\leq \Lambda^{\varphi}.\end{align*} 
\end{lemma}
\textit{Proof:} The result follows from the fact that $\varphi$ is Lipschitz with Lipschitz constant $L_{\varphi}$ from~\eqref{eqn:lipschitz_constant_rff}, the triangle inequality and Minkowski inequality, and applying the bounds from Assumption~\ref{asmpt:temporal_dependence_continuous_time}. The same line of reasoning is used in the proofs of Theorems~\ref{thm:injective_continuous_time} and~\ref{thm:consistency_rolling_window_estimator}, so the details are omitted.

\hfill$\square$

The following result establishes that the time series of random features satisfies a nonstationarity condition, analogous to Assumption~\ref{asmpt:nonstationarity_continuous_time}.
\begin{lemma}\label{lma:nonstationarity_of_random_features_continuous_time}
For some $q^{\varphi} > 2$ and some $\kappa^{\varphi}\in [1,4)$, the constant $\Lambda^{\varphi} >0$ from Lemma~\ref{lma:temporal_dependence_of_random_features_continuous_time} also satisfies \begin{align*} \underset{i \in\mathbb{N}}{\sup} \ \underset{\theta\in\Theta}{\sup} \ \underset{0 \leq r \leq s}{\max} \left(\begin{aligned}& \underset{u\in [0,1]}{\sup} \  \norm{\varphi\left(Z_{u}^{(i,r)}(\theta)\right)}_{\mathcal{L}^{q^{\varphi}}(\theta)} \\& +  \norm{\left(\varphi\left(Z_{u}^{(i,r)}(\theta)\right)\right)_u}_{\kappa^{\varphi}\text{-}\mathrm{var},\mathcal{L}^{q^{\varphi}}(\theta)} \end{aligned} \right) &\leq \Lambda^{\varphi}.\end{align*}
\end{lemma}
\textit{Proof:} The result follows from the fact that $\varphi$ is Lipschitz with Lipschitz constant $L_{\varphi}$ from~\eqref{eqn:lipschitz_constant_rff}, the triangle inequality and Minkowski inequality, and applying the bounds from Assumption~\ref{asmpt:nonstationarity_continuous_time}. The same line of reasoning is used in the proofs of Theorems~\ref{thm:injective_continuous_time} and~\ref{thm:consistency_rolling_window_estimator}, so the details are omitted.

\hfill$\square$

For all times $t=m+1,\ldots,n$, simulation indices $r=0,1,\ldots,s$, and parameter values $\theta\in\Theta$, recall the $r$-th random feature $f_t^{(r)}(\theta)$ at time $t$ based on $\theta$ as in~\eqref{eqn:random_features_at_time_t}. The time series of random features satisfies a decaying autocovariance condition, analogous to Proposition 1 from~\cite{mies_random_mult_sga_2024}. For the following result, denote the trace norm of a matrix $A\in\mathbb{R}^{k\times k}$ by $\norm{A}_{\mathrm{tr}}=\mathrm{tr}((A^{\top}A)^{\frac{1}{2}})$.
\begin{lemma}\label{lma:temporal_dependence_of_random_features_continuous_time_implies_covariance_decay}
 There exist constants $C^{\varphi}\equiv C^{\varphi}(\rho^{\varphi},\Lambda^{\varphi})>0$, $K^{\varphi}\equiv K^{\varphi}(\rho^{\varphi},\Lambda^{\varphi})>1$ depending on the constants $\rho^{\varphi}\in (0,1)$, $\Lambda^{\varphi} >0$ from Lemma~\ref{lma:temporal_dependence_of_random_features_continuous_time} and Lemma~\ref{lma:nonstationarity_of_random_features_continuous_time} such that, for all $n\in\mathbb{N}$ and $t_1,t_2\in\{m+1,\ldots,n\}$, we have $$\underset{\theta\in\Theta}{\sup} \ \underset{0\leq r \leq s}{\max} \  \norm{\mathrm{Cov}\left(f_{t_2}^{(r)}(\theta),f_{t_1}^{(r)}(\theta)\right)}_{\mathrm{tr}}\leq C^{\varphi}(|t_2-t_1|+1)^{-K^{\varphi}}.$$
\end{lemma}
\textit{Proof:} Lemma~\ref{lma:temporal_dependence_of_random_features_continuous_time} states that, \textit{uniformly} over the parameter values $\theta\in\Theta$, the physical dependence measure of the time series of random features decays exponentially. In particular, this implies a (slower) polynomial decay of the physical dependence measure of the random features as in Assumption A.2 of~\cite{mies_random_mult_sga_2024}.  That is, there exist constants $A^{\varphi}\equiv A^{\varphi}(\rho^{\varphi},\Lambda^{\varphi})>0$, $K^{\varphi}\equiv K^{\varphi}(\rho^{\varphi},\Lambda^{\varphi})>1$ such that, for some $q^{\varphi} > 2$, all $j\in\mathbb{N}_0$, all $i\in \mathbb{N}$, and all $r=0,1,\ldots,s$, we have \begin{align*}  \underset{\theta\in\Theta}{\sup} \ \underset{u\in [0,1]}{\sup}   \    \norm{\varphi\left(Z_{u}^{(i,r)}(\theta)\right) - \varphi\left(
\tilde{Z}_{u,j}^{(i,r)}(\theta)\right)}_{\mathcal{L}^{q^{\varphi}}(\theta)}   & \leq   A^{\varphi} (j + 1)^{-K^{\varphi}}. \end{align*}

Proposition 1 in~\cite{mies_random_mult_sga_2024} proves that a polynomial decay of the physical dependence measure is sufficient to imply a decay of the trace norm of the autocovariance matrix of a time series. The same arguments used to prove Proposition 1 in~\cite{mies_random_mult_sga_2024} can be applied to the time series of random features for each parameter value, i.e., for all $t_1,t_2\in\{m+1,\ldots,n\}$, $r=0,\ldots,s$, and $\theta\in\Theta$ we have $$ \norm{\mathrm{Cov}\left(f_{t_2}^{(r)}(\theta),f_{t_1}^{(r)}(\theta)\right)}_{\mathrm{tr}}\leq C^{\varphi}(|t_2-t_1|+1)^{-K^{\varphi}}.$$ Since the suprema over all $\theta\in\Theta$ of the upper bounds are always finite, i.e., because $C^{\varphi}$ and $K^{\varphi}$ do not depend on $\theta$, the $\Theta$-uniform inequality holds by basic properties of the supremum.
\hfill$\square$

The following result is Theorem 3.2 from~\cite{seq_gauss_approx2022} using our notation, which is possible because the triangular array framework of~\cite{seq_gauss_approx2022} nests the locally stationary time series setting; see Examples 2 and 3 therein.

\begin{lemma}\label{lma:lp_lln_continuous_time}
For $r=0,1,\ldots,s$ and $\theta\in\Theta$, suppose a time series $X_t^{(r)}(\theta)$, $t=1,\ldots,n$, is generated as $$X_t^{(r)}(\theta)=G_{t/n}^{(n,r)}(\symbfit{\varepsilon}_{t},\theta).$$ Suppose there exists a constant $q\geq 2$, such that $$\underset{i\in\mathbb{N}}{\sup} \ \underset{\theta\in\Theta}{\sup} \  \underset{u\in [0,1]}{\sup} \  \underset{0\leq r \leq s}{\max} \ \norm{G_u^{(i,r)}(\symbfit{\varepsilon}_{0},\theta)}_{\mathcal{L}^q(\theta)}<\infty.$$ For some $2\leq h \leq q < \infty$, define the physical dependence measure at time $t$ and sample size $n$ as $$\eta_{t,j,q,h}^{(n,r)}(\theta)= \mathbb{E}_{\theta}\left(\norm{G_{t/n}^{(n,r)}(\symbfit{\varepsilon}_{t},\theta) - G_{t/n}^{(n,r)}(\tilde{\symbfit{\varepsilon}}_{t,j},\theta)}_h^q\right)^{1/q}.$$ There exists a constant $C = C(q,h)$, such that for all $r=0,1,\ldots,s$, $\theta\in\Theta$, and $n \in \mathbb{N}$, we have
\begin{align*}
&
 \left(
\mathbb{E}_{\theta} \left[\max_{\ell \leq n}
\left\|
\sum_{t=1}^{\ell} \left(X_{t}^{(r)}(\theta) - \mathbb{E}_{\theta}\left[X_{t}^{(r)}(\theta)\right]\right)
\right\|_h^q\right]
\right)^{\frac{1}{q}}\\
&=
 \left(
\mathbb{E}_{\theta} \left[\max_{\ell \leq n}
\left\|
\sum_{t=1}^{\ell} \left(G_{t/n}^{(n,r)}(\symbfit{\varepsilon}_{t},\theta) - \mathbb{E}_{\theta}\left[G_{t/n}^{(n,r)}(\symbfit{\varepsilon}_{t},\theta)\right]\right)
\right\|_h^q\right]
\right)^{\frac{1}{q}}\\
&\leq
C n^{\frac{1}{2} - \frac{1}{q}}
\sum_{j=0}^{\infty}
\left(
\sum_{t=1}^n (\eta_{t,j,q,h}^{(n,r)}(\theta))^q
\right)^{\frac{1}{q}} \\ &\leq
C n^{\frac{1}{2}}
\sum_{j=0}^{\infty}
\max_{t \leq n} \eta_{t,j,q,h}^{(n,r)}(\theta).
\end{align*} % and in the special case of $h = 2$, the inequality may be improved to\begin{align*}& \left(\mathbb{E}_{\theta}\left[ \max_{\ell \leq n}\left\|\sum_{t=1}^{\ell} \left(X_{t}^{(r)}(\theta) - \mathbb{E}_{\theta}\left[X_{t}^{(r)}(\theta)\right]\right)\right\|_2^q \right]\right)^{\frac{1}{q}}\\&= \left(\mathbb{E}_{\theta}\left[ \max_{\ell \leq n}\left\|\sum_{t=1}^{\ell} \left(G_{t/n}^{(n,r)}(\symbfit{\varepsilon}_{t},\theta) - \mathbb{E}_{\theta}\left[G_{t/n}^{(n,r)}(\symbfit{\varepsilon}_{t},\theta)\right]\right)\right\|_2^q\right]\right)^{\frac{1}{q}}\\&\leq C \sum_{j=0}^{\infty} (j \wedge n)^{\frac{1}{2} - \frac{1}{q}}\left(\sum_{t=1}^n (\eta_{t,j,q,2}^{(n,r)}(\theta))^q\right)^{\frac{1}{q}}+ C \sum_{j=0}^n \left(\sum_{t=1}^n (\eta_{t,j,2,2}^{(n,r)}(\theta))^2 \right)^{\frac{1}{2}}.\end{align*}
\end{lemma}

\textit{Proof:} See the proof of Theorem 3.2 in~\cite{seq_gauss_approx2022}.
\hfill$\square$

% the existing simulations are enough 
\section{Additional experiments}\label{section:additional_experiments}

\subsection{More experiments with time-average estimator}\label{section:additional_experiments_with_time_average_estimator}

\textbf{Gaussian model.} For $t=1,\ldots,n$, we observe $$X_t = \mu + \varepsilon_t,$$ where each $\varepsilon_t \overset{\mathrm{iid}}{\sim} N(0,1)$. The unknown parameter is $\mu$, so we use $2p+1=3$ random features. We aim to estimate $\mu_0=0.5$.

\begin{figure}[H]
%\captionsetup{width=\linewidth, margin=0pt}
  \centering \includegraphics[width=0.25\linewidth]{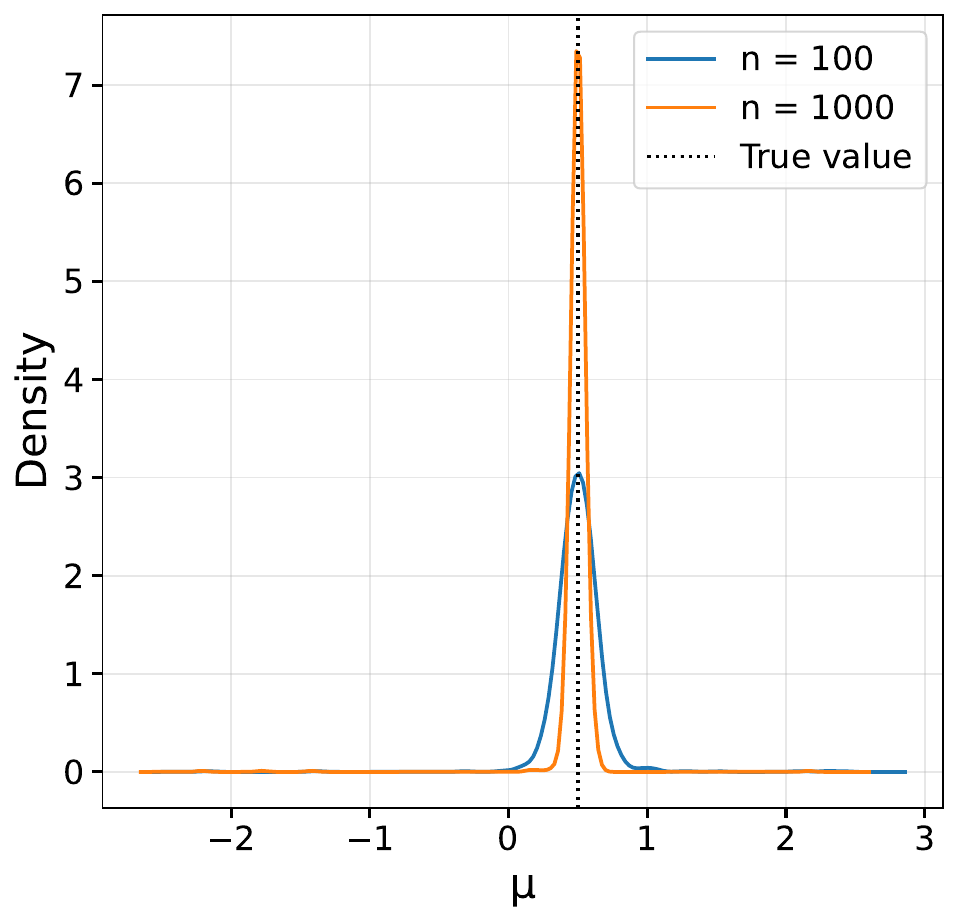}
  \caption{\textbf{Gaussian model.} Density of estimates using $3$ features.}
  \par\smallskip{\raggedright
\noindent\textbf{Alt text:} Density of estimates at sample sizes $n=100$ and $n=1{,}000$, showing the concentration of the estimates around the true parameter values as $n$ grows.\par}\label{fig:iid_gaussian_estimate_densities_param0} 
\end{figure}

\textbf{Autoregressive process.} For $t=1,\ldots,n$, we observe \begin{align*}X_t &= X_{t-1} + \psi (\mu - X_{t-1}) + \varepsilon_t,\end{align*} where the noise inputs $\varepsilon_{t}$ are sampled iid from the intractable g-and-k distribution~\citep{gk_MLE} and the initial value $X_0=3$ is known, so the model may be viewed as a discretized version of a L\'{e}vy-driven Ornstein--Uhlenbeck process. We set the location component to zero, the kurtosis component to $0.2$, and the $c$ component to $0.8$. The remaining components $\sigma$ and $g$ control the scale and skewness. We aim to estimate $\psi_0=0.3$, $\mu_0 = 2$, $\sigma_0=0.3$, $g_0=0.2$.

\begin{figure}[H]
%\captionsetup{width=\linewidth, margin=0pt}
  \centering \includegraphics[width=1\linewidth]{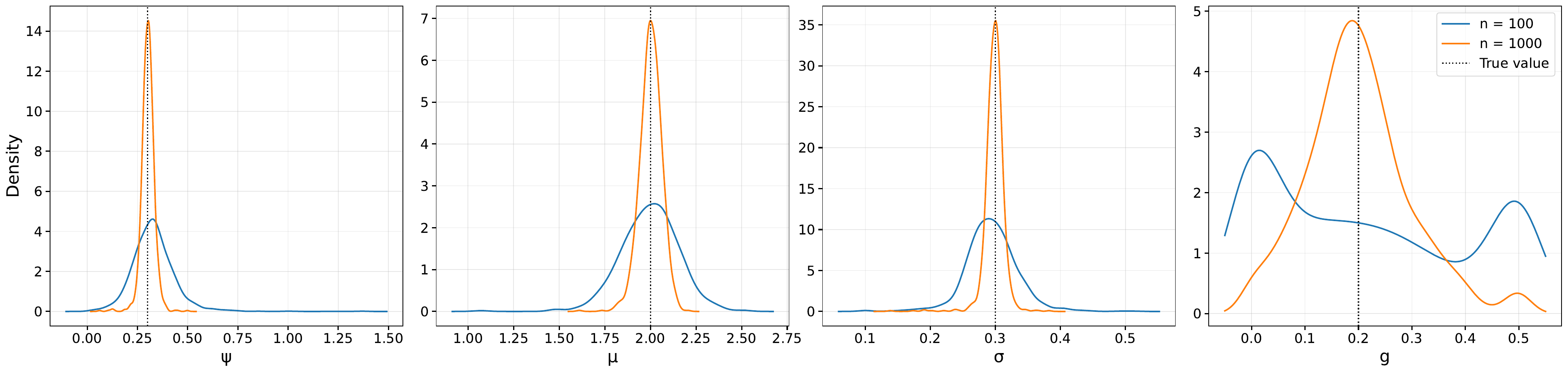}
  \caption{\textbf{Autoregressive process.} Density of estimates using $9$ features.}
\par\smallskip{\raggedright
\noindent\textbf{Alt text:} Density of estimates at sample sizes $n=100$ and $n=1{,}000$, showing the concentration of the estimates around the true parameter values as $n$ grows.\par}
\label{fig:autoregressive_estimate_densities_params0-3} 
\end{figure}

\section{Additional discussions}\label{section:additional_discussions}

We provide a review of additional related topics.

\paragraph*{Random features in machine learning.}

\cite{Rahimi_Recht_2007_random_features} originally proposed using random features to approximate kernel-based predictors. Following up on this work,~\cite{Rahimi_Recht_2008_uniform_approximation,Rahimi_Recht_2008_weighted_sums} showed that smooth functions can be well-approximated by a linear combination of randomly chosen basis functions. Since then, researchers have explored using random features for various statistical tasks, such as measuring dependence~\citep{random_dependence_coeff}, independence testing~\citep{kernel_indep_test}, and conditional independence testing~\citep{kernel_conditional_indep_test}.

\paragraph*{Embeddings in nonlinear dynamics.}

%%% Note: UK spelling so programme not program
\cite{geom_ts_Packard_et_al_1980} initiated the ``geometry from a time series'' research programme, which seeks to reconstruct the attractor of a deterministic dynamical system from an observed signal. They propose concatenating a finite number of lags of the observed signal into a vector, known as a time-delay embedding, and analysing the resulting space. This idea inspired many subsequent developments in nonlinear dynamics~\citep{nonlinear_time_series_analysis_kantz_schreiber_2003}.

\cite{embedding_takens_1981} shows that the time-delay embedding space is topologically equivalent to the underlying $p$-dimensional state space, once $2p+1$ lags are used for the time-delay embedding. In fact, the minimal embedding dimension is often less, and depends on the box-counting dimension of the attractor~\citep{embedology_Sauer_et_al_1991}.~\cite{sontag_2003} shows that a generic set of $2p+1$ experiments suffices to distinguish two different $p$-dimensional parameters in analytically parametrized dynamical systems.~\cite{aeyels_1981_2,aeyels_1981} shows that only $2p+1$ exact measurements at random times are needed for the observability of $p$-dimensional dynamical systems.

The foundation of these $2p+1$ results is the classical Whitney embedding theorem~\citep{whitney_embedding_1936}. Notably,~\cite{embedology_Sauer_et_al_1991} developed probabilistic versions of this result.

\paragraph*{Limit theorems for stochastic processes.}

Although our theory is based on (extensions of) the physical dependence measure of~\cite{wu_funct_dep_meas}, numerous central limit theorems (CLTs) and laws of large numbers (LLNs) have been established under different regularity conditions. We emphasize that the physical dependence measure conditions we use imply and are implied by $\beta$-mixing and strong mixing conditions under certain regularity conditions~\citep{Hill2025_physical_dependence_mixingale,Heinrichs2026_physical_dependence_mixing}. In what follows, we review existing CLTs and LLNs for mixing processes, mixingales, and Markov chains.

Many CLTs and LLNs have been proven under mixing conditions, which control the decay of temporal dependence so that observations far apart in time become asymptotically independent. For a comprehensive overview of CLTs for strongly mixing processes, see~\cite{CLT_mixing_Bradley2005} and Chapter 4 of~\cite{CLT_mixing_book_Rio}. A seminal result by~\cite{CLT_mixing_Rosenblatt1956} showed that stationary processes satisfying a strong mixing condition admit a CLT. Subsequently,~\cite{CLT_mixing_Ibragimov1975} demonstrated that stationary $\rho$-mixing processes satisfy a CLT. Regarding LLNs for mixing processes,~\cite{LLN_mixing_BlumHansonKoopmans1963} and~\cite{LLN_mixing_Berbee1987} proved strong LLNs for mixing processes (cf. Chapter 3 of~\cite{CLT_mixing_book_Rio}).

We now discuss mixingales, which are a generalization of mixing processes and martingale difference sequences.~\cite{CLT_mixingale_McLeish1974, LLN_mixingale_McLeish1975} introduced a CLT and an LLN for mixingales. Chapter 2 of~\cite{HallHeyde1980} discusses how martingale difference sequences, strongly mixing sequences, and certain linear processes are all examples of mixingales. We also mention the foundational work of~\cite{CLT_MartingaleApprox_Gordin1969}, who developed a CLT for stationary processes using a martingale approximation technique.

Next, we discuss Markov chains.
\cite{CLT_MarkovChain_Jones2004} provides sufficient conditions for a CLT to hold for Markov chains, and also discusses relationships to results for mixing processes.~\cite{CLT_MarkovChain_KipnisVaradhan1986} prove a CLT for reversible Markov chains with a finite limiting variance.~\cite{CLT_IteratedRandomFunctions_WuShao2004} provide a CLT for iterated random functions; see~\cite{IteratedRandomFunctions_DiaconisFreedman1999}. Regarding LLNs,~\cite{LLN_MarkovChain_Breiman} shows that a strong LLN holds for Markov chains. For more discussion of CLTs and LLNs for geometrically ergodic Markov chains, see Chapter 17 of~\cite{CLT_LLN_MarkovChains_MeynTweedie2009}.

\putbib
\end{bibunit}
\end{appendices}

\end{document}